%% file: main.tex
\documentclass[aps,reprint,11pt,tightenlines,secnumarabic,nobibnotes,nofootinbib,longbibliography,notitlepage,floatfix,onecolumn,superscriptaddress]{revtex4-2}

\usepackage{amsmath,amssymb,bm,graphicx,xcolor,tikz}
\usepackage{mathtools,amsthm}
\usepackage{microtype}
\usetikzlibrary{arrows.meta,decorations.markings}

\usepackage[T1,OT1]{fontenc}

\usepackage{type1cm}

\definecolor{sysblue}{HTML}{245C80}
\definecolor{bathteal}{HTML}{287D80}
\definecolor{warmorange}{HTML}{AD6512}

\definecolor{fig167D89}{HTML}{167D89}
\definecolor{fig172B3A}{HTML}{172B3A}
\definecolor{fig285A83}{HTML}{285A83}
\definecolor{fig566975}{HTML}{566975}
\definecolor{figAA681D}{HTML}{AA681D}
\definecolor{figC6D1D7}{HTML}{C6D1D7}
\definecolor{figDDE5E9}{HTML}{DDE5E9}
\definecolor{figEDF7F6}{HTML}{EDF7F6}
\definecolor{figEEF3F8}{HTML}{EEF3F8}
\definecolor{figFCF4E8}{HTML}{FCF4E8}
\definecolor{figFFFFFF}{HTML}{FFFFFF}
\newtheorem{Theorem}{Theorem}
\newtheorem{thm}{Theorem}[section]
\newtheorem{lem}[thm]{Lemma}
\newtheorem{sublem}[thm]{Sublemma}
\newtheorem{prop}[thm]{Proposition}
\DeclareMathVersion{coolingfigure}
\SetSymbolFont{operators}{coolingfigure}{OT1}{cmr}{m}{n}
\SetSymbolFont{letters}{coolingfigure}{OML}{cmm}{m}{it}
\SetSymbolFont{symbols}{coolingfigure}{OMS}{cmsy}{m}{n}
\SetSymbolFont{largesymbols}{coolingfigure}{OMX}{cmex}{m}{n}
\SetMathAlphabet{\mathbf}{coolingfigure}{OT1}{cmr}{bx}{n}
\SetMathAlphabet{\mathit}{coolingfigure}{OT1}{cmr}{m}{it}
\SetMathAlphabet{\mathsf}{coolingfigure}{OT1}{cmss}{m}{n}
\SetMathAlphabet{\mathtt}{coolingfigure}{OT1}{cmtt}{m}{n}

\usepackage{comment}
\usepackage{braket}
\usepackage[
  colorlinks=true,
  linkcolor=black,
  citecolor=black,
  urlcolor=black
]{hyperref}

\newcommand{\Tr}{\operatorname{Tr}}

\newcommand{\prlsection}[1]{\section{#1}}

\hypersetup{
  pdftitle={Bath-assisted cooling without resets},
  pdfauthor={Xie-Hang Yu, Zherui Chen, Lin Lin}
}
\begin{document}
\title{Bath-assisted cooling without resets}

\author{Xie-Hang Yu}
\thanks{These authors contributed equally.}
\affiliation{Department of Computing and Mathematical Sciences, California Institute of Technology, Pasadena, CA 91125, USA}

\author{Zherui Chen}
\thanks{These authors contributed equally.}
\affiliation{Department of Mathematics, University of California, Berkeley, CA 94720, USA}

\author{Lin Lin}
\email{lin@caltech.edu}
\affiliation{Department of Computing and Mathematical Sciences, California Institute of Technology, Pasadena, CA 91125, USA}
\affiliation{Department of Mathematics, University of California, Berkeley, CA 94720, USA}
\affiliation{Applied Mathematics and Computational Research Division, Lawrence Berkeley National Laboratory, Berkeley, CA 94720, USA}
\begin{abstract}
Cooling an arbitrary mixed state to a pure ground state requires transferring its entropy to an environment. Can a single coherent bath contact accomplish this without repeated bath resets? We develop a new cooling mechanism that transfers the input information into bath degrees of freedom that subsequently decouple, while the system and the remaining bath follow a ground-state path. The bath is discarded only at the end, and the system Hamiltonian remains on and unmodified throughout. We rigorously realize this mechanism for a class of weakly interacting, gapped fermionic systems with local interactions. Starting from any mixed state, our protocol prepares the interacting ground state in total physical time polylogarithmic in the system size and inverse global trace-norm error. It uses two initially empty bath modes per system mode and a single pulse shared across all onsite system--bath couplings. The same pulse prepares each admissible system's own ground state without knowledge of its microscopic parameters. The guarantees follow directly from a non-Markovian system--bath dynamics without an effective Lindbladian description.
\end{abstract}
\maketitle

\prlsection{Introduction}

Cooling quantum systems toward their ground states is a central task in condensed-matter physics~\cite{anderson1995observation,Diehl2008Quantum,Bloch2012Simulation,Wecker2015Correlated,Xu}, quantum chemistry~\cite{aspuruguzik2005simulated,motta2020determining}, and quantum information science~\cite{Kraus2008Preparation,verstraete2009quantum,Barreiro2011Simulator}. Unitary evolution of an isolated system preserves its entropy, so preparing a pure ground state from an arbitrary mixed state requires transferring entropy to auxiliary degrees of freedom, which we call a bath. In a quantum circuit, ancillary qubits can provide such an engineered bath, allowing cooling through joint system--bath evolution followed by discarding the bath~\cite{terhal2000problem}.

Starting from the seminal works of Redfield and Davies~\cite{redfield1957theory,Davies1974Markovian}, much of the theory of system--bath cooling assumes weak system--bath coupling. Together with a bath memory time short compared with the system's relaxation time, this assumption leads to an effective Markovian description~\cite{spohn1978irreversible,lidar2001completely}. A widely used model is the repeated-interaction setting, in which the bath is reset after each interaction round, thereby removing system--bath correlations and restoring its capacity to absorb entropy~\cite{scarani2002thermalizing,strasberg2017quantum,ciccarello2022quantum}. Recent algorithmic advances in preparing thermal and ground states of general interacting Hamiltonians with rigorous performance guarantees have renewed interest in system--bath cooling~\cite{ding2024single,ding2025efficient,Chen2025EfficientThermal,rouze2026efficient,zhan2026rapid,ding2026simple,hahn2025provably}. Such protocols combine local system--bath interactions with bath resets~\cite{hahn2025provably,molpeceres2025quantum,ding2026simple,wang2025beyond,slezak2026polynomialtime,chen2026overcoming} and may benefit from the robustness of dissipative preparation~\cite{verstraete2009quantum,molpeceres2025quantum,molpeceres2026benchmark}.

Weak coupling, however, slows cooling because dissipative rates arise only at second order in the system--bath interaction~\cite{lidar2001completely,ciccarello2022quantum}. Existing implementations accumulate cooling over many short bath contacts, each consisting of uninterrupted joint evolution with an initialized bath. The repeated contacts and bath resets can incur physical space--time costs that grow polynomially with the system size~\cite{ding2026simple,molpeceres2025quantum,wang2025beyond}. This motivates us to ask:
\begin{center}
\emph{Can a single coherent bath contact cool any initial state to the ground state?}
\end{center}

In this work, we develop a new \emph{non-Markovian} cooling mechanism that starts with a bath in a pure product state. We first transfer the initial system information into bath degrees of freedom that subsequently decouple, then guide the system and the remaining bath along a ground-state path. By the end of this single coherent contact, the system is arbitrarily close to its ground state and the initial entropy is stored in the bath, which we discard. We analyze the joint evolution directly, which avoids the weak-coupling approximation or an effective Markovian description.

This mechanism of coherent entropy transfer followed by ground-state transport can be applicable to many quantum systems. For concreteness, we consider a class of weakly interacting, gapped fermionic systems and prove ground-state preparation from any initial state in a single coherent contact. The total duration of joint system--bath evolution (referred to as the physical time), scales polylogarithmically with the system size and inverse global trace-norm error. The peak local control amplitude and the maximum phase rate are also polylogarithmically bounded.  A common pulse controls only the system--bath coupling and works throughout the class, without knowledge of the initial state or the microscopic Hamiltonian coefficients. Table~\ref{tab:physical-time-comparison} compares our physical-time bound with the polynomial bounds of previous system--bath protocols~\cite{ding2026simple,molpeceres2025quantum,wang2025beyond,slezak2026polynomialtime,chen2026overcoming}.

\prlsection{Related work}
Ground-state preparation for general local Hamiltonians is QMA-hard in the worst case~\cite{kempe2006complexity}, so efficient quantum cooling algorithms are expected to require additional structure or prior information. Filtering methods, such as quantum phase estimation, exploit an initial state with sufficiently large ground-state overlap~\cite{abrams1999quantum,aspuruguzik2005simulated,poulin2009preparing,Ge2019Faster,lin2020nearoptimal,dong2022ground}, and finding and preparing such states using classical approximations can be costly~\cite{lee2023evaluating,berry2025rapid}. Adiabatic state preparation avoids requiring initial overlap with the target by transporting an easily prepared ground state along a suitably gapped Hamiltonian path~\cite{farhi2001quantum,aspuruguzik2005simulated,JansenRuskaiSeiler2007,Wecker2015Correlated,albash2018adiabatic}. Constructing such a path and establishing a sufficiently large gap along it can itself be nontrivial. Moreover, unitary evolution on the system alone preserves entropy and therefore cannot prepare a pure ground state from arbitrary mixed inputs. Our work instead allows arbitrary input, and guarantees the preparation accuracy by a rigorous analysis using local adabic and perturbative dressing methods~\cite{bachmann2018adiabatic,teufel2020nonequilibrium}.

Dissipative cooling can prepare ground states from arbitrary initial states, including mixed states~\cite{Kraus2008Preparation,verstraete2009quantum}. Recent  Lindbladian approaches to thermal and ground-state preparation have established fast convergence through spectral-gap estimates~\cite{tong2025fast,vsmid2025rapid} and, for suitable model classes, the stronger property of rapid mixing~\cite{zhan2026rapid}. Even when rapid mixing holds, however, a digital implementation must realize the corresponding jump operators, which can be quasi-local and may require block encodings or linear combinations of unitaries.  High-precision algorithms for Lindblad simulation~\cite{CleveWang2017Lindblad,li2023simulating,ding2024simulating, yu2025optimizingdigitalquantumsimulation, trivedi2026nearoptimalquantumsimulationlattice}, improve the dependence on simulation accuracy, but constructing and simulating these operators can still incur substantial circuit and ancilla overheads, particularly on early fault-tolerant devices.

Direct system--bath protocols can avoid explicit jump-operator synthesis by using only Hamiltonian simulation and bath resets~\cite{ding2026simple}. Their simpler circuits are better suited to devices with limited resources~\cite{farrell2026preparing}. The analyses in Refs.~\cite{ding2026simple,slezak2026polynomialtime,chen2026overcoming} nevertheless establish convergence through mixing estimates for an effective Lindbladian. Our protocol shares these Hamiltonian simulation primitives, but we prove convergence directly through coherent entropy transfer and ground-state transport, establishing a non-Markovian cooling mechanism without an effective Lindbladian or intermediate bath resets.

For gapped free fermions, microscopic system--bath protocols with local-mode coupling have rigorous ground-state preparation bounds polynomial in the system size~\cite{ding2026simple,wang2025beyond}. For weakly interacting fermions, the cited results establish finite-temperature Gibbs preparation~\cite{slezak2026polynomialtime,wang2025beyond,chen2026overcoming}, for which the system-size dependence is typically even more severe. To our knowledge, extending such system--bath analyses to ground-state preparation at fixed interaction strength remains open. Our protocol prepares the ground state of a class of weakly interacting fermionic systems from arbitrary initial states in a single coherent bath contact, with physical time polylogarithmic in the system size and inverse global trace-norm error (Table~\ref{tab:physical-time-comparison}).

\begin{table*}[tbp]
\caption{Comparison of system--bath state-preparation protocols.
The system size $N$ counts fermionic modes, the target error $\epsilon$
is the full trace-norm error, and the total physical time $\tau$ sums
all coherent Hamiltonian evolution, including rewinding when required.
Local energy scales, Hamiltonian gaps, and locality parameters are fixed,
and reset latency is excluded.
The notation $\operatorname{polylog}$ suppresses fixed logarithmic powers.
The thermal bounds assume finite inverse temperature $\beta\geq1$,
parity-preserving inputs, and interactions below a temperature-dependent
threshold that preserves a fixed fraction of the free dynamical gap.
The entry of Ref. \cite{molpeceres2025quantum} is a perturbative estimate at fixed global fidelity,
and Ref. \cite{chen2026overcoming} gives the expected total time $\mathbb E\tau$.}
\label{tab:physical-time-comparison}
\centering
\small
\setlength{\tabcolsep}{0pt}
\renewcommand{\arraystretch}{1.18}
\begin{tabular*}{\textwidth}{@{\extracolsep{\fill}}lll@{}}
\hline\hline
\noalign{\vskip 3pt}
\textbf{Source} & \textbf{Total physical time} & \textbf{Comments} \\
\noalign{\vskip 3pt}
\hline
\noalign{\vskip 4pt}
\cite[Thm.~S18]{ding2026simple}
& $O\!\left(N^2\epsilon^{-1}\operatorname{polylog}(N/\epsilon)\right)$
& Free fermions, ground state \\[4pt]
\cite[App.~A.1.d]{molpeceres2026benchmark}
&$O(N^2\epsilon^{-2}\log(N/\epsilon))$
&Free fermions, ground state\\[4pt]
\cite[Thm.~S15]{wang2025beyond}
& $O\!\left(N^4\operatorname{polylog}(N/\epsilon)\right)$
& Free fermions, ground state \\[4pt]
\cite[Sec.~V.3]{molpeceres2025quantum}
& $O(N^4)$
& \begin{tabular}[c]{@{}l@{}}Free fermions, ground state\\[-2pt]
Perturbative estimate\end{tabular} \\
\noalign{\vskip 5pt}
\hline
\noalign{\vskip 5pt}
\cite[Prop.~11]{slezak2026polynomialtime}
& $N^{10}\epsilon^{-4}e^{O(\beta^2)}\operatorname{polylog}(N/\epsilon)$
& \begin{tabular}[c]{@{}l@{}}Weakly interacting fermions\\[-2pt]
Finite temperature\end{tabular} \\[4pt]
\cite[Cor.~S11]{wang2025beyond}
& $N^7\epsilon^{-2}e^{O(\beta^2)}\operatorname{polylog}(N/\epsilon)$
& \begin{tabular}[c]{@{}l@{}}Weakly interacting fermions\\[-2pt]
Finite temperature\end{tabular} \\[4pt]
\cite[Cor.~3.5]{chen2026overcoming}
& $\mathbb E\tau\leq N^6\epsilon^{-1}e^{O(\beta^2)}\operatorname{polylog}(N/\epsilon)$
& \begin{tabular}[c]{@{}l@{}}Weakly interacting fermions\\[-2pt]
Finite temperature\end{tabular} \\
\noalign{\vskip 5pt}
\hline
\noalign{\vskip 5pt}
\textbf{This work, Thm.~\ref{Prop:main_text}}
& {\boldmath$\operatorname{polylog}(N/\epsilon)$}
& \begin{tabular}[c]{@{}l@{}}\textbf{Weakly interacting fermions}\\[-2pt]
\textbf{Ground state}\end{tabular} \\
\noalign{\vskip 4pt}
\hline\hline
\end{tabular*}
\end{table*}

\prlsection{Main result}
The setup of our weakly interacting fermionic systems is as follows. We consider $N$ fermionic modes on a lattice of fixed spatial dimension, with annihilation operators $c_j$. Let $h=h^\dagger$ be a finite-range hopping matrix, and let the Hermitian operators $V_X$ describe number-conserving interactions supported on finite sets $X$, with a real interaction strength $u$. The system Hamiltonian $H_u$ is
\begin{equation}
 H_u=\sum_{j,k}h_{jk}c_j^\dagger c_k+u\sum_XV_X.
 \label{eq:system}
\end{equation}
 We require that the hopping and interaction are both geometrically local. In particular, we have
 \[
\sup_j\sum_k |h_{jk}|\le \Lambda,
\qquad
J_{\mathrm{int}}
:=|u|\sup_j\sum_{X\ni j}|V_X|,
\]
where $\Lambda>0$ sets the local energy scale and both $\Lambda$ and $J_{\mathrm{int}}$ are independent of the system size. We further assume that the free-fermionic part has a finite gap at the single-particle level. By denoting $\sigma(h)$ as the spectrum of the hopping matrix $h$, we require
$
\sigma(h)\subset[-\Lambda,-\Delta]\cup[\Delta,\Lambda],$ with $
0<\Delta\le\Lambda.$

\begin{figure*}[t]
\centering
\resizebox{\linewidth}{!}{\input{figures/figure-bath-strategies-input.tex}}
\caption{Two organizations of cooling. Left: previous repeated cooling rounds with freshly reset bath modes can realize a dissipative evolution. Right: we design one continuous joint evolution with a finite coherent bath. The system approaches its interacting ground state for every input $\rho_S$, while the bath retains the input entropy and information. The distinction lies in the joint dynamics during the system--bath evolution.}
\label{fig:mechanism}
\end{figure*}

The external control couples the system to two bath modes at each site through a time-dependent system--bath coupling Hamiltonian $H_{SE}(t)$, as illustrated in the right panel of Fig.~\ref{fig:mechanism}. Denoting the corresponding bath annihilation operators by $e_{j,+}$ and $e_{j,-}$, we write
\begin{equation}
 H_{SE}(t)=A(t)\sum_j\left[
 e^{i\phi(t)}\bigl(c_j^\dagger e_{j,+}^\dagger+c_je_{j,-}^\dagger\bigr)
 +\mathrm{h.c.}\right].
 \label{eq:physical}
\end{equation}
The $+$ mode couples to the system through pair-creation and pair-annihilation processes, whereas the $-$ mode couples through particle exchange. Physically, these two modes provide local particle-source and particle-sink channels, respectively. Both processes are driven uniformly across all sites by the same nonnegative amplitude $A(t)$ and real phase $\phi(t)$.

The bath has no intrinsic Hamiltonian and is initially prepared in the physical vacuum $|0\rangle_B$. It is kept coherent throughout the entire cooling process without intermediate reset or refresh, and the bath is discarded only at the end. The full evolution is therefore generated by
\[
H_{\mathrm{total}}(t)=H_u+H_{SE}(t),
\]
with the system Hamiltonian $H_u$ kept fixed throughout the protocol. In contrast to standard adiabatic preparation, the protocol does not require time-dependent control or fine-tuning of $H_u$ or of its microscopic parameters.

Let $\mathcal{T}$ denote time ordering. For a pulse of duration $\tau$, the joint system--bath propagator is
\begin{equation}
 U(\tau)=\mathcal{T}\exp\!\left[
 -i\int_0^\tau \bigl(H_u+H_{SE}(t)\bigr)\,dt
 \right].
 \label{eq:propagator}
\end{equation}
Starting from an arbitrary system density matrix $\rho_S$ and the bath vacuum $|0\rangle_B$, we discard the bath only at the end of the evolution. The resulting output system state is
\begin{equation}
 \rho_{\mathrm{out}}(\rho_S)=
 \Tr_B\!\left[
 U(\tau)(\rho_S\otimes|0\rangle_B\langle0|)
 U(\tau)^\dagger
 \right].
 \label{eq:output}
\end{equation}
Our main result is the following.

\begin{Theorem}\label{Prop:main_text}
There exists a constant $c_{\rm int}>0$ such that, whenever
\[
J_{\mathrm{int}}\le c_{\rm int}\Lambda,
\]
the interacting Hamiltonian $H_u$ has a unique ground state $|G\rangle$. For every $0<\epsilon\le1$, the cooling protocol described below and illustrated in Fig.~\ref{fig:cooling-protocol} can be chosen such that
\begin{align}
\sup_{\rho_S}
\left\lVert
\rho_{\rm out}(\rho_S)-|G\rangle\langle G|
\right\rVert_1
&\le\epsilon,
\qquad
\tau\le\frac{C}{\Lambda}
\log^\kappa\!\left(\frac{2N}{\epsilon}\right),
\nonumber\\
\max_t\{A(t),|\dot\phi(t)|\}
&\le
C\Lambda\log^2\!\left(\frac{2N}{\epsilon}\right).
\label{eq:main}
\end{align}
Here $\|\cdot\|_1$ denotes the full trace norm. The constants $c_{\rm int},C>0$ and the finite exponent $\kappa$ depend only on the local lattice structure and the single-particle gap ratio $\Delta/\Lambda$, and are independent of $N$ and $\epsilon$.
\end{Theorem}

The trace-norm bound controls the complete interacting many-body state, including correlations across the entire system. The initial state may have extensive entropy or exponentially small ground-state population, and no knowledge of that state is required. Moreover, the system Hamiltonian $H_u$ can be treated as a black box: it may arise from an unknown experimental setting, remains fixed throughout the evolution, and the control  pulse  $A(t),\phi(t)$ does not depend on its microscopic hopping or interaction coefficients. Thus, the same single coherent system--bath coupling protocol applies uniformly to every Hamiltonian within the stated class and to arbitrary initial states. It prepares the interacting many-body ground state in a physical evolution time that grows only polylogarithmically with the system size and the inverse global trace-norm error

In this protocol, the system Hamiltonian $H_u$ remains fully active throughout the evolution; in particular, the interaction strength $u$ does not need to be reduced or fine-tuned as the system--bath coupling is varied. Although the system--bright gap decreases as $A\to 0$, the corresponding low-energy excitation becomes increasingly localized in the bath. In the free model, both its excitation energy and its weight on the system vanish quadratically with $A$. Since the interaction terms act only on the system degrees of freedom, their effect on this soft mode is suppressed by the same loss of system weight. Using the local stability estimates of Ref.~\cite{DRS}, we make this compensation mechanism rigorous in the interacting setting, allowing the same fixed interacting Hamiltonian $H_u$ to remain on throughout the protocol.
  
We control the errors of all three stages by constructing
locally dressed subspaces and bounding the leakage from
them~\cite{BMNS,bachmann2018adiabatic}.
The total leakage decreases exponentially with the expansion
order $n$, whereas
the required control strength and evolution times grow only polynomially with $n$.
The logarithmic choice of $n$ in Eq.~\eqref{eq:main_text_choose_of_n} therefore gives the
resource bounds in Eq.~\eqref{eq:main}.
The leakage estimates hold uniformly over the entire initial
input subspace and yield the stated full trace-norm accuracy
for every input state.


\begin{figure*}[t]
\centering
\resizebox{\linewidth}{!}{\input{figures/figure-cooling-stages-input.tex}}
\caption{Three-stage cooling control path $(A(t),\omega(t))$ in the rotating frame. The numbered points on the control path connect to the corresponding stage diagrams. Entrance transfers input entropy to the dark sector; transport follows the system--bright ground-state branch while the dark sector remains decoupled; exit releases the system and discards the bath. The state labels describe the ideal limit.}
\label{fig:cooling-protocol}
\end{figure*}

\prlsection{Proof ideas}
We specify the pulse through the coupling amplitude $A(t)$ and
phase rate $\omega(t)=\dot\phi(t)$, with $\phi(0)=0$. The pulse has three stages: entrance, transport, and exit, as shown in Fig.~\ref{fig:cooling-protocol}.
All stages use the same smooth, monotonically increasing switch
$f_{\uparrow}:[0,1]\to[0,1]$, specified explicitly in Appendix D.
It satisfies $f_{\uparrow}(0)=0$, $f_{\uparrow}(1)=1$, and all
positive-order derivatives vanish at both endpoints.

During entrance, the coupling is turned on over a duration
$\tau_{\mathrm{ent}}=\Lambda^{-1}$. We set
$\vartheta(t)=(\pi/2)f_{\uparrow}(t/\tau_{\mathrm{ent}})$ and
\begin{equation}
 A(t)=\frac{r}{\sqrt{2}}\sin\vartheta(t),\qquad
 \omega(t)=r\cos\vartheta(t),
 \label{eq:entrance-pulse}
\end{equation}
where $r>0$ sets the control strength. Thus, the coupling rises
from zero to $A_{\rm ent}=r/\sqrt{2}$ while the phase rate
decreases from $r$ to zero.

During transport and exit, the phase remains constant,
so that $\omega=0$. Transport lowers the amplitude from
$A_{\rm ent}$ to a positive matching amplitude $A_c$, which is given in Eq. (\ref{eq:matching-amplitude}),
through successive ramps. Each complete ramp halves
the amplitude, with the last ramp ending at $A_c$.
At fixed $N$ and $\epsilon$, our chosen duration $\tau_i$
for a ramp starting at $A_i\le\Lambda$ scales as
$\tau_i\propto\Lambda^{-1}(\Lambda/A_i)^\nu$, where the
fixed exponent $\nu>1$ depends only on the spatial
dimension (Appendix D).
Exit consists of a single ramp from $A_c$ to zero,
with duration $\tau_{\mathrm{exit}}=1/A_c$.
Each ramp from $A_i$ to $A_f$ uses
$A_i+(A_f-A_i)f_{\uparrow}(s)$, where $s\in[0,1]$
is the elapsed fraction of its duration.
The endpoint conditions join all ramps smoothly into
the single pulse shown in Fig.~\ref{fig:cooling-protocol}.

To explain how this pulse transfers the input information
and prepares the system ground state, we introduce the
bright and dark combinations of the bath modes,
\begin{equation}
 b_j=\frac{e_{j,-}-e_{j,+}^\dagger}{\sqrt2},\qquad
 d_j=\frac{e_{j,-}+e_{j,+}^\dagger}{\sqrt2}.
 \label{eq:modes}
\end{equation}
The two families obey canonical anticommutation relations. In the bath-rotating frame with respect to $\exp(i\phi(t)N_B)$ with $N_B=\sum_{j,\sigma\in{\pm}}e_{j,\sigma}^\dagger e_{j,\sigma}$, the phase rate $\omega(t)$ produces bright--dark hopping:
\begin{align}
 H_{\mathrm{rot}}(t)=H_{\mathrm{br}}(A(t))
 +\omega(t)\sum_j(b_j^\dagger d_j+d_j^\dagger b_j),
 \label{eq:rotating}\,  \qquad
 H_{\mathrm{br}}(A)=H_u-\sqrt2A\sum_j(c_j^\dagger b_j+b_j^\dagger c_j).
\end{align}
Here $H_{\mathrm{rot}}(t)$ denotes the full system--bath Hamiltonian in the bath-rotating frame, while $H_{\mathrm{br}}(A)$ denotes the Hamiltonian of the coupled system and bright modes.
An additive scalar has been omitted. The bright modes couple directly to the system. The dark modes only couple to the bright modes when the phase changes and decouple exactly once $\omega=0$. Snapshots (1)--(3) in Fig.~\ref{fig:cooling-protocol} follow the state through these coupled and decoupled stages.

\paragraph*{Entrance.}
In the bath-rotating frame, we define $H_{\mathrm{ctrl}}=H_{\mathrm{rot}}-H_u$. Along the entrance path,
\begin{equation}
H_{\mathrm{ctrl}}(\vartheta)
=-r\sin\vartheta\sum_j(c_j^\dagger b_j+b_j^\dagger c_j)
+r\cos\vartheta\sum_j(b_j^\dagger d_j+d_j^\dagger b_j).
\label{eq:entrance-control}
\end{equation}
We first illustrate the information-transfer mechanism using $H_{\mathrm{ctrl}}(\vartheta)$ alone, and temporarily neglect $H_u$. For each system site together with its two bath modes, the corresponding single-particle energies are $0$ and $\pm r$, with the zero-energy mode given by
\begin{equation}
\cos\vartheta\,c_j+\sin\vartheta\,d_j.
\label{eq:entrance-mode}
\end{equation}
As $\vartheta$ increases from $0$ to $\pi/2$, this mode is continuously rotated from the system mode $c_j$ into the dark bath mode $d_j$. Initially, the physical bath vacuum occupies the negative-energy modes and leaves the positive-energy modes empty, while the zero modes encode the arbitrary system input. These states form a $2^N$-dimensional ground-state manifold of $H_{\mathrm{ctrl}}$, separated from the excited states by the gap $r$. In the adiabatic limit, the entire manifold is transported coherently: at the end of the entrance stage, the system--bright sector is brought to the ground state of $H_{\mathrm{ctrl}}(\pi/2)$, while all dependence on the initial system state, including its quantum coherence, is transferred to the dark modes.

In the protocol, the system Hamiltonian $H_u$ remains present throughout the entrance stage. We choose
\[
r\sim\Lambda\log^2(2N/\epsilon),
\qquad
\tau_{\mathrm{ent}}=\Lambda^{-1},
\]
so that the control scale $r$ is parametrically larger than the local energy scale $O(\Lambda)$ of $H_u$, while the adiabatic parameter is
\[
\frac{1}{r\tau_{\mathrm{ent}}}=\frac{\Lambda}{r}.
\]
A local dressing of the referenced ground-state manifold then incorporates the full Hamiltonian $H_u$, including the interactions, and allows us to control the leakage uniformly over all initial system states. At the endpoint, we show that the dressed system--bright sector is close to the interacting ground state of
\[
H_{\mathrm{br}}(r/\sqrt2)
=H_u+H_{\mathrm{ctrl}}(\pi/2).
\]
Thus, after a physical time $\tau_{\mathrm{ent}}=\Lambda^{-1}$, the joint state is close to one in which the system--bright sector occupies this common interacting ground state, while the dark sector stores the information carried by the initial system state. The system--bright sector is thereby prepared for the subsequent transport stage, with the dark sector acting as an information memory.

\paragraph*{Transport.}
During this step, the dark modes remain exactly
decoupled, while the system and bright modes follow
the interacting ground-state branch of $H_{\mathrm{br}}(A)$. Under the assumption in Theorem~\ref{Prop:main_text}, this ground state is unique for every $A>0$. For a constant $c_{\mathrm{gap}}>0$ independent of $A$ and $N$, its gap obeys
\begin{equation}
 \operatorname{gap}H_{\mathrm{br}}(A)\ge c_{\mathrm{gap}}\min\{A,A^2/\Lambda\}.
 \label{eq:gaps}
\end{equation}
To specify the terminal point $A_c$, let $D$ denote the spatial dimension
and fix constants $n_0\ge3$ and $c_0>0$ as in Appendix E.
We set
\begin{equation}\label{eq:main_text_choose_of_n}
 n=\max\left\{
 n_0,\left\lceil c_0\log\!\left(\frac{2N}{\epsilon}\right)
 \right\rceil\right\}
\end{equation}
and choose
\begin{equation}
 A_c=\frac{c_A\Lambda}
 {n^{6D}[\log(e+n)]^{8D}},
 \label{eq:matching-amplitude}
\end{equation}
with a sufficiently small fixed $c_A>0$.
These constants are independent of $N$ and $\epsilon$.

This choice makes the remaining coupling weak enough, so that the subsequent exit stage can be performed efficiently,  while keeping the transport gap bounded below by
$c_{\mathrm{gap}}A_c^2/\Lambda$. Since the system--bright ground-state branch remains gapped throughout the entire transport path, each ramp duration $\tau_i$ can be chosen according to the corresponding gap bound so that the adiabatic error remains controlled. Applying the local adiabatic estimates successively to all ramps then gives
\begin{equation}
 \tau_{\mathrm{tr}}\le\frac{C}{\Lambda}\log^\kappa\!\left(\frac{2N}{\epsilon}\right),
 \label{eq:transport-time}
\end{equation}
with the same exponent $\kappa$ as in Eq.~\eqref{eq:main}.

\paragraph*{Exit.}
The final stage turns off the remaining system--bath coupling while protecting only the system ground-state sector. As $A\to0$, the bath decouples from the system and its excitations become arbitrarily soft. Consequently, the gap protecting the joint system--bright ground state closes at the endpoint, and the ground-state transport used in the preceding stage cannot be continued all the way to $A=0$. Importantly, however, the cooling task does not require the bath to remain in its ground state: only the system must stay close to its target ground state, while excitations in the bath are harmless. The relevant gap for the exit stage is therefore the system gap, which remains open even as the full system--bath gap closes. For the matching amplitude in Eq.~\eqref{eq:matching-amplitude}, we choose
\begin{equation}
 \tau_{\mathrm{exit}}
 =\frac{1}{A_c}
 =\frac{n^{6D}[\log(e+n)]^{8D}}{c_A\Lambda}.
 \label{eq:exit-time}
\end{equation}
Thus, the release requires only polylogarithmic physical time despite the closing joint gap.

To make this idea precise, let $P_u=\ket{G}\bra{G}$ and let $I_B$ denote the identity on the full bath. At $A=0$, the target projector is
\[
P_u\otimes I_B,
\]
which fixes the system in its ground state while placing no restriction on the bath. For $0<A<A_c$, we construct a locally dressed version of this projector that incorporates the residual system--bath correlations. Unlike the transport stage, this construction relies only on the fixed system gap
\[
\operatorname{gap}H_u\ge \Delta/4,
\]
rather than on a gap of the full system--bath Hamiltonian. It therefore suppresses excitations of the system while allowing arbitrarily low-energy excitations in the bath. We show that, already at $A=A_c$, the incoming interacting system--bright ground state lies within controlled error of this dressed target sector. During the exit, the dressing is continuously removed together with the system--bath coupling. At $A=0$, the dressed projector reduces to $P_u\otimes I_B$: the system is in its ground state, while the bath is free to retain both the information carried by the initial state and any residual excitations. Discarding the two bath modes at each site then completes the cooling protocol.

\prlsection{Discussion}
To our knowledge, we have established the first system--bath cooling protocol for a class of interacting many-body systems that prepares the ground state from arbitrary initial states in a single coherent bath contact with polylogarithmic physical time and peak local control strength. The system Hamiltonian remains fixed throughout, and the controls require no knowledge of its microscopic coefficients. A finite coherent bath retains the information and entropy transferred from the system, allowing cooling without intermediate resets or an effective Markovian description. This raises a natural question: how long must the bath remain coherent for cooling without resets?

Our construction combines entropy removal with adiabatic transport in the joint system--bath Hilbert space. We first transfer the initial-state dependence into bath modes that subsequently decouple, then guide the system and the remaining bath along a common interacting ground-state path. Adiabatic estimates give finite-time error bounds in terms of the relevant spectral gaps. Whether this coherent protocol can also inherit the robustness of dissipative preparation remains an open question.

Although our analysis concerns weakly interacting fermions, the weak-interaction assumption is used primarily to establish the spectral gaps of $H_u$ and of the joint transport path, rather than as an essential ingredient of the cooling mechanism itself. We therefore expect that the three-stage construction and its underlying cooling mechanism may extend to spin and bosonic systems, as well as to strongly interacting regimes in which the required gap estimates can be established by other means.

\section*{Acknowledgments}
This material is based upon work supported by the U.S. Department of Energy, Office of Science, Accelerated Research in Quantum Computing Centers, Quantum Utility through Advanced Computational Quantum Algorithms, grant no. DE-SC0025572 (Z.C., L.L.), by the Peterson postdoctoral fellowship from the Department of Computing and Mathematical Sciences at Caltech  (X.Y.), and by the Simons Targeted Grant in Mathematics and Physical Sciences on Moire Materials Magic, Award No. 896630 (X.Y., L.L.). L.L. is a Simons Investigator in Mathematics.  Z. C. thanks the hospitality of the Department of Computing and Mathematical Sciences at Caltech, where part of this work was done.  The authors thank Zhiyan Ding, Haoen Li, Yilun Yang for helpful discussions.

\section*{AI statement}
We used OpenAI GPT-5.5, GPT-5.6 Sol, and GPT-6 Astra to explore candidate proof strategies and to assist with manuscript preparation. We checked the AI-assisted material and take full responsibility for the content of this work.

\nocite{feng2022quantum}
\bibliographystyle{apsrev4-2}
\bibliography{ref}

\include{appendix}

\end{document}

%% file: figures/figure-bath-strategies-input.tex
\begingroup
\definecolor{ink}{HTML}{172B3A}
\definecolor{slate}{HTML}{566975}
\definecolor{guide}{HTML}{C6D1D7}
\definecolor{rulegray}{HTML}{DDE5E9}
\definecolor{sysblue}{HTML}{245C80}
\definecolor{bathteal}{HTML}{287D80}
\definecolor{darkgold}{HTML}{AD6512}
\definecolor{sysfill}{HTML}{EEF3F8}
\definecolor{brightfill}{HTML}{EDF7F6}
\definecolor{darkfill}{HTML}{FCF4E8}

\providecommand{\Tr}{\operatorname{Tr}}

\begin{tikzpicture}[x=1mm,y=1mm,
  font=\fontsize{8}{9.4}\selectfont,text=ink,
  line cap=round,line join=round,
  every node/.style={inner sep=0.6pt,outer sep=0pt},
  panel title/.style={anchor=west,font=\fontsize{9.5}{11}\selectfont\bfseries},
  section title/.style={anchor=west,font=\fontsize{9}{10.5}\selectfont\bfseries},
  note/.style={font=\fontsize{7.3}{8.6}\selectfont,text=slate},
  wire/.style={draw=slate,line width=0.65pt},
  flow/.style={-{Stealth[length=1.65mm,width=1.05mm]},line width=0.65pt},
  axis/.style={draw=guide,line width=0.45pt,-{Stealth[length=1.7mm,width=1.15mm]}},
  stage badge/.style={circle,minimum size=3.7mm,inner sep=0pt,
    font=\fontsize{7.8}{8}\selectfont\bfseries,text=white},
  panel badge/.style={stage badge,fill=ink,minimum size=3.8mm},
  object/.style={rounded corners=0.8mm,line width=0.65pt,
    minimum width=16mm,minimum height=8.4mm,
    font=\fontsize{8.1}{9.4}\selectfont\bfseries},
  system/.style={object,draw=sysblue,fill=sysfill,text=sysblue},
  bright/.style={object,draw=bathteal,fill=brightfill,text=bathteal},
  dark/.style={object,draw=darkgold,fill=darkfill,text=darkgold}]
\path[use as bounding box] (0,0) rectangle (170,36);
\draw[draw=rulegray,line width=0.4pt] (83.5,0.5) -- (83.5,35.5);

\begin{scope}[yshift=-32mm]
\node[section title] at (0.5,64.5) {Previous work: repeated resets};
\node[note,anchor=east,text=ink] at (9,54) {System};
\node[note,anchor=east,text=ink] at (9,46) {Bath};
\draw[draw=gray!75,flow,line width=0.65pt] (14.2,54) -- (65,54);
\node[text=sysblue,anchor=east] at (14.2,54) {$\rho_S$};
\foreach \x in {22,39,56}{
  \draw[draw=gray!75,line width=0.65pt] (\x-6,46) -- (\x+6.5,46);
  \draw[draw=gray!75,flow,line width=0.65pt]
    (\x+6.5,46) -- (\x+6.5,38);
  \node[note,text=ink] at (\x+6.5,34.6) {$\Tr_B$};
  \node[note,text=bathteal,anchor=east] at (\x-6,46) {$|0\rangle$};
  \path[draw=sysblue,fill=sysblue!7,line width=0.75pt]
    (\x-3.3,43.1) rectangle (\x+3.3,56.9);
  \node[font=\fontsize{9}{10.5}\selectfont] at (\x,50) {$U$};
}
\node[fill=white,inner xsep=1.1mm,inner ysep=0.3mm] at (47.5,54) {$\cdots$};
\node[note,text=darkgold] at (30.5,49.3) {reset};
\node[note,text=darkgold] at (47.5,49.3) {reset};
\node[note,text=sysblue,align=left,anchor=west] at (66,54) {Ground\\state};
\end{scope}

\begin{scope}[xshift=89mm,yshift=2mm]
\node[section title] at (0.5,30.5) {This work: without repeated resets};
\node[note,anchor=east,text=ink] at (9,20) {System};
\node[note,anchor=east,text=ink] at (9,13.35) {Bath};
\draw[draw=gray!75,flow,line width=0.65pt] (14.2,20) -- (65,20);
\draw[draw=gray!75,line width=0.65pt] (17.2,15.2) -- (18.7,15.2);
\draw[draw=gray!75,line width=0.65pt] (17.2,11.5) -- (18.7,11.5);
\draw[draw=gray!75,line width=0.65pt] (55,13) -- (66.8,13);
\draw[draw=gray!75,flow,line width=0.65pt] (66.8,13) -- (66.8,4.5);
\path[draw=bathteal,fill=bathteal!7,line width=0.8pt,rounded corners=0.8mm]
  (18.7,10) rectangle (55,23);
\node[text=bathteal,align=center,font=\fontsize{7.7}{10.1}\selectfont]
  at (36.85,16.5) {Finite-time joint evolution\\[1.4pt]$H_u+H_{SE}(t)$};
\node[text=sysblue,anchor=east] at (14.2,20) {$\rho_S$};
\node[note,text=bathteal,anchor=east] at (17.8,15.2) {$|0\rangle_{B_+}$};
\node[note,text=bathteal,anchor=east] at (17.8,11.5) {$|0\rangle_{B_-}$};
\node[note,text=sysblue,align=left,anchor=west] at (66,20) {Ground\\state};
\node[text=darkgold,align=center,font=\fontsize{6.8}{8}\selectfont]
  at (60.7,8.3) {Input\\entropy};
\node[note,text=ink] at (66.8,1.2) {$\Tr_B$};
\end{scope}
\end{tikzpicture}
\endgroup

%% file: figures/figure-cooling-stages-input.tex
\begingroup
\definecolor{ink}{HTML}{172B3A}
\definecolor{slate}{HTML}{566975}
\definecolor{guide}{HTML}{C6D1D7}
\definecolor{rulegray}{HTML}{DDE5E9}
\definecolor{sysblue}{HTML}{245C80}
\definecolor{bathteal}{HTML}{287D80}
\definecolor{darkgold}{HTML}{AD6512}
\definecolor{sysfill}{HTML}{EEF3F8}
\definecolor{brightfill}{HTML}{EDF7F6}
\definecolor{darkfill}{HTML}{FCF4E8}

\providecommand{\Tr}{\operatorname{Tr}}

\begin{tikzpicture}[x=1mm,y=1mm,
  font=\fontsize{8}{9.4}\selectfont,text=ink,
  line cap=round,line join=round,
  every node/.style={inner sep=0.6pt,outer sep=0pt},
  panel title/.style={anchor=west,font=\fontsize{9.5}{11}\selectfont\bfseries},
  section title/.style={anchor=west,font=\fontsize{9}{10.5}\selectfont\bfseries},
  note/.style={font=\fontsize{7.3}{8.6}\selectfont,text=slate},
  wire/.style={draw=slate,line width=0.65pt},
  flow/.style={-{Stealth[length=1.65mm,width=1.05mm]},line width=0.65pt},
  axis/.style={draw=guide,line width=0.45pt,-{Stealth[length=1.7mm,width=1.15mm]}},
  stage badge/.style={circle,minimum size=3.7mm,inner sep=0pt,
    font=\fontsize{7.8}{8}\selectfont\bfseries,text=white},
  panel badge/.style={stage badge,fill=ink,minimum size=3.8mm},
  object/.style={rounded corners=0.8mm,line width=0.65pt,
    minimum width=16mm,minimum height=8.4mm,
    font=\fontsize{8.1}{9.4}\selectfont\bfseries},
  system/.style={object,draw=sysblue,fill=sysfill,text=sysblue},
  bright/.style={object,draw=bathteal,fill=brightfill,text=bathteal},
  dark/.style={object,draw=darkgold,fill=darkfill,text=darkgold}]
\path[use as bounding box] (0,0) rectangle (110,56);
\begin{scope}[
  object/.append style={minimum width=11mm,minimum height=6.2mm,
    rounded corners=0.65mm,font=\fontsize{7}{8.2}\selectfont\bfseries},
  object label/.style={font=\fontsize{6.3}{7.4}\selectfont\bfseries},
  note/.append style={font=\fontsize{6.4}{7.5}\selectfont},
  section title/.append style={font=\fontsize{8.2}{9.5}\selectfont\bfseries},
  stage badge/.append style={minimum size=3.5mm,
    font=\fontsize{7.2}{8}\selectfont\bfseries},
  stage frame/.style={draw=rulegray,line width=0.5pt,rounded corners=1.2mm},
  control/.style={font=\fontsize{7.1}{8.3}\selectfont},
  leader/.style={draw=guide,line width=0.5pt,rounded corners=0.7mm}]
\draw[leader] (24,44.81) -- (58,44.81);
\draw[leader] (28,32.5) -- (28,26.2) -- (76,26.2) -- (76,24);
\draw[leader] (13,32.5) -- (13,24);

\begin{scope}[xshift=-82mm,yshift=1mm]
\draw[axis] (88.5,31.5) -- (122.5,31.5);
\draw[axis] (90,30) -- (90,50);
\node[control,anchor=west] at (123.3,31.5) {$A$};
\node at (90,52) {$\omega$};
\node[control,anchor=east] at (88.5,46.5) {$r_{\rm ent}$};
\node[note,anchor=north east] at (88.8,30.2) {$0$};
\draw[draw=sysblue,line width=1.2pt,
  postaction={decorate},decoration={markings,
    mark=at position 0.24 with {\arrow{Stealth[length=1.5mm,width=1mm]}},
    mark=at position 0.84 with {\arrow{Stealth[length=1.5mm,width=1mm]}}}]
  plot[domain=0:90,samples=80,variable=\t]
    ({90+28*sin(\t)},{31.5+15*cos(\t)});
\draw[draw=bathteal,line width=1.2pt,
  postaction={decorate},decoration={markings,
    mark=at position 0.79 with {\arrow{Stealth[length=1.5mm,width=1mm]}}}]
  (118,31.5) -- (99,31.5);
\draw[draw=darkgold,line width=1.2pt,
  postaction={decorate},decoration={markings,
    mark=at position 0.8 with {\arrow{Stealth[length=1.4mm,width=0.9mm]}}}]
  (99,31.5) -- (90,31.5);
\fill[sysblue] (90,46.5) circle (0.65mm);
\fill[sysblue] (118,31.5) circle (0.7mm);
\fill[bathteal] (99,31.5) circle (0.65mm);
\fill[darkgold] (90,31.5) circle (0.65mm);
\node[stage badge,fill=sysblue] at (106,43.81) {1};
\node[stage badge,fill=bathteal] at (110,31.5) {2};
\node[stage badge,fill=darkgold] at (95,31.5) {3};
\node[control,anchor=north] at (99,29.8) {$A_c$};
\node[control,anchor=north] at (118,29.8) {$A_{\rm ent}$};
\end{scope}

\begin{scope}[xshift=-60mm,yshift=-11mm]
\draw[stage frame] (118,42) rectangle (170,65);
\node[stage badge,fill=sysblue] at (122,61.3) {1};
\node[section title,text=sysblue] at (126.2,61.3) {Entrance};
\node[system] (s1) at (126,55) {\phantom{System}};
\node[bright] (b1) at (144,55) {\phantom{Bright}};
\node[dark] (d1) at (162,55) {\phantom{Dark}};
\foreach \objectname/\objectword/\objectcolor in
  {s1/System/sysblue,b1/Bright/bathteal,d1/Dark/darkgold}{
  \node[object label,text=\objectcolor] at (\objectname.center) {\objectword};
}
\draw[draw=ink,line width=0.65pt] (s1.east) -- (b1.west);
\draw[draw=ink,line width=0.65pt] (b1.east) -- (d1.west);
\node[control] at (135,56.9) {$A\uparrow$};
\node[control] at (153,56.9) {$\omega\downarrow$};
\draw[flow,draw=darkgold] (126,49.1) -- (162,49.1);
\node[text=darkgold] at (144,46.3) {Input entropy};
\end{scope}

\begin{scope}[
  object/.append style={minimum width=8.6mm,minimum height=5.2mm,
    font=\fontsize{6.2}{7.3}\selectfont\bfseries},
  object label/.append style={font=\fontsize{5.6}{6.6}\selectfont\bfseries},
  control/.append style={font=\fontsize{5.6}{6.6}\selectfont,inner sep=0.2pt}]
\draw[stage frame] (58,0) rectangle (110,24);
\node[stage badge,fill=bathteal] at (62,20.3) {2};
\node[section title,text=bathteal] at (66.2,20.3) {Transport};
\begin{scope}[xshift=-65.6mm]
\node[system] (s2) at (135.2,13.7) {\phantom{System}};
\node[bright] (b2) at (149.4,13.7) {\phantom{Bright}};
\node[dark] (d2) at (164.8,13.7) {\phantom{Dark}};
\foreach \objectname/\objectword/\objectcolor in
  {s2/System/sysblue,b2/Bright/bathteal,d2/Dark/darkgold}{
  \node[object label,text=\objectcolor] at (\objectname.center) {\objectword};
}
\draw[draw=ink,line width=0.65pt] (s2.east) -- (b2.west);
\node[control] at (142.3,15.5) {$A\downarrow$};
\node[control] at (157.1,13.7) {$\omega{=}0$};
\draw[draw=guide,line width=0.45pt]
  (130.2,9.7) -- (130.2,8.8) -- (154,8.8) -- (154,9.7);
\node[note,align=center] at (142,5.4) {Ground-state\\transport};
\node[note,align=center] at (164,5.4) {Input\\memory};
\end{scope}

\draw[stage frame] (0,0) rectangle (52,24);
\node[stage badge,fill=darkgold] at (4,20.3) {3};
\node[section title,text=darkgold] at (8.2,20.3) {Exit};
\begin{scope}[xshift=-78.75mm]
\draw[draw=guide,line width=0.5pt,rounded corners=0.85mm]
  (99.8,10.2) rectangle (125.9,17.2);
\node[system] (s3) at (89,13.7) {\phantom{System}};
\node[bright] (b3) at (105.2,13.7) {\phantom{Bright}};
\node[dark] (d3) at (120.8,13.7) {\phantom{Dark}};
\foreach \objectname/\objectword/\objectcolor in
  {s3/System/sysblue,b3/Bright/bathteal,d3/Dark/darkgold}{
  \node[object label,text=\objectcolor] at (\objectname.center) {\objectword};
}
\node[control] at (96.95,13.7) {$A{=}0$};
\node[control] at (113,13.7) {$\omega{=}0$};
\node[note,align=center] at (89,5.4) {Ground\\state};
\draw[flow,draw=slate,line width=0.55pt] (112.85,10.2) -- (112.85,7.7);
\node[note] at (112.85,5.2) {Discard bath once};
\end{scope}
\end{scope}
\end{scope}
\end{tikzpicture}
\endgroup

%% file: appendix.tex

\appendix
\numberwithin{equation}{section}
\allowdisplaybreaks[2]
\setlength{\emergencystretch}{2em}

\input{appendix/appendix_A_model.tex}

\input{appendix/appendix_B_stability.tex}

\input{appendix/appendix_C_local_dressing.tex}

\input{appendix/appendix_D_stages.tex}

\input{appendix/appendix_E_resources.tex}

%% file: appendix/appendix_A_model.tex
\section{Model, conventions, and the main proposition}

\label{app:model:setting}

We give the detailed description of the model and our main result
in the main text. All Hamiltonians below act on finite-dimensional
fermionic Fock spaces. We set $\hbar=1$ and use natural logarithms
throughout. We keep the convention $\|\cdot\|$ for the operator norm
and $\|\cdot\|_{1}$ for the full trace norm, without a factor of
$1/2$.

\subsection{Local Hamiltonians and the weak-interaction condition}

For the system modes $c_{j}$, $1\le j\le N$, lying on a lattice
of spatial dimension $D$, we assume a ball-count bound with volume-growth
constant $C_{\mathrm{vol}}$, valid for every center $x$ and ball
radius $\ell\ge0$:
\begin{equation}
|\{y:\operatorname{dist}(x,y)\le\ell\}|\le C_{\mathrm{vol}}(1+\ell)^{D}.\label{app:model:ball-growth}
\end{equation}
The system Hamiltonian is given by

\begin{equation}
H_{u}=H_{0}+u\sum_{X}V_{X},\qquad H_{0}=\sum_{j,k}h_{jk}c^{\dagger}_{j}c_{k},\qquad h=h^{\dagger},\label{app:model:Hu}
\end{equation}
where $H_{0}$ is a fermionic Gaussian Hamiltonian, $u$ is real and
each $V_{X}$ is a Hermitian, number-conserving operator supported
on system modes in $X$. We assume that the Hamiltonian is geometrically
local, namely, the region of $|X|$ and the hopping range of $H_{0}$
are both bounded by $R$. We further assume that the free Hamiltonian
$H_{0}$ is itself gapped and satisfies
\begin{equation}
\sup_{j}\sum_{k}|h_{jk}|\le\Lambda,\qquad\sigma(h)\subset[-\Lambda,-\Delta]\cup[\Delta,\Lambda],\qquad0<\Delta\le\Lambda\label{app:model:freegap}
\end{equation}
for some constant $\Lambda,\Delta$.

For any global operator $\Phi$ which has decomposition $\Phi=\sum_{X}\Phi_{X}$
into local terms, we define 
\begin{equation}
\|\Phi\|_{\mathrm{loc}}=\sup_{j}\sum_{X\ni j}\|\Phi_{X}\|.\label{app:model:localnorm}
\end{equation}
This will be referred as the local term and is dependent on the specific
decomposition we choose. In particular, $\|\Phi\|\le N\|\Phi\|_{\mathrm{loc}}$since
any term in $\Phi$ must contains at least one site $j$. For the
physical interaction, we introduce $J_{\mathrm{int}}=|u|\sup_{j}\sum_{X\ni j}\|V_{X}\|$.
Our sufficient weak-interaction assumption is 
\begin{equation}
J_{\mathrm{int}}\le c_{\mathrm{loc}}\frac{\Delta^{2}}{\Lambda}\left(1+\frac{R\Lambda}{\Delta}\right)^{-D},\label{app:model:weak}
\end{equation}
where $c_{\mathrm{loc}}>0$ is a constant. With $\Delta/\Lambda$
and the range $R$ fixed, this is the condition $J_{\mathrm{int}}\le c_{\mathrm{int}}\Lambda$
used in the main text, where $c_{\text{int}}$ depends on the ratio
of $\Delta/\Lambda$, $R$ and the dimension $D$ but is independent
of $N$. Theorem~\ref{app:stab:theorem} establishes the interacting
system gaps under this condition.

Unless stated otherwise, general constants denoted by $C$ and $c$
are positive, may change from line to line, and depend only on the
local lattice structure as well as $\Delta/\Lambda$. They do not
depend on $N$, the accuracy $\epsilon$, and the expansion order
$n$ we introduced in the proof below.

\subsection{Environmental bath control and the main proposition}

Each system mode receives two bath modes $e_{j,-},e_{j,+}$, initially
in their joint bath vacuum $|0\rangle_{B}$. The bare bath Hamiltonian
$H_{E}$ is zero. The time-dependent coupling control between the
system and bath is
\begin{equation}
H_{SE}(t)=A(t)\sum_{j}\left[e^{i\phi(t)}\bigl(c^{\dagger}_{j}e^{\dagger}_{j,+}+c_{j}e^{\dagger}_{j,-}\bigr)+\mathrm{h.c.}\right],\qquad A(t)\ge0.\label{app:model:physical}
\end{equation}
The full evolution Hamiltonian is $H_{u}+H_{SE}(t)$, with $H_{u}$
only acting on the system and does not change with time.

Let $U(t)$ is the evolution driven by the full Hamiltonian: 
\begin{equation}
i\partial_{t}U(t)=[H_{u}+H_{SE}(t)]U(t),
\end{equation}
 with $U(0)=I$. We turn on $H_{SE}(t)$ for a duration $\tau$, namely,
$H_{SE}(t)=0$ if $t\notin[0,\tau]$. After this control sequence,
the output system $\rho_{\text{out}}$
\begin{equation}
\rho_{\mathrm{out}}(\rho)=\operatorname{Tr}_{B}\!\left[U(\tau)(\rho\otimes|0\rangle_{B}\langle0|)U(\tau)^{\dagger}\right]\label{app:model:output}
\end{equation}
can be cooled down very close to the ground state of $H_{u}$ for
any input system state $\rho$, as stated by the following proposition:
\begin{prop}[Single-contact cooling]
\label{app:model:main} Under Eqs.~\eqref{app:model:freegap} and
\eqref{app:model:weak}, $H_{u}$ has a unique ground state with density
matrix $P_{u}$ and gap at least $\Delta/4$ on its full Fock space.
For every $0<\epsilon\le1$, there is a pulse of the form \eqref{app:model:physical}
such that 
\begin{align}
\sup_{\rho}\|\rho_{\mathrm{out}}(\rho)-P_{u}\|_{1} & \le\epsilon,\label{app:model:accuracy}\\
\tau & \le\frac{C}{\Lambda}\log^{\kappa}\!\left(\frac{2N}{\epsilon}\right),\label{app:model:time}\\
\max_{t}\{A(t),|\dot{\phi}(t)|\} & \le C\Lambda\log^{2}\!\left(\frac{2N}{\epsilon}\right).\label{app:model:peak}
\end{align}
The constants $C$ and the finite exponent $\kappa$ can depend on
$D$, but are independent of $N$ and $\epsilon$. The same pulse
$H_{SE}(t)$ works for every Hamiltonian satisfying the common local,
energy, and gap bounds Eqs.~\eqref{app:model:freegap} and \eqref{app:model:weak}
without specific dependence on individual hopping or interaction coefficients. 
\end{prop}

The protocol in Proposition \eqref{app:model:main} uses exactly $2N$
bath modes initialized in the vacuum and discards the bath only once,
at the end. The pulse $H_{SE}(t)$ consists of three consecutive stages:
entrance, transport, and exit. Appendix \eqref{app:local:section}
develops a common local dressing theorem for controlling the errors
of all three stages. Using these results, Appendix \eqref{app:stages}
analyzes each stage and the matching between them. Finally, appendix
\eqref{app:resources:section} combines the stage estimates to prove
the global trace-norm accuracy and resource bounds. We also establishes
spectral-gap bounds for the system and system--bright Hamiltonians
in Appendix \eqref{app:stab:section}, which is based on fermionic
stability theorem.

\subsection{Rotating frame and the brick-dark sector decomposition}

\label{app:model:frame}

The coupling Hamiltonian $H_{SE}(t)$ becomes simpler in another set
of bath variables: bright mode ($b_{j})$ and dark mode $(d_{j})$,
whose meaning will be clear later:
\begin{equation}
b_{j}=\frac{e_{j,-}-e^{\dagger}_{j,+}}{\sqrt{2}},\qquad d_{j}=\frac{e_{j,-}+e^{\dagger}_{j,+}}{\sqrt{2}},\label{app:model:modes}
\end{equation}
which also obeys the canonical fermionic anticommutation relation:
$\{b_{j},b^{\dagger}_{k}\}=\{d_{j},d^{\dagger}_{k}\}=\delta_{jk}$,
and zero if we mix $b_{j},d_{k}$ in the anticommutators. 

Define the environmental bath number $N_{E}=\sum_{j,\sigma\in\{+,-\}}e^{\dagger}_{j,\sigma}e_{j,\sigma}$
and the rotation $R_{B}(t)=e^{i\phi(t)N_{E}}$ with $\phi(0)=0$.
In the rotating frame $|\widetilde{\psi}(t)\rangle=R_{B}(t)^{\dagger}|\psi(t)\rangle$,
the Hamiltonian can be rewritten as 
\begin{align}
H_{\mathrm{rot}}(t) & =R^{\dagger}_{B}[H_{u}+H_{SE}(t)]R_{B}-iR^{\dagger}_{B}\dot{R}_{B}\nonumber \\
 & =H_{\mathrm{br}}(A(t))+\omega(t)\sum_{j}(b^{\dagger}_{j}d_{j}+d^{\dagger}_{j}b_{j}),\label{app:model:Hrot}\\
H_{\mathrm{br}}(A) & =H_{u}-\sqrt{2}A\sum_{j}(c^{\dagger}_{j}b_{j}+b^{\dagger}_{j}c_{j}),\qquad\omega(t)=\dot{\phi}(t),\label{app:model:Hbr}
\end{align}
and an irrelevant additive term proportional to the identity has been
omitted. Here, we have used that $e^{\dagger}_{j,-}e_{j,-}+e^{\dagger}_{j,+}e_{j,+}=1+b^{\dagger}_{j}d_{j}+d^{\dagger}_{j}b_{j}.$

The system only couples to the environment via the bright mode $b_{j}$,
and thus the name ``bright'' is given. At nonzero phase rate $\omega(t)\neq0$,
bright and dark modes are coupled. The mode $d_{j}$ is entirely decoupled
from the system and $b_{j}$ when $\omega=0$, thus the name ``dark''
is given. In this case, $H_{\mathrm{rot}}=H_{\mathrm{br}}(A)\otimes I_{\mathrm{dark}}$.
Write $P_{\mathrm{br}}(A)$ for the rank-one ground projection of
$H_{\mathrm{br}}(A)$, whose existence for $A>0$ is proved in Theorem~\ref{app:stab:theorem}.
We will introduce
\begin{equation}
\Pi_{\mathrm{in}}=I_{S}\otimes|0\rangle_{B}\langle0|,\qquad\Pi_{\mathrm{br}}(A)=P_{\mathrm{br}}(A)\otimes I_{\mathrm{dark}},\qquad\Pi_{\mathrm{out}}=P_{u}\otimes I_{B}.\label{app:model:interfaces}
\end{equation}
These projectors specify, respectively, a vacuum bath with an arbitrary
system state, the system--bright ground state with an arbitrary state
of the dark modes, and the system ground state with an unrestricted
bath. The pulse $H_{SE}(t)$ is designed to implement the sequence

\begin{equation}
\Pi_{\mathrm{in}}\longrightarrow\Pi_{\mathrm{br}}(A)\longrightarrow\Pi_{\mathrm{out}},
\end{equation}
where each arrow means that the evolution maps states from the subspace
selected by one projector into that selected by the next, up to a
controlled error. Thus, the physical system is cooled to its ground
state, while the bath retains the information carried by the initial
system state.

%% file: appendix/appendix_B_stability.tex
\section{Uniform interacting spectral stability}

\label{app:stab:section}

The interaction in \eqref{app:model:Hu} acts only on system modes.
This allows one fixed weak-interaction condition to control the system--bright
Hamiltonian for every positive coupling amplitude, although its bath-like
excitations become soft. 

This section establishes the gap for $H_{u}$ and $H_{\mathrm{br}}(A)$
as stated in Theorem~\ref{app:stab:theorem}. The proof techniques
rely heavily on the stability result of fermions, with bath part separated
from the tree expansion~\cite{DRS,chen2025cumulant}. If one is happy to accept Theorem~\ref{app:stab:theorem},
this section can be skipped at first, as it will not influence the
understanding of the following sections.
\begin{thm}[Uniform spectral stability]
\label{app:stab:theorem} Assume the local model, free spectral gap,
and weak-interaction bound of Eqs.~\eqref{app:model:Hu}--\eqref{app:model:weak}.
There are fixed positive constants $c_{\mathrm{loc}}$ and $c_{\mathrm{gap}}$,
determined by the local assumptions, for which $H_{u}$ and $H_{\mathrm{br}}(A)$
have unique ground states for every $A>0$, and 
\begin{equation}
\operatorname{gap}(H_{u})\ge\frac{\Delta}{4},\qquad\operatorname{gap}(H_{\mathrm{br}}(A))\ge c_{\mathrm{gap}}\min\!\left\{ A,\frac{A^{2}}{\Lambda}\right\} .\label{app:stab:gaps}
\end{equation}
The same choice of $c_{\mathrm{loc}}$ works for all $A>0$ and all
finite system sizes. No fixed particle-number or parity sector is
imposed. 
\end{thm}

Here and below the gap is the difference between the two lowest distinct
eigenvalues, after uniqueness of the ground state has been established.
The proof does not require the interaction to be smaller than the
right-hand side of the bright-gap bound.

\subsection{Free soft modes and the covariance to be controlled}

Set $u=0$ temporarily. Define the system field vector $\boldsymbol{c}$
and the bright field vector $\boldsymbol{b}$ by
$\boldsymbol{c}\coloneqq (c_{1},\ldots,c_{N})^{\mathsf{T}}, \, \boldsymbol{b}\coloneqq (b_{1},\ldots,b_{N})^{\mathsf{T}}.$
Then one can rewrite $H_\mathrm{br}(A)$ as
\begin{equation}
\begin{aligned}
H_{\mathrm{br}}(A)
&=\sum_{j,k}h_{jk}c_{j}^{\dagger}c_{k}
-\sqrt{2}A\sum_{j}\left(c_{j}^{\dagger}b_{j}+b_{j}^{\dagger}c_{j}\right)
=\begin{pmatrix}\boldsymbol{c}^{\dagger}&\boldsymbol{b}^{\dagger}\end{pmatrix}
\begin{pmatrix}h&-\sqrt{2}AI\\-\sqrt{2}AI&0\end{pmatrix}
\begin{pmatrix}\boldsymbol{c}\\\boldsymbol{b}\end{pmatrix}.
\end{aligned}
\end{equation}
On the one-particle space ordered as system
then bright, the one-particle Hamiltonian matrix $h_A$ is
\begin{equation}
h_{A}\coloneqq \begin{pmatrix}h & -\sqrt{2}AI\\
-\sqrt{2}AI & 0
\end{pmatrix}.\label{app:stab:free-matrix}
\end{equation}
For a hopping eigenvalue $\lambda$ of $h$, the corresponding
system--bright block of \eqref{app:stab:free-matrix} has one positive
and one negative eigenvalue. Define the positive fast energy
$E_{\mathrm{f}}(\lambda,A)$ and soft energy $E_{\mathrm{s}}(\lambda,A)$
as their larger and smaller absolute values, respectively:
\begin{equation}
E_{\mathrm{f}}(\lambda,A)\coloneqq   \frac{\sqrt{\lambda^{2}+8A^{2}}+|\lambda|}{2},\qquad
E_{\mathrm{s}}(\lambda,A)\coloneqq \frac{\sqrt{\lambda^{2}+8A^{2}}-|\lambda|}{2}.
\label{app:stab:energies}
\end{equation}
The signed eigenvalues are $E_{\mathrm{f}},-E_{\mathrm{s}}$ for
$\lambda>0$ and $E_{\mathrm{s}},-E_{\mathrm{f}}$ for $\lambda<0$.
As $A\to0$, the fast energy tends to $|\lambda|$ and the soft
energy tends to zero.

For a normalized eigenvector of this block with signed eigenvalue
$\varepsilon$, denote its system and bright amplitudes by
$\alpha_{\varepsilon}$ and $\beta_{\varepsilon}$. The eigenvalue
equation and normalization give
$-\sqrt{2}A\alpha_{\varepsilon}=\varepsilon\beta_{\varepsilon}, \, |\alpha_{\varepsilon}|^{2}+|\beta_{\varepsilon}|^{2}=1.$
Its system weight $Z(\varepsilon)$ is the squared system amplitude:
$Z(\varepsilon)\coloneqq |\alpha_{\varepsilon}|^{2} =\frac{\varepsilon^{2}}{\varepsilon^{2}+2A^{2}}.$
Using $E_{\mathrm{f}}E_{\mathrm{s}}=2A^{2}$ and
$E_{\mathrm{f}}+E_{\mathrm{s}}=\sqrt{\lambda^{2}+8A^{2}}$, the fast
and soft system weights $Z_{\mathrm{f}}(\lambda,A)$ and
$Z_{\mathrm{s}}(\lambda,A)$ therefore become
\begin{equation}
Z_{\mathrm{f}}(\lambda,A)\coloneqq \frac{E_{\mathrm{f}}(\lambda,A)}{\sqrt{\lambda^{2}+8A^{2}}},\qquad
Z_{\mathrm{s}}(\lambda,A)\coloneqq \frac{E_{\mathrm{s}}(\lambda,A)}{\sqrt{\lambda^{2}+8A^{2}}}.
\label{app:stab:weights}
\end{equation}
These probability weights multiply the corresponding contributions
to the system covariance. They satisfy $Z_{\mathrm{f}}+Z_{\mathrm{s}}=1$;
as $A\to0$, the fast mode becomes system-like and the soft mode
becomes bright-bath-like. This is because $Z_{\mathrm{f}} \to 1 $ and $Z_{\mathrm{s}} \to 0$ as $A\to0$.
Here ``system-like'' means that the bright-bath component vanishes.
In particular,
\begin{equation}
\frac{Z_{\mathrm{s}}}{E_{\mathrm{s}}}=\frac{Z_{\mathrm{f}}}{E_{\mathrm{f}}}=\frac{1}{\sqrt{\lambda^{2}+8A^{2}}}\le\frac{1}{\Delta}.\label{app:stab:compensation}
\end{equation}
Thus the inverse soft energy is compensated by its small system weight.
For small $A/|\lambda|$, both $E_{\mathrm{s}}$ and $Z_{\mathrm{s}}$
vanish quadratically in $A$. This modewise observation must still
be converted to an absolute spatial sum; cancellation between different
hopping eigenvectors cannot be used for that purpose.

Define the free-gap scale $g_A$ by
\begin{equation}
g_{A}\coloneqq \min\!\left\{ A,\frac{A^{2}}{\Lambda}\right\} .\label{app:stab:free-gap-scale}
\end{equation}
The spectrum of $h_{A}$ avoids $(-g_{A},g_{A})$. Indeed, $E_{\mathrm{s}}$
decreases with $|\lambda|$, and $|\lambda|\le\Lambda$; squaring
the positive quantities in \eqref{app:stab:energies} gives $E_{\mathrm{s}}\ge A^{2}/\Lambda$
for $A\le\Lambda$ and $E_{\mathrm{s}}\ge A$ for $A\ge\Lambda$.
The free many-body ground state fills all negative one-particle levels.
It is unique on the full Fock space, and an excitation in that space
costs at least $g_{A}$, whether it is a particle or a hole.

We use imaginary-time covariances to estimate the interacting expansion~\cite{DRS,chen2025cumulant}.
Let $\beta>0$ be inverse temperature and let $t_{\mathrm{im}},s_{\mathrm{im}}\in[0,\beta]$
be imaginary times. Write 
$d_{\beta}(t_{\mathrm{im}},s_{\mathrm{im}})\coloneqq \min_{k\in\mathbb{Z}}|t_{\mathrm{im}}-s_{\mathrm{im}}+k\beta|$
for distance on the imaginary-time circle. For $0<t_{\mathrm{im}}<\beta$,
the free covariance matrix $C_{A,\beta}(t_{\mathrm{im}},0)$ is~\cite[Eqs.~(19)--(20)]{DRS}
\begin{equation}
C_{A,\beta}(t_{\mathrm{im}},0)\coloneqq -e^{-t_{\mathrm{im}}h_{A}}(I+e^{-\beta h_{A}})^{-1}.\label{app:stab:covariance-definition}
\end{equation}
Time translation and fermionic antiperiodicity define the other unequal-time
entries. At equal times we use the normal-order convention $C_{A,\beta}(t_{\mathrm{im}},t_{\mathrm{im}})\coloneqq (I+e^{\beta h_{A}})^{-1}$,
the limit from negative time difference. Products of distinct operator
insertions retain their ordering in coincident-time limits; their
Grassmann symbols are not first combined into a single vertex.

Let $C_{A,\beta}(t_{\mathrm{im}},x;s_{\mathrm{im}},y)$ denote the
$(x,y)$ matrix entry. Its first argument labels an annihilation field
and its second a creation field. Mode labels may be system or bright;
let $S$ be the set of system-mode labels. Following \cite[Eq.~(21)]{DRS},
for a positive decay rate $\gamma$ we define the weighted system row norm
$K_{SS,\gamma}$ as the larger of the two expressions
\begin{equation}
\begin{split}K_{SS,\gamma}\coloneqq \max\Bigg\{ & \sup_{t_{\mathrm{im}},\,x\in S}\sum_{y\in S}\int^{\beta}_{0}e^{\gamma d_{\beta}(t_{\mathrm{im}},s_{\mathrm{im}})}|C_{A,\beta}(t_{\mathrm{im}},x;s_{\mathrm{im}},y)|\,ds_{\mathrm{im}},\\
 & \sup_{t_{\mathrm{im}},\,x\in S}\sum_{y\in S}\int^{\beta}_{0}e^{\gamma d_{\beta}(t_{\mathrm{im}},s_{\mathrm{im}})}|C_{A,\beta}(s_{\mathrm{im}},y;t_{\mathrm{im}},x)|\,ds_{\mathrm{im}}\Bigg\}.
\end{split}
\label{app:stab:row-norm}
\end{equation}
The reversed row norm is included because a tree edge can have either
orientation. Let $\mathcal{I}$ be the set of all system and bright mode
labels. We can define the weighted full row norm $K_{\mathrm{full},\gamma}$ similarly: 
\begin{equation}
\begin{split}K_{\mathrm{full},\gamma}\coloneqq \max\Bigg\{ & \sup_{t_{\mathrm{im}},\,x\in\mathcal{I}}\sum_{y\in\mathcal{I}}\int^{\beta}_{0}e^{\gamma d_{\beta}(t_{\mathrm{im}},s_{\mathrm{im}})}|C_{A,\beta}(t_{\mathrm{im}},x;s_{\mathrm{im}},y)|\,ds_{\mathrm{im}},\\
 & \sup_{t_{\mathrm{im}},\,x\in\mathcal{I}}\sum_{y\in\mathcal{I}}\int^{\beta}_{0}e^{\gamma d_{\beta}(t_{\mathrm{im}},s_{\mathrm{im}})}|C_{A,\beta}(s_{\mathrm{im}},y;t_{\mathrm{im}},x)|\,ds_{\mathrm{im}}\Bigg\}.
\end{split}
\label{app:stab:full-row-norm}
\end{equation}

\begin{lem}[Uniform system covariance]
\label{app:stab:covariance} Choose a fixed sufficiently small $c_{\mathrm{gap}}>0$
and set the decay rate $\gamma_{A}\coloneqq c_{\mathrm{gap}}g_{A}$.
There is a local constant $C_{\mathrm{cov}}$ defining the uniform
covariance bound $K_{\mathrm{cov}}$ such that
\begin{equation}
K_{SS,\gamma_{A}}\le K_{\mathrm{cov}},\qquad K_{\mathrm{cov}}\coloneqq C_{\mathrm{cov}}\frac{\Lambda}{\Delta^{2}}\left(1+\frac{R\Lambda}{\Delta}\right)^{D},\label{app:stab:Kcov}
\end{equation}
uniformly in $N$, $A>0$, and $\beta>0$. With $h_{A}$ replaced
by $h$, the system-only covariance obeys the same bound at rate $\Delta/4$.
For fixed $N$ and $A>0$, the full norm also has a finite bound independent
of $\beta$: 
\begin{equation}
K_{\mathrm{full},\gamma_{A}}\le\frac{CN}{g_{A}-\gamma_{A}}.\label{app:stab:Kfull}
\end{equation}
Only \eqref{app:stab:Kcov}, not \eqref{app:stab:Kfull}, controls
the interaction threshold. 
\end{lem}
In the following subsection, we provide a proof of Lemma~\ref{app:stab:covariance}. 

\subsection{Proof of Lemma~\ref{app:stab:covariance}}
\begin{proof}
We first reduce the two imaginary times in \eqref{app:stab:row-norm}
to the relative imaginary time $\tau_{\mathrm{im}}$. For fixed
$t_{\mathrm{im}}$, set
$\tau_{\mathrm{im}}\coloneqq(t_{\mathrm{im}}-s_{\mathrm{im}})\bmod\beta$.
Time translation and fermionic antiperiodicity give \cite[Eqs.~(19)--(20)]{DRS}, at unequal times,
\begin{equation}
\begin{aligned}
|C_{A,\beta}(t_{\mathrm{im}},x;s_{\mathrm{im}},y)|
&=|C_{A,\beta}(\tau_{\mathrm{im}},x;0,y)|,\\
d_\beta(t_{\mathrm{im}},s_{\mathrm{im}})
&=d_\beta(\tau_{\mathrm{im}},0).
\end{aligned}
\label{app:stab:relative-time-covariance}
\end{equation}
The absolute value removes the antiperiodic sign. As $s_{\mathrm{im}}$
runs once around the time circle, so does $\tau_{\mathrm{im}}$, hence
\begin{equation}
\begin{aligned}
&\int_0^\beta e^{\gamma_A d_\beta(t_{\mathrm{im}},s_{\mathrm{im}})}
 |C_{A,\beta}(t_{\mathrm{im}},x;s_{\mathrm{im}},y)|\,ds_{\mathrm{im}}=\int_0^\beta e^{\gamma_A d_\beta(\tau_{\mathrm{im}},0)}
 |C_{A,\beta}(\tau_{\mathrm{im}},x;0,y)|\,d\tau_{\mathrm{im}}.
\end{aligned}
\label{app:stab:relative-time-integral}
\end{equation}
For the reversed covariance
in \eqref{app:stab:row-norm}, use
$\tau_{\mathrm{im}}\coloneqq(s_{\mathrm{im}}-t_{\mathrm{im}})\bmod\beta$
and interchange $x,y$. Both resulting integrals are independent of the
fixed time $t_{\mathrm{im}}$, so its supremum can be omitted. We retain
$\tau_{\mathrm{im}}$ for relative imaginary time throughout the remainder
of this section.

We now express the system covariance as a function of $h$ and then
as a contour integral. Let $P_{\lambda}$ be the spectral projection
of $h$ at $\lambda$, so that
$h=\sum_{\lambda\in\sigma(h)}\lambda P_{\lambda}$. For
$0<\tau_{\mathrm{im}}<\beta$, define the scalar thermal function
$q_{\beta,\tau_{\mathrm{im}}}(\varepsilon)
\coloneqq -e^{-\tau_{\mathrm{im}}\varepsilon}/(1+e^{-\beta\varepsilon})$.
The full system--bright covariance in
\eqref{app:stab:covariance-definition} is therefore
$C_{A,\beta}(\tau_{\mathrm{im}},0)=q_{\beta,\tau_{\mathrm{im}}}(h_A)$.

Since the system--bright coupling in \eqref{app:stab:free-matrix}
is proportional to the identity, the system and bright copies of each
hopping eigenvector span an invariant two-dimensional subspace.
Denote its block matrix by $m_A(\lambda)$ and its scalar system
covariance by $F_{A,\beta,\tau_{\mathrm{im}}}(\lambda)$:
\begin{equation}
m_A(\lambda)\coloneqq 
\begin{pmatrix}\lambda&-\sqrt{2}A\\-\sqrt{2}A&0\end{pmatrix},
\qquad
F_{A,\beta,\tau_{\mathrm{im}}}(\lambda)
\coloneqq \bigl[q_{\beta,\tau_{\mathrm{im}}}(m_A(\lambda))\bigr]_{11}.
\label{app:stab:scalar-block}
\end{equation}
Here the $(1,1)$ entry selects the system component of the
$2\times2$ covariance matrix $q_{\beta,\tau_{\mathrm{im}}}(m_A(\lambda))$.
For the full $2N\times2N$ covariance, let $[\cdot]_{SS}$ denote
its upper-left $N\times N$ system block. Taking this block in
\eqref{app:stab:covariance-definition} and summing over the hopping
eigenspaces gives
\begin{equation}
\begin{aligned}
C^{SS}_{A,\beta}(\tau_{\mathrm{im}},0)
&=\bigl[q_{\beta,\tau_{\mathrm{im}}}(h_A)\bigr]_{SS}\\
&=\sum_{\lambda\in\sigma(h)}
\bigl[q_{\beta,\tau_{\mathrm{im}}}(m_A(\lambda))\bigr]_{11}P_{\lambda}\\
&=\sum_{\lambda\in\sigma(h)}
F_{A,\beta,\tau_{\mathrm{im}}}(\lambda)P_{\lambda}
=F_{A,\beta,\tau_{\mathrm{im}}}(h).
\end{aligned}
\label{app:stab:covariance-functional}
\end{equation}
The bath-coupling dependence is retained in the scalar function
$F_{A,\beta,\tau_{\mathrm{im}}}$.

Let $\Gamma\coloneqq \Gamma_+\cup\Gamma_-$ consist of counterclockwise
contours enclosing the positive and negative parts of the spectrum of $h$, respectively.
We choose these contours and verify the analyticity of
$F_{A,\beta,\tau_{\mathrm{im}}}$ on and inside them below.
For such contours, the scalar Cauchy formula gives
$F_{A,\beta,\tau_{\mathrm{im}}}(\lambda)
=(2\pi i)^{-1}\int_\Gamma F_{A,\beta,\tau_{\mathrm{im}}}(z)/(z-\lambda)\,dz$
for each $\lambda\in\sigma(h)$. Substituting this identity into
\eqref{app:stab:covariance-functional} and using the resolvent expansion
$(z-h)^{-1}=\sum_{\lambda\in\sigma(h)}P_\lambda/(z-\lambda)$ yields
\begin{equation}
\begin{aligned}
F_{A,\beta,\tau_{\mathrm{im}}}(h)
&=\sum_{\lambda\in\sigma(h)}P_\lambda
\frac{1}{2\pi i}\int_\Gamma
\frac{F_{A,\beta,\tau_{\mathrm{im}}}(z)}{z-\lambda}\,dz\\
&=\frac{1}{2\pi i}\int_\Gamma
F_{A,\beta,\tau_{\mathrm{im}}}(z)(z-h)^{-1}\,dz.
\end{aligned}
\label{app:stab:covariance-contour}
\end{equation}
To bound $K_{SS,\gamma_A}$ in \eqref{app:stab:row-norm}, take absolute
values in \eqref{app:stab:covariance-contour}, integrate over the relative
imaginary time $\tau_{\mathrm{im}}$ as in
\eqref{app:stab:relative-time-integral}, and sum over the system endpoint.
For the conjugation-symmetric contours chosen below, Hermiticity gives
$[(z-h)^{-1}]_{yx}=\overline{[(\overline z-h)^{-1}]_{xy}}$,
so the same spatial bound covers both orientations. Writing $|\Gamma|$
for the total contour length, we obtain
\begin{equation}
\begin{aligned}
K_{SS,\gamma_A}
&\le\frac{|\Gamma|}{2\pi}
\left[\sup_{z\in\Gamma}\int_0^\beta
 e^{\gamma_A d_\beta(\tau_{\mathrm{im}},0)}
 |F_{A,\beta,\tau_{\mathrm{im}}}(z)|\,d\tau_{\mathrm{im}}\right]\times
\left[\sup_{z\in\Gamma,\,x\in S}
 \sum_{y\in S}|(z-h)^{-1}_{xy}|\right].
\end{aligned}
\label{app:stab:time-space-bound}
\end{equation}
The proof below bounds the term in the first square bracket by $C/\Delta$ and the term in the second square bracket by $(C/\Delta)(1+R\Lambda/\Delta)^D$.
Together with $|\Gamma|\le C\Lambda$, these estimates give
$K_{SS,\gamma_A}\le C\Lambda\Delta^{-2}(1+R\Lambda/\Delta)^D$,
which is the uniform bound \eqref{app:stab:Kcov}.

Choose $\Gamma_{+}$ as the counterclockwise boundary of the rectangle
\begin{equation}
\frac{\Delta}{2}\le\operatorname{Re}z\le\Lambda+\frac{\Delta}{2},
\qquad|\operatorname{Im}z|\le\frac{\Delta}{4},
\label{app:stab:contours}
\end{equation}
and let $\Gamma_{-}$ be its image under $z\mapsto-z$, also oriented
counterclockwise; see Fig.~\ref{app:stab:fig-contours}.
Their total length is at most $C\Lambda$, and their distance from
$\sigma(h)$ is at least $\Delta/4$.

\begin{figure}[htbp]
\centering
\begin{tikzpicture}[
  x=1cm,y=0.5cm,font=\small,
  every node/.style={inner sep=2pt},
  spectral contour/.style={draw=blue!45!black,line width=0.85pt},
  spectral guide/.style={draw=black!35,densely dashed,line width=0.4pt}
]
\fill[blue!3] (-5.5,-0.9) rectangle (-0.9,0.9);
\fill[blue!3] (0.9,-0.9) rectangle (5.5,0.9);

\draw[->,draw=black!50,line width=0.45pt]
  (-6.3,0) -- (6.3,0) node[right,text=black!70] {$\operatorname{Re}z$};
\draw[<->,draw=black!50,line width=0.45pt]
  (0,-1.8) -- (0,1.8) node[above,text=black!70] {$\operatorname{Im}z$};
\draw[spectral guide] (-0.9,0.9) -- (0.9,0.9);
\draw[spectral guide] (-0.9,-0.9) -- (0.9,-0.9);
\foreach \x in {-5.5,-0.9,0.9,5.5}
  \draw[spectral guide] (\x,-0.9) -- (\x,-1.3);

\draw[spectral contour] (-5.5,-0.9) rectangle (-0.9,0.9);
\draw[spectral contour] (0.9,-0.9) rectangle (5.5,0.9);
\draw[spectral contour,->] (-2.7,0.9) -- (-3.7,0.9);
\draw[spectral contour,->] (-3.7,-0.9) -- (-2.7,-0.9);
\draw[spectral contour,->] (3.7,0.9) -- (2.7,0.9);
\draw[spectral contour,->] (2.7,-0.9) -- (3.7,-0.9);
\node[above,text=blue!45!black] at (-3.2,1.0) {$\Gamma_-$};
\node[above,text=blue!45!black] at (3.2,1.0) {$\Gamma_+$};

\draw[black,line width=2pt] (-4.6,0) -- (-1.8,0);
\draw[black,line width=2pt] (1.8,0) -- (4.6,0);
\foreach \x in {-4.6,-1.8,1.8,4.6}
  \draw[black,line width=0.45pt] (\x,-0.065) -- (\x,0.065);
\node[below] at (-4.6,-0.08) {$-\Lambda$};
\node[below] at (-1.8,-0.08) {$-\Delta$};
\node[below] at (1.8,-0.08) {$\Delta$};
\node[below] at (4.6,-0.08) {$\Lambda$};
\node[below left] at (0,0) {$0$};

\draw[spectral guide] (5.5,0.9) -- (5.85,0.9)
  node[right] {$\Delta/4$};
\draw[spectral guide] (5.5,-0.9) -- (5.85,-0.9)
  node[right] {$-\Delta/4$};
\node[below] at (-5.5,-1.34) {$-\Lambda-\Delta/2$};
\node[below] at (-0.9,-1.34) {$-\Delta/2$};
\node[below] at (0.9,-1.34) {$\Delta/2$};
\node[below] at (5.5,-1.34) {$\Lambda+\Delta/2$};
\end{tikzpicture}
\caption{Counterclockwise contours $\Gamma_\pm$ enclosing intervals
that contain $\sigma(h)$ (black segments; schematic).}
\label{app:stab:fig-contours}
\end{figure}

For $z$ on or inside $\Gamma_{+}$ take $w\coloneqq z$, and for $z$ on or
inside $\Gamma_{-}$ take $w\coloneqq -z$. Thus $\operatorname{Re}w\ge\Delta/2$, $|\operatorname{Im}w|\le\operatorname{Re}w/2$,
and $|w|\le2\Lambda$. We analytically continue the energies and system weights in
Eqs.~\eqref{app:stab:energies} and~\eqref{app:stab:weights} by setting
\begin{equation}
E_{\mathrm{f}}(w,A)\coloneqq \frac{\sqrt{w^{2}+8A^{2}}+w}{2},\qquad E_{\mathrm{s}}(w,A)\coloneqq \frac{2A^{2}}{E_{\mathrm{f}}(w,A)},\qquad Z_{\mathrm{f},\mathrm{s}}(w,A)\coloneqq \frac{E_{\mathrm{f},\mathrm{s}}(w,A)}{\sqrt{w^{2}+8A^{2}}},\label{app:stab:analytic-energies}
\end{equation}
where the square root has positive real part. It is analytic on the
rectangles: with $x\coloneqq\operatorname{Re}w$ and
$y\coloneqq\operatorname{Im}w$, the real part of its radicand is
$x^{2}-y^{2}+8A^{2}>0$.

The needed bounds can be checked without estimating a nearly cancelling
difference of roots. Set $v_{1}+iv_{2}\coloneqq\sqrt{w^{2}+8A^{2}}$, with
$v_{1}>0$. Equating real and imaginary parts gives 
\begin{equation}
v_{2}=\frac{xy}{v_{1}},\qquad(v^{2}_{1}-x^{2})(v^{2}_{1}+y^{2})=8A^{2}v^{2}_{1}.
\label{app:stab:square-root-identities}
\end{equation}
Consequently $v_{1}\ge x$, $|v_{2}|\le|y|$, and $v^{2}_{1}\ge x^{2}+(32/5)A^{2}$.
Moreover $v^{2}_{1}\le x^{2}+8A^{2}$. These inequalities imply 
\begin{align}
\operatorname{Re}E_{\mathrm{f}} & \ge A, & \operatorname{Re}E_{\mathrm{f}} & \le2(\Lambda+A),\nonumber \\
\operatorname{Re}E_{\mathrm{s}} & \ge\frac{4A^{2}}{5(\Lambda+A)}\ge\frac{2g_{A}}{5}, & |\operatorname{Im}E_{\mathrm{f},\mathrm{s}}| & \le\frac{1}{2}\operatorname{Re}E_{\mathrm{f},\mathrm{s}},\label{app:stab:sector}\\
|\sqrt{w^{2}+8A^{2}}| & \ge c(\Delta+A).\nonumber 
\end{align}
For the sector bound, $|\operatorname{Im}E_{\mathrm{f}}|/\operatorname{Re}E_{\mathrm{f}}=|y|/v_{1}\le1/2$,
and taking the reciprocal preserves this ratio for $E_{\mathrm{s}}$.
Also $\operatorname{Re}E_{\mathrm{s}}=2A^{2}\operatorname{Re}E_{\mathrm{f}}/|E_{\mathrm{f}}|^{2}\ge(8/5)A^{2}/\operatorname{Re}E_{\mathrm{f}}$,
which proves its displayed lower bound. Fix $c_{\mathrm{gap}}$ small
enough that $\gamma_{A}\le\frac{1}{2}\operatorname{Re}E_{\mathrm{f},\mathrm{s}}$
throughout both contours. This choice is independent of $A$.

The thermal denominator for either complex energy $E$ specified below
obeys the uniform bound
\begin{equation}
|1+e^{-\beta E}|\ge\frac{1}{2}\qquad(\beta>0).\label{app:stab:thermal-denominator}
\end{equation}
Here $
E\in\{E_{\mathrm{f}}(w,A),E_{\mathrm{s}}(w,A)\}.
$
As above, $w=z$ on or inside $\Gamma_{+}$ and $w=-z$ on or inside
$\Gamma_{-}$. 
The reasoning of Eq.~\eqref{app:stab:thermal-denominator} is as followed: 
Equation~\eqref{app:stab:sector} gives
$\operatorname{Re}E>0$ and $|\operatorname{Im}E|\le\operatorname{Re}E/2$
for both choices. If $\beta\operatorname{Re}E\ge\log2$, the reverse
triangle inequality yields
\begin{equation}
|1+e^{-\beta E}|\ge1-|e^{-\beta E}|
=1-e^{-\beta\operatorname{Re}E}\ge\frac12.
\label{app:stab:thermal-large-real-part}
\end{equation}
Otherwise, the same sector bound implies
\begin{equation}
\begin{aligned}
|\beta\operatorname{Im}E|&\le\frac{\beta\operatorname{Re}E}{2}
<\frac{\log2}{2}<\frac{\pi}{2},\\
|1+e^{-\beta E}|&\ge\operatorname{Re}(1+e^{-\beta E})
=1+e^{-\beta\operatorname{Re}E}\cos(\beta\operatorname{Im}E)\ge1.
\end{aligned}
\label{app:stab:thermal-small-real-part}
\end{equation}
Together these prove Eq.~\eqref{app:stab:thermal-denominator}.

Diagonalizing the block in \eqref{app:stab:scalar-block} with the
energies and system weights of \eqref{app:stab:analytic-energies}
gives, on and inside the positive contour $\Gamma_+$, 
\begin{equation}
F_{A,\beta,\tau_{\mathrm{im}}}(z)=-\frac{Z_{\mathrm{f}}e^{-\tau_{\mathrm{im}}E_{\mathrm{f}}}}{1+e^{-\beta E_{\mathrm{f}}}}-\frac{Z_{\mathrm{s}}e^{-(\beta-\tau_{\mathrm{im}})E_{\mathrm{s}}}}{1+e^{-\beta E_{\mathrm{s}}}}.\label{app:stab:scalar-covariance}
\end{equation}
On and inside $\Gamma_{-}$, take $w=-z$. The fast term then contains
$e^{-(\beta-\tau_{\mathrm{im}})E_{\mathrm{f}}}$, and the soft term contains
$e^{-\tau_{\mathrm{im}}E_{\mathrm{s}}}$; their weights remain
$Z_{\mathrm{f}}(w,A)$ and $Z_{\mathrm{s}}(w,A)$, respectively.
The square root in Eq.~\eqref{app:stab:analytic-energies} is analytic,
and Eq.~\eqref{app:stab:thermal-denominator} ensures that neither
denominator vanishes. Thus $F_{A,\beta,\tau_{\mathrm{im}}}(z)$ is analytic
on and inside each contour, as required in Eq.~\eqref{app:stab:covariance-contour}.

To estimate both terms together, let $Z$ denote the system weight and
$E$ its associated complex energy, with the paired choices
\begin{equation}
(Z,E)\in\left\{
\bigl(Z_{\mathrm{f}}(w,A),E_{\mathrm{f}}(w,A)\bigr),
\bigl(Z_{\mathrm{s}}(w,A),E_{\mathrm{s}}(w,A)\bigr)
\right\}.
\label{app:stab:mode-pair}
\end{equation}
The imaginary-time circle distance satisfies
$d_{\beta}(\tau_{\mathrm{im}},0)=\min\{\tau_{\mathrm{im}},\beta-\tau_{\mathrm{im}}\}$.
For the factor $e^{-\tau_{\mathrm{im}}E}$, use
$d_{\beta}(\tau_{\mathrm{im}},0)\le\tau_{\mathrm{im}}$ together with
Eqs.~\eqref{app:stab:sector} and~\eqref{app:stab:thermal-denominator} to obtain
\begin{equation}
\int^{\beta}_{0}e^{\gamma_{A}d_{\beta}(\tau_{\mathrm{im}},0)}\left|\frac{Ze^{-\tau_{\mathrm{im}}E}}{1+e^{-\beta E}}\right|d\tau_{\mathrm{im}}\le\frac{2|Z|}{\operatorname{Re}E-\gamma_{A}}\le\frac{C|Z|}{|E|}.
\label{app:stab:single-mode-time-integral}
\end{equation}
For the factor $e^{-(\beta-\tau_{\mathrm{im}})E}$, the substitution
$\tau_{\mathrm{im}}\mapsto\beta-\tau_{\mathrm{im}}$ leaves the circle distance
unchanged and gives the same bound. The identity $Z/E=1/\sqrt{w^{2}+8A^{2}}$
now yields 
\begin{equation}
\sup_{z\in\Gamma_{+}\cup\Gamma_{-}}\int^{\beta}_{0}e^{\gamma_{A}d_{\beta}(\tau_{\mathrm{im}},0)}|F_{A,\beta,\tau_{\mathrm{im}}}(z)|\,d\tau_{\mathrm{im}}\le\frac{C}{\Delta}.\label{app:stab:scalar-integral}
\end{equation}
This is the analytic version of the compensation in \eqref{app:stab:compensation}.
For the system-only covariance the same argument uses energy $w$
on each contour. Its real part is at least $\Delta/2$, so the rate
$\Delta/4$ is allowed and gives the same bound.
This establishes the $C/\Delta$ bound for the first square bracket
in Eq.~\eqref{app:stab:time-space-bound}.  

It remains to bound the spatial resolvent sum in the second square
bracket of Eq.~\eqref{app:stab:time-space-bound}. Choose the spatial
decay rate $\mu_{\mathrm{sp}}\coloneqq R^{-1}\log(1+\Delta/(8\Lambda))$.
Following Ref.~\cite[Eq.~(A.4)]{Aza2019Conductivity}, define the
weighted hopping sum $S(h,\mu_{\mathrm{sp}})$ using the present
lattice distance. The finite hopping range gives $h_{xx'}=0$ whenever
$\operatorname{dist}(x,x')>R$; the row-sum bound in
Eq.~\eqref{app:model:freegap} then yields
\begin{equation}
\begin{aligned}
S(h,\mu_{\mathrm{sp}})
&\coloneqq\sup_x\sum_{x'}|h_{xx'}|
\left(e^{\mu_{\mathrm{sp}}\operatorname{dist}(x,x')}-1\right)\\
&\le\left(e^{\mu_{\mathrm{sp}}R}-1\right)
\sup_x\sum_{x'}|h_{xx'}|\\
&\le\Lambda\left(e^{\mu_{\mathrm{sp}}R}-1\right)
=\frac{\Delta}{8}.
\end{aligned}
\label{app:stab:weighted-hopping-sum}
\end{equation}
For the spectral distance $\delta_z\coloneqq\operatorname{dist}(z,\sigma(h))$,
the contour choice gives $\delta_z\ge\Delta/4>S(h,\mu_{\mathrm{sp}})$
for every $z\in\Gamma$. The finite-metric Combes--Thomas estimate
\cite[Theorem~A.1]{Aza2019Conductivity} therefore gives\footnote{The
cited estimate extends to the present finite metric space: its
exponential-weight proof uses only the triangle inequality
$|\operatorname{dist}(x,y)-\operatorname{dist}(x',y)|\le\operatorname{dist}(x,x')$,
the assumed row-sum bound, and the corresponding column-sum bound
from Hermiticity. All sums are finite, and the constants are unchanged.}
\begin{equation}
\begin{aligned}
|(z-h)^{-1}_{xy}|
&\le\frac{e^{-\mu_{\mathrm{sp}}\operatorname{dist}(x,y)}}
{\delta_z-S(h,\mu_{\mathrm{sp}})}\\
&\le\frac{8}{\Delta}e^{-\mu_{\mathrm{sp}}\operatorname{dist}(x,y)}.
\end{aligned}
\label{app:stab:resolvent}
\end{equation}
The reversed estimate follows by the adjoint and the reflected contour.
Using the ball-count bound in Eq.~\eqref{app:model:ball-growth}, we obtain
\begin{equation}
\begin{aligned}
\sum_{y}e^{-\mu_{\mathrm{sp}}\operatorname{dist}(x,y)}
&=\mu_{\mathrm{sp}}\int^{\infty}_{0}e^{-\mu_{\mathrm{sp}}\ell}|\{y:\operatorname{dist}(x,y)\le\ell\}|\,d\ell\\
&\le C_{\mathrm{vol}}\mu_{\mathrm{sp}}\int^{\infty}_{0}e^{-\mu_{\mathrm{sp}}\ell}(1+\ell)^{D}\,d\ell\\
&\le C(1+\mu^{-1}_{\mathrm{sp}})^{D}\le C\left(1+\frac{R\Lambda}{\Delta}\right)^{D}.
\end{aligned}
\label{app:stab:spatial-sum}
\end{equation}
The last line uses $s=\mu_{\mathrm{sp}}\ell$,
$1+s/\mu_{\mathrm{sp}}\le(1+s)(1+\mu_{\mathrm{sp}}^{-1})$, and
$\mu_{\mathrm{sp}}^{-1}\le9R\Lambda/\Delta$.

Equations~\eqref{app:stab:resolvent} and~\eqref{app:stab:spatial-sum}
bound the second square bracket in Eq.~\eqref{app:stab:time-space-bound}
by $(C/\Delta)(1+R\Lambda/\Delta)^D$.
The first square bracket is bounded by $C/\Delta$ in
Eq.~\eqref{app:stab:scalar-integral}, while the contours in
Eq.~\eqref{app:stab:contours} satisfy $|\Gamma|\le C\Lambda$.
Substituting these three estimates into
Eq.~\eqref{app:stab:time-space-bound} proves Eq.~\eqref{app:stab:Kcov}.

Finally, let $U$ diagonalize the full system--bright one-particle
matrix $h_A$, with one-particle energies $\varepsilon_j$ satisfying
$|\varepsilon_j|\ge g_A$. For $0<\tau_{\mathrm{im}}<\beta$,
\begin{equation}
\begin{aligned}
h_A
&=U\operatorname{diag}(\varepsilon_1,\ldots,\varepsilon_{2N})U^\dagger,\\
C_{A,\beta}(\tau_{\mathrm{im}},0)
&=-U\operatorname{diag}\left(
\frac{e^{-\tau_{\mathrm{im}}\varepsilon_j}}
{1+e^{-\beta\varepsilon_j}}
\right)_{j=1}^{2N}U^\dagger.
\end{aligned}
\label{app:stab:full-covariance-spectral}
\end{equation}
For positive energies $\varepsilon\ge g_A$ and negative energies
$\varepsilon=-E$ with $E\ge g_A$,
\begin{equation}
\begin{aligned}
\frac{e^{-\tau_{\mathrm{im}}\varepsilon}}
{1+e^{-\beta\varepsilon}}
&\le e^{-g_A\tau_{\mathrm{im}}}
&&(\varepsilon\ge g_A),\\
\frac{e^{\tau_{\mathrm{im}}E}}{1+e^{\beta E}}
&=\frac{e^{-(\beta-\tau_{\mathrm{im}})E}}
{1+e^{-\beta E}}
\le e^{-g_A(\beta-\tau_{\mathrm{im}})}
&& (E\ge g_A).
\end{aligned}
\label{app:stab:full-covariance-mode-decay}
\end{equation}
Thus the $\ell^2$ operator norm satisfies
\begin{equation}
\begin{aligned}
\|C_{A,\beta}(\tau_{\mathrm{im}},0)\|
&=\max_j
\frac{e^{-\tau_{\mathrm{im}}\varepsilon_j}}
{1+e^{-\beta\varepsilon_j}}\\
&\le\max\left\{
e^{-g_A\tau_{\mathrm{im}}},
e^{-g_A(\beta-\tau_{\mathrm{im}})}
\right\}\\
&=e^{-g_A d_\beta(\tau_{\mathrm{im}},0)},
\end{aligned}
\label{app:stab:pointwise}
\end{equation}
where $d_\beta(\tau_{\mathrm{im}},0)
=\min\{\tau_{\mathrm{im}},\beta-\tau_{\mathrm{im}}\}$.
At equal times,
$0\le C_{A,\beta}(0,0)=(I+e^{\beta h_A})^{-1}\le I$,
so the same bound holds.

Using
$|C_{A,\beta}(\tau_{\mathrm{im}},x;0,y)|
\le\|C_{A,\beta}(\tau_{\mathrm{im}},0)\|$
and the relative-time reduction for both orientations in
\eqref{app:stab:full-row-norm}, the sum over $2N$ labels gives,
for $\gamma_A<g_A$,
\begin{equation}
\begin{aligned}
K_{\mathrm{full},\gamma_A}
&\le 2N\int_0^\beta
e^{-(g_A-\gamma_A)d_\beta(\tau_{\mathrm{im}},0)}
\,d\tau_{\mathrm{im}}\\
&=4N\int_0^{\beta/2}
e^{-(g_A-\gamma_A)\tau_{\mathrm{im}}}
\,d\tau_{\mathrm{im}}\\
&=\frac{4N}{g_A-\gamma_A}
\left(1-e^{-(g_A-\gamma_A)\beta/2}\right)
\le\frac{4N}{g_A-\gamma_A},
\end{aligned}
\label{app:stab:full-time-integral}
\end{equation}
proving \eqref{app:stab:Kfull} uniformly in $\beta$.
For the uncoupled system matrix $h$, $|\lambda|\ge\Delta$ gives
the same pointwise estimate with $g_A$ replaced by $\Delta$.
\end{proof}

\subsection{Connected correlations from system-only interactions}

We expand the connected correlation in the coupling $u$ and bound
its coefficients by a geometric sequence. The covariance bounds
\eqref{app:stab:Kcov} and \eqref{app:stab:Kfull} will control the
time and mode sums. The tree representation below specifies which
covariance factors must be summed.

Write the interaction in normal-ordered
system CAR monomials $M$. Define the monomial coefficients $v_M$
and the local coefficient sum $v_{\mathrm{loc}}$ by
\begin{equation}
\sum_X V_X=\sum_M v_M M,\qquad
v_{\mathrm{loc}}\coloneqq
\sup_{x\in S}\sum_{M:\,\operatorname{supp}(M)\ni x}|v_M|.
\label{app:stab:interaction-coefficients}
\end{equation}
For a monomial with creation count $r$ and annihilation count $s$,
its field degree is $\deg M=r+s$. The interaction degree bound
$m_{\mathrm{int}}\ge2$ satisfies $\deg M\le m_{\mathrm{int}}$ for
every nonzero interaction monomial.

The interaction range assumption $\operatorname{diam}(X)\le R$
and the ball-count bound \eqref{app:model:ball-growth} give
$|X|\le C_{\mathrm{vol}}(1+R)^D$, where $|X|$ counts fermionic modes.
Choose the local support bound
$s_V\coloneqq\lceil C_{\mathrm{vol}}(1+R)^D\rceil\ge1$,
independent of $N,A,\beta,u$.

The next lemma controls the local coefficient sum
$v_{\mathrm{loc}}$ and the interaction degree bound $m_{\mathrm{int}}$
needed for the tree expansion below.

\begin{lem}[Monomial expansion of local interactions]
\label{app:stab:monomial-expansion}
Every monomial $M$ in
 \eqref{app:stab:interaction-coefficients} has degree at most $2s_V$,
so one may choose $m_{\mathrm{int}}=2s_V$. Moreover,
\begin{equation}
|u|v_{\mathrm{loc}}\le C_{\mathrm{mon}}J_{\mathrm{int}}.
\label{app:stab:monomial-norm}
\end{equation}
Here the monomial-expansion constant $C_{\mathrm{mon}}$ bounds the
sum of absolute monomial coefficients of each $V_X$ by
$C_{\mathrm{mon}}\|V_X\|$. One may take $C_{\mathrm{mon}}=8^{s_V}$,
which depends only on $s_V$ and is independent of $N,A,\beta,u$.
\end{lem}

\begin{proof}
Each mode appears at most once among the creation fields and once
among the annihilation fields of a nonzero normal-ordered monomial.
Hence $\deg M\le2|X|\le2s_V$ for every monomial arising from $V_X$.

The local Fock space on $|X|$ modes has dimension $2^{|X|}$.
In its occupation basis, write the matrix units as
$E_{IJ}=|I\rangle\langle J|$ and the matrix coefficients as
$a_{IJ}=\langle I|V_X|J\rangle$. There are $4^{|X|}$ matrix units,
and $|a_{IJ}|\le\|V_X\|$. Fix an ordering of the modes in $X$.
For $I=\{i_1<\cdots<i_r\}\subseteq X$, define the creation product
$c_I^\dagger=c_{i_1}^\dagger\cdots c_{i_r}^\dagger$.
Similarly, for $J=\{j_1<\cdots<j_s\}\subseteq X$, the annihilation
product is $c_J=(c_J^\dagger)^\dagger=c_{j_s}\cdots c_{j_1}$.
Empty products are the identity. With the local vacuum state $|0_X\rangle$,
choose the occupation states as $|I\rangle=c_I^\dagger|0_X\rangle$.
The vacuum projector $P_{\mathrm{vac},X}$ then gives
\begin{equation}
E_{IJ}=c_I^\dagger P_{\mathrm{vac},X}c_J,
\qquad
P_{\mathrm{vac},X}=\prod_{j\in X}(1-c_j^\dagger c_j).
\label{app:stab:matrix-unit-vacuum}
\end{equation}
Expanding the projector in \eqref{app:stab:matrix-unit-vacuum}
gives at most $2^{|X|}$ normal-ordered
monomials per matrix unit, with coefficients of absolute value one.
Thus the local monomial coefficients $v_{X,M}$ in
$V_X=\sum_M v_{X,M}M$ satisfy
\begin{equation}
\sum_M|v_{X,M}|
\le 4^{|X|}2^{|X|}\|V_X\|
\le 8^{s_V}\|V_X\|.
\label{app:stab:local-monomial-sum}
\end{equation}
Collecting equal monomials gives $v_M=\sum_X v_{X,M}$, with
$v_{X,M}=0$ unless $\operatorname{supp}(M)\subseteq X$.
For fixed $x$, the triangle inequality and
\eqref{app:stab:local-monomial-sum} therefore give
\begin{equation}
\begin{aligned}
\sum_{M:\,\operatorname{supp}(M)\ni x}|v_M|
&=\sum_{M:\,\operatorname{supp}(M)\ni x}
\left|\sum_{X:\,X\ni x}v_{X,M}\right|\\
&\le\sum_{X:\,X\ni x}
\sum_{M:\,\operatorname{supp}(M)\ni x}|v_{X,M}|\\
&\le\sum_{X:\,X\ni x}\sum_M|v_{X,M}|\\
&\le 8^{s_V}\sum_{X:\,X\ni x}\|V_X\|.
\end{aligned}
\label{app:stab:global-monomial-sum}
\end{equation}
Taking the supremum over $x$ and multiplying by $|u|$ proves
\eqref{app:stab:monomial-norm}, where we also use the definition
$J_{\mathrm{int}}=|u|\sup_x\sum_{X\ni x}\|V_X\|$
following \eqref{app:model:localnorm}.
\end{proof}

For the following expansion, interaction and Hamiltonian symbols denote
their scalar-subtracted versions, and interaction monomial sums exclude the identity.
This leaves spectral gaps, connected correlations, and $v_{\mathrm{loc}}$
unchanged; $J_{\mathrm{int}}$ retains its definition from the original
local terms, so \eqref{app:stab:monomial-norm} still applies.

Fix the volume and amplitude $A>0$. Display the interaction-strength
dependence explicitly as $H_{\mathrm{br}}(A;u)$, with free Hamiltonian
$H_{\mathrm{br}}(A;0)$ and interaction operator $V\coloneqq\sum_Mv_MM$,
so that $H_{\mathrm{br}}(A;u)=H_{\mathrm{br}}(A;0)+uV$.
Define the partition function
$Z_\beta(u)\coloneqq\operatorname{Tr}e^{-\beta H_{\mathrm{br}}(A;u)}$
and the Gibbs expectation
$\langle O\rangle_{\beta,u}\coloneqq
Z_\beta(u)^{-1}\operatorname{Tr}(e^{-\beta H_{\mathrm{br}}(A;u)}O)$.
For observables $O_1,O_2$ and relative imaginary time
$\tau_{\mathrm{im}}\in[0,\beta]$, define the connected
correlation $\mathcal G_{\beta,u}(\tau_{\mathrm{im}})$ by
\begin{equation}
\mathcal G_{\beta,u}(\tau_{\mathrm{im}})\coloneqq
\frac{\operatorname{Tr}\!\left[
e^{-(\beta-\tau_{\mathrm{im}})H_{\mathrm{br}}(A;u)}O_1
e^{-\tau_{\mathrm{im}} H_{\mathrm{br}}(A;u)}O_2\right]}{Z_\beta(u)}
-\langle O_1\rangle_{\beta,u}
\langle O_2\rangle_{\beta,u}.
\label{app:stab:connected}
\end{equation}

\begin{lem}[Uniform convergence from system-only interaction vertices]
\label{app:stab:decay}
There is a constant $C_{\mathrm{deg}}$, depending only on
$m_{\mathrm{int}}$, such that if
\begin{equation}
C_{\mathrm{deg}}K_{\mathrm{cov}}|u|v_{\mathrm{loc}}\le\frac12,
\label{app:stab:tree-smallness}
\end{equation}
then the actual Gibbs connected correlation
\eqref{app:stab:connected} of any finite polynomial observables
$O_1,O_2$ satisfies
\begin{equation}
|\mathcal G_{\beta,u}(\tau_{\mathrm{im}})|
\le C_{\mathrm{obs}}e^{-\gamma_A d_\beta(\tau_{\mathrm{im}},0)}.
\label{app:stab:correlation-decay}
\end{equation}
Here the decay rate $\gamma_A=c_{\mathrm{gap}}g_A$ is fixed as in
Lemma~\ref{app:stab:covariance}, with the free-gap scale $g_A$
defined in \eqref{app:stab:free-gap-scale}.
The constant $C_{\mathrm{obs}}$ may depend on $N,A,u,O_1,O_2$,
but is independent of $\beta$ and $\tau_{\mathrm{im}}$. Odd and number-changing
observables are included. The same conclusion holds for $H_u$
alone, with decay rate $\Delta/4$ and the same smallness condition.
\end{lem}

\begin{proof}
We first express each interaction order through free correlations,
then bound a single tree and sum over all trees and orders.

\medskip
\noindent\emph{1. Expand in the interaction strength.}
At fixed volume and $\beta$, the connected correlation is analytic
near $u=0$ by \cite[Theorem~6, Eq.~(42), and Sec.~4.1]{DRS}.
Define its Taylor coefficients $\mathcal G_{\beta,p}(\tau_{\mathrm{im}})$ by
\begin{equation}
\mathcal G_{\beta,u}(\tau_{\mathrm{im}})
=\sum_{p\ge0}u^p\mathcal G_{\beta,p}(\tau_{\mathrm{im}}),
\qquad
\mathcal G_{\beta,p}(\tau_{\mathrm{im}})
=\left.\frac1{p!}\partial_u^p\mathcal G_{\beta,u}(\tau_{\mathrm{im}})\right|_{u=0}.
\label{app:stab:taylor-expansion}
\end{equation}
This initial neighborhood may depend on the volume and $\beta$;
the estimates below give the uniform interaction threshold.

For even observables, define
the free expectation $\langle\cdot\rangle_0\coloneqq
\langle\cdot\rangle_{\beta,0}$, the free observable insertions
$X\coloneqq e^{\tau_{\mathrm{im}} H_{\mathrm{br}}(A;0)}O_1e^{-\tau_{\mathrm{im}} H_{\mathrm{br}}(A;0)}$ and $Y\coloneqq O_2$,
and the interaction-picture insertion
$V_I(t)\coloneqq e^{tH_{\mathrm{br}}(A;0)}Ve^{-tH_{\mathrm{br}}(A;0)}$.
The time-ordering operation $\mathcal T_{\tau_{\mathrm{im}}}$ puts later insertions
to the left, retaining the prescribed order at equal times.
The Duhamel expansion rewrites the two-point expectation as
\begin{equation}
\frac{\operatorname{Tr}\!\left[
e^{-(\beta-\tau_{\mathrm{im}})H_{\mathrm{br}}(A;u)}O_1e^{-\tau_{\mathrm{im}} H_{\mathrm{br}}(A;u)}O_2\right]}{Z_\beta(u)}
=\frac{\left\langle\mathcal T_{\tau_{\mathrm{im}}}\!\left[
XYe^{-u\int_0^\beta V_I(t)\,dt}\right]\right\rangle_0}
{\left\langle\mathcal T_{\tau_{\mathrm{im}}}
e^{-u\int_0^\beta V_I(t)\,dt}\right\rangle_0}.
\label{app:stab:duhamel-ratio}
\end{equation}
Here $X,Y$ have times $\tau_{\mathrm{im}},0$, and time ordering acts on every
insertion in the exponential series. The denominator equals
$Z_\beta(u)/Z_\beta(0)$. Expanding both numerator and denominator
and subtracting the one-point product in \eqref{app:stab:connected}
gives free connected cumulants. For $n\ge1$ even operator
insertions $X_1,\ldots,X_n$ at prescribed imaginary times, introduce
commuting auxiliary sources $\eta_1,\ldots,\eta_n$.
Let $\Pi_n$ denote the set of partitions of $\{1,\ldots,n\}$
into nonempty blocks, and let $|\pi|$ be the number of blocks in
$\pi$. Define the free connected cumulant
$\langle X_1;\ldots;X_n\rangle_0^c$ by
\begin{equation}
\begin{aligned}
\langle X_1;\ldots;X_n\rangle_0^c
&\coloneqq
\left.\partial_{\eta_1}\cdots\partial_{\eta_n}
\log\left\langle\mathcal T_{\tau_{\mathrm{im}}}
\prod_{j=1}^n(1+\eta_jX_j)\right\rangle_0
\right|_{\boldsymbol\eta=0}\\
&=\sum_{\pi\in\Pi_n}(-1)^{|\pi|-1}(|\pi|-1)!
\prod_{B\in\pi}\left\langle\mathcal T_{\tau_{\mathrm{im}}}
\prod_{j\in B}X_j\right\rangle_0.
\end{aligned}
\label{app:stab:cumulant-definition}
\end{equation}
Each block product retains the prescribed order at equal times.
The two-observable form of the logarithmic source identity is
given in \cite[Sec.~4.1, Eqs.~(62)--(63)]{DRS}. At order $p$,
\begin{equation}
\mathcal G_{\beta,p}(\tau_{\mathrm{im}})
=\frac{(-1)^p}{p!}
\int_{[0,\beta]^p}dt_1\cdots dt_p
\left\langle X;Y;
V_I(t_1);\ldots;V_I(t_p)\right\rangle_0^c.
\label{app:stab:cumulant-coefficient}
\end{equation}
The cumulant now connects all $p+2$ insertions. The factor $1/p!$
converts the ordered time simplex to the full cube, and $(-1)^p$
comes from the exponential in \eqref{app:stab:duhamel-ratio}; both
are included in $\mathcal G_{\beta,p}$. The tree formula used next
holds for arbitrary observable operators, including odd ones,
by \cite[Theorem~6]{DRS}, and applies to the full scope of the lemma.

\medskip
\noindent\emph{2. Represent the connected coefficients by trees.}
Expand every insertion into normal-ordered monomials. A vertex
represents one whole monomial, and each creation or annihilation
operator in it is a field. Wick's rule pairs fields through free
covariances. A contraction between two vertices is an edge;
the connected cumulant retains the contributions connecting all
vertices. De Roeck and Salmhofer reorganize these contributions
exactly into trees with residual determinants
\cite[Theorem~6, Eqs.~(43)--(46)]{DRS}.

At order $p$, let $M_1,\ldots,M_p$ be interaction monomials and
$M_{p+1},M_{p+2}$ observable monomials selected from $O_1,O_2$.
There are $p+2$ vertices and $p+1$ tree edges. Fix these monomials
and their vertex times $t_j$, with $t_{p+1}=\tau_{\mathrm{im}}$ and $t_{p+2}=0$.
Let the directed tree be $T$ and its field choices be $\alpha$.
Each edge uses one creation field and one annihilation field,
and no field is used twice. For $e=(i,j)$, denote the selected
creation-mode label at $i$ by $x_e$ and the annihilation-mode label
at $j$ by $y_e$. The corresponding contraction is
$C_{A,\beta}(t_j,y_e;t_i,x_e)$, using the free covariance in
Eq.~\eqref{app:stab:covariance-definition} with its annihilation
argument first. The edge direction records creation to annihilation,
independently of the ordering of the vertex times.

Define the tree contribution $\mathcal A_{T,\alpha}$, its
contraction sign $\sigma_{T,\alpha}\in\{-1,1\}$, the residual
covariance matrix $D_{T,\alpha}(Q)$, the interpolation matrix
$Q$, and its probability measure $\mu_T$. In our covariance convention,
the representation in \cite[Theorem~6, Eqs.~(43)--(46)]{DRS} reads
\begin{equation}
\mathcal A_{T,\alpha}
=\sigma_{T,\alpha}
\underbrace{\prod_{e=(i,j)\in T}
C_{A,\beta}(t_j,y_e;t_i,x_e)}_{p+1\text{ selected contractions}}
\underbrace{\int\det D_{T,\alpha}(Q)\,d\mu_T(Q)}_{
\text{remaining contractions}}.
\label{app:stab:tree-representation}
\end{equation}
The rows and columns of $D_{T,\alpha}(Q)$ are the unused creation
and annihilation fields, respectively. For a row $(i,x)$ and a column
$(j,y)$, the matrix entry is $Q_{ij}C_{A,\beta}(t_j,y;t_i,x)$,
with $Q\succeq0$ and $Q_{jj}=1$.
The probability measure supplies the weights for the
exact reorganization; contractions forming loops are retained
in the determinant.

To recover the coefficient, write the observable coefficients
$a_M^{(i)}$ through $O_i=\sum_Ma_M^{(i)}M$.
Writing $d\mathbf t=dt_1\cdots dt_p$, the coefficient formula in
\cite[Theorem~6, Eqs.~(42)--(46)]{DRS} becomes
\begin{equation}
\begin{aligned}
\mathcal G_{\beta,p}(\tau_{\mathrm{im}})
={}\frac{(-1)^p}{p!}
\sum_{M_1,\ldots,M_{p+2}}
\left(\prod_{j=1}^pv_{M_j}\right)
a_{M_{p+1}}^{(1)}a_{M_{p+2}}^{(2)}\times\int_{[0,\beta]^p}\sum_{T,\alpha}
\mathcal A_{T,\alpha}\,d\mathbf t.
\end{aligned}
\label{app:stab:coefficient-tree-sum}
\end{equation}
The sign $\sigma_{T,\alpha}$ contains the fermionic permutation
signs for the field and row/column orders fixed above. It does not
include the overall Duhamel factor $(-1)^p$ in
\eqref{app:stab:coefficient-tree-sum}.

The remaining contractions can be bounded together. If the
remaining creation and annihilation counts are $\bar\nu,\nu$,
the determinant bound of \cite[Corollary~9, Eq.~(81)]{DRS} gives
\begin{equation}
\left|\int\det D_{T,\alpha}(Q)\,d\mu_T(Q)\right|
\le\sup_Q|\det D_{T,\alpha}(Q)|
\le 2^{\bar\nu+\nu}.
\label{app:stab:determinant-bound}
\end{equation}
Unequal counts give zero. Our finite Fock space, self-adjoint
one-particle matrix $h_A$, even interaction, and normalized mode
vectors satisfy the hypotheses of both cited results.
Thus only the selected tree covariances require time and mode
summations in the next step.

\medskip
\noindent\emph{3. Extract time decay and integrate one tree.}
For a fixed tree, let $P$ be the unique path between its two
observable vertices. The path triangle inequality
\cite[Sec.~3.2, Eqs.~(49)--(50)]{DRS} gives
\begin{equation}
\begin{gathered}
d_\beta(\tau_{\mathrm{im}},0)\le\sum_{(i,j)\in P}d_\beta(t_i,t_j),\\
\begin{aligned}
e^{\gamma_A d_\beta(\tau_{\mathrm{im}},0)}
\prod_{e=(i,j)\in T}|C_{A,\beta}(t_j,y_e;t_i,x_e)|
\le\prod_{e=(i,j)\in T}
\left[e^{\gamma_A d_\beta(t_i,t_j)}
|C_{A,\beta}(t_j,y_e;t_i,x_e)|\right].
\end{aligned}
\end{gathered}
\label{app:stab:tree-time-weight}
\end{equation}
Since $\gamma_A<g_A$, \eqref{app:stab:pointwise} implies
\begin{equation}
e^{\gamma_A d_\beta(t_i,t_j)}
|C_{A,\beta}(t_j,y_e;t_i,x_e)|
\le e^{-(g_A-\gamma_A)d_\beta(t_i,t_j)}\le1.
\label{app:stab:weighted-edge-bound}
\end{equation}
Choose an edge $e_*$ on $P$. Replacing its weighted factor by one gives
\begin{equation}
\begin{aligned}
&\prod_{e=(i,j)\in T}
\left[e^{\gamma_A d_\beta(t_i,t_j)}
|C_{A,\beta}(t_j,y_e;t_i,x_e)|\right]\\
\qquad\le&\prod_{e=(i,j)\in T\setminus\{e_*\}}
\left[e^{\gamma_A d_\beta(t_i,t_j)}
|C_{A,\beta}(t_j,y_e;t_i,x_e)|\right],
\qquad |T\setminus\{e_*\}|=p.
\end{aligned}
\label{app:stab:tree-cut}
\end{equation}
Cutting this edge gives two trees, each rooted at an observable.
There are now $p$ edges for the $p$ interaction vertices whose
times must be integrated; both root times stay fixed.
For $p=0$, cutting the sole edge leaves just the two observable
roots, and the remaining product is the empty product, equal to one.

Eliminate interaction vertices from the leaves towards the two
observable roots. The weighted full covariance norm $K_{\mathrm{full},\gamma_A}$
in \eqref{app:stab:full-row-norm} has the $\beta$-independent
upper bound $K_{\mathrm{full}}$ from \eqref{app:stab:Kfull}:
\begin{equation}
K_{\mathrm{full},\gamma_A}
\le K_{\mathrm{full}}\coloneqq\frac{CN}{g_A-\gamma_A}.
\label{app:stab:full-bound-constant}
\end{equation}
Fix the leaf's field slot $a$ used by its parent edge and denote
its mode label in monomial $M$ by $x_a(M)$. At fixed parent
time $s$ and mode $y$, group the monomials by $x_a(M)=x$.
Equation~\eqref{app:stab:interaction-coefficients} and the covariance
norms \eqref{app:stab:row-norm}, \eqref{app:stab:full-row-norm},
bounded in \eqref{app:stab:Kcov} and \eqref{app:stab:Kfull}, imply
\begin{equation}
\begin{aligned}
&\sum_M|v_M|\int_0^\beta e^{\gamma_A d_\beta(t,s)}
|C_{A,\beta}(t,x_a(M);s,y)|\,dt\\
\quad=&\sum_{x\in S}
\left(\sum_{M:\,x_a(M)=x}|v_M|\right)
\int_0^\beta e^{\gamma_A d_\beta(t,s)}
|C_{A,\beta}(t,x;s,y)|\,dt\\
\quad\le& v_{\mathrm{loc}}
\begin{cases}
K_{\mathrm{cov}},&y\in S,\\
K_{\mathrm{full}},&y\in\mathcal I\setminus S.
\end{cases}
\end{aligned}
\label{app:stab:leaf-bound}
\end{equation}
Only monomials with the selected field slot are summed, and the
same estimate holds for reversed covariance arguments. The selected
field slots and observable monomials are fixed here and will be
summed in Step~4. An edge touching an
observable can always use $K_{\mathrm{full}}$; if its observable
mode is in $S$, the sharper bound $K_{\mathrm{cov}}$ also applies.

Apply \eqref{app:stab:leaf-bound} successively towards the two roots. Let $k_i\coloneqq\deg M_{p+i}$ be the
observable monomial degrees, and let $p_{\mathrm{ext}}$ count the
surviving edges touching the roots. Distinct edges use distinct
observable fields, so $p_{\mathrm{ext}}\le k_1+k_2$.
Define the external covariance factor
$\chi_{\mathrm{ext}}\coloneqq\max\{1,K_{\mathrm{full}}/K_{\mathrm{cov}}\}$.
The integrated tree factors are bounded by
\begin{equation}
\begin{aligned}
&
v_{\mathrm{loc}}^p
K_{\mathrm{cov}}^{p-p_{\mathrm{ext}}}
K_{\mathrm{full}}^{p_{\mathrm{ext}}}
\le (K_{\mathrm{cov}}v_{\mathrm{loc}})^p
\chi_{\mathrm{ext}}^{k_1+k_2}.
\end{aligned}
\label{app:stab:mixed-tree-bound}
\end{equation}
This bounded number of external edges keeps $K_{\mathrm{full}}$
out of the factor raised to the interaction order $p$.

\medskip
\noindent\emph{4. Sum the trees and bound each coefficient.}
We now include the monomial coefficients, field choices, and
trees in \eqref{app:stab:coefficient-tree-sum}.
For an observable polynomial $O=\sum_Ma_MM$, define its field
degree $k(O)\coloneqq\max_{a_M\ne0}\deg M$, with $k(0)=0$,
and coefficient sum $L(O)\coloneqq\sum_M|a_M|$.
From now on $k_i\coloneqq k(O_i)$ bounds every monomial degree
at observable $i$. Scalar observable components cancel from the
connected correlation, so only nonconstant monomials contribute.
Define the observable weight $L_{\mathrm{obs}}(O)$ by
\begin{equation}
L_{\mathrm{obs}}(O)\coloneqq
2^{k(O)}
\max\{1,k(O)\}^{k(O)}
L(O).
\label{app:stab:observable-weight}
\end{equation}
The degree-zero factor is one.

There are $(p+2)^p$ labeled trees and at most $2^{p+1}$
orientations. A tree has total degree $2(p+1)$, and each observable
vertex has degree at least one. Thus at most $2p$ fields are
selected at interaction vertices, giving at most
$m_{\mathrm{int}}^{2p}$ choices. At observable $i$, the number
of field choices is at most $\max\{1,k_i\}^{k_i}$.
The residual determinant bound gives
$2^{\bar\nu+\nu}\le2^{m_{\mathrm{int}}p+k_1+k_2}$.
As in \cite[Sec.~3.2, proof of Lemma~7]{DRS}, the Taylor factorial
controls the tree count:
\begin{equation}
\begin{aligned}
\frac{(p+2)^p}{p!}
2^{p+1}
m_{\mathrm{int}}^{2p}
2^{m_{\mathrm{int}}p}
\le
2e^2\left(e\,2^{m_{\mathrm{int}}+1}m_{\mathrm{int}}^2\right)^p,
\qquad p\ge0.
\end{aligned}
\label{app:stab:tree-count}
\end{equation}
Here $(p+2)^p/p!\le e^{p+2}$, including $p=0$.
Choose the degree constant
$C_{\mathrm{deg}}\ge e\,2^{m_{\mathrm{int}}+1}m_{\mathrm{int}}^2$
and define the observable prefactor
$C_{\mathrm{ext}}\coloneqq
CL_{\mathrm{obs}}(O_1)L_{\mathrm{obs}}(O_2)
\chi_{\mathrm{ext}}^{k_1+k_2}$, with a numerical constant
$C\ge2e^2$. It is independent of $p,\beta,\tau_{\mathrm{im}}$.
Combining \eqref{app:stab:determinant-bound},
\eqref{app:stab:tree-time-weight}, \eqref{app:stab:tree-cut},
\eqref{app:stab:mixed-tree-bound}, and \eqref{app:stab:tree-count}
now yields
\begin{equation}
e^{\gamma_A d_\beta(\tau_{\mathrm{im}},0)}|\mathcal G_{\beta,p}(\tau_{\mathrm{im}})|
\le
\underbrace{C_{\mathrm{ext}}}_{\text{observable factors}}
\underbrace{(C_{\mathrm{deg}}K_{\mathrm{cov}}v_{\mathrm{loc}})^p}_{
\text{one factor per interaction insertion}}.
\label{app:stab:coefficient-bound}
\end{equation}
The tree identity and determinant bound used in this estimate
apply to arbitrary observables by \cite[Theorem~6]{DRS}.
Decomposing into even and odd parts therefore includes odd and
number-changing observables; mixed-parity correlations vanish
because the Gibbs state is even.

\medskip
\noindent\emph{5. Sum the series and recover the Gibbs correlation.}
For fixed $\tau_{\mathrm{im}}$, define the Taylor sum
$\mathcal S_\beta(\zeta)\coloneqq
\sum_{p\ge0}\zeta^p\mathcal G_{\beta,p}(\tau_{\mathrm{im}})$ for complex
coupling $\zeta$, and the expansion parameter
$q\coloneqq C_{\mathrm{deg}}K_{\mathrm{cov}}|u|v_{\mathrm{loc}}$.
The smallness condition \eqref{app:stab:tree-smallness} gives
$q\le1/2$, hence
\begin{equation}
|\mathcal S_\beta(u)|
\le C_{\mathrm{ext}}
e^{-\gamma_A d_\beta(\tau_{\mathrm{im}},0)}
\sum_{p\ge0}q^p
\le2C_{\mathrm{ext}}e^{-\gamma_A d_\beta(\tau_{\mathrm{im}},0)}.
\label{app:stab:geometric-sum}
\end{equation}
It remains to identify this sum with the actual correlation
throughout the stated weak-coupling range, since
\eqref{app:stab:taylor-expansion} was initially local near zero.

At fixed finite volume and $\beta$, the complex-coupling
Hamiltonian $H_{\mathrm{br}}(A;\zeta)$ is obtained by replacing
$u$ with $\zeta$. Its partition function $Z_\beta(\zeta)$ and
all unnormalized traces are entire. Define the connected
numerator $\mathcal N_\beta(\zeta)$ by
\begin{equation}
\begin{aligned}
\mathcal N_\beta(\zeta)\coloneqq{}&
Z_\beta(\zeta)\operatorname{Tr}\!\left[
e^{-(\beta-\tau_{\mathrm{im}})H_{\mathrm{br}}(A;\zeta)}O_1
e^{-\tau_{\mathrm{im}} H_{\mathrm{br}}(A;\zeta)}O_2\right]\\
&-\operatorname{Tr}\!\left[e^{-\beta H_{\mathrm{br}}(A;\zeta)}O_1\right]
\operatorname{Tr}\!\left[e^{-\beta H_{\mathrm{br}}(A;\zeta)}O_2\right].
\end{aligned}
\label{app:stab:connected-numerator}
\end{equation}
The local Taylor identity from \cite[Theorem~6]{DRS} gives
$Z_\beta^2\mathcal S_\beta=\mathcal N_\beta$ near zero.
The coefficient bound \eqref{app:stab:coefficient-bound} makes
$\mathcal S_\beta$ analytic on the disk
$C_{\mathrm{deg}}K_{\mathrm{cov}}v_{\mathrm{loc}}|\zeta|<1$.
The identity theorem extends the equality across this disk:
\begin{equation}
Z_\beta(\zeta)^2\mathcal S_\beta(\zeta)
=\mathcal N_\beta(\zeta)
\qquad
\bigl(C_{\mathrm{deg}}K_{\mathrm{cov}}v_{\mathrm{loc}}|\zeta|<1\bigr).
\label{app:stab:analytic-identification}
\end{equation}
For the allowed real coupling, $Z_\beta(u)>0$, so
$\mathcal S_\beta(u)=\mathcal G_{\beta,u}(\tau_{\mathrm{im}})$.
Equation~\eqref{app:stab:geometric-sum} proves
\eqref{app:stab:correlation-decay} with $C_{\mathrm{obs}}=2C_{\mathrm{ext}}$.
If $v_{\mathrm{loc}}=0$, only the free term remains.

For $H_u$ alone, Lemma~\ref{app:stab:covariance} supplies the same
system covariance bound at rate $\Delta/4$. All mode labels are
then system labels, so every surviving edge uses $K_{\mathrm{cov}}$.
The same proof gives the stated decay and the same interaction threshold.
\end{proof}

\subsection{From correlation decay to the full Fock-space gap}

Lemma~\ref{app:stab:decay} controls connected correlations uniformly
in the inverse temperature $\beta$. At fixed volume and model
parameters, this bound survives the limit $\beta\to\infty$.
We now use the following spectral criterion to convert the limiting
imaginary-time decay into the gap bounds of
Theorem~\ref{app:stab:theorem}.

\begin{lem}[Correlation decay and the spectral gap, Lemma 4 and 5 of \cite{DRS}]
\label{app:stab:decay-gap-criterion} Let $H$ be a finite-dimensional
Hermitian Hamiltonian with a unique ground state, and let $\langle\cdot\rangle_{\infty}$
denote its ground-state expectation. If, for some $\gamma>0$ and
every pair of operators $O_{1},O_{2}$, 
\begin{equation}
\left|\langle e^{\tau_{\mathrm{im}}H}O_{1}e^{-\tau_{\mathrm{im}}H}O_{2}\rangle_{\infty}-\langle O_{1}\rangle_{\infty}\langle O_{2}\rangle_{\infty}\right|\le C_{\mathrm{obs}}e^{-\gamma \tau_{\mathrm{im}}},\qquad \tau_{\mathrm{im}}\ge0,
\label{app:stab:ground-correlation-decay}
\end{equation}
where $C_{\mathrm{obs}}$ may depend on $H,O_{1},O_{2}$ but is finite
and independent of $\tau_{\mathrm{im}}$, then $\operatorname{gap}(H)\ge\gamma$.
\end{lem}

The criterion assumes a unique ground state, which remains to be
proved. We therefore first verify the common interaction threshold,
pass to zero temperature, and use the limiting connected correlations
to rule out ground-state degeneracy. The criterion then gives the
gap on the full Fock space.

\begin{proof}[Proof of Theorem~\ref{app:stab:theorem}]
Choose
$c_{\mathrm{loc}}\le
(2C_{\mathrm{deg}}C_{\mathrm{mon}}C_{\mathrm{cov}})^{-1}$.
Combining \eqref{app:model:weak}, \eqref{app:stab:Kcov}, and
\eqref{app:stab:monomial-norm} from
Lemma~\ref{app:stab:monomial-expansion} gives
\begin{equation}
\begin{aligned}
C_{\mathrm{deg}}K_{\mathrm{cov}}|u|v_{\mathrm{loc}}
&\le C_{\mathrm{deg}}C_{\mathrm{mon}}
K_{\mathrm{cov}}J_{\mathrm{int}}\\
&\le C_{\mathrm{deg}}C_{\mathrm{mon}}
C_{\mathrm{cov}}c_{\mathrm{loc}}
\le \frac12.
\end{aligned}
\label{app:stab:threshold-check}
\end{equation}
Thus \eqref{app:stab:tree-smallness} holds with a single interaction
threshold for all amplitudes and volumes.

Fix a finite volume, $A>0$, and an allowed real coupling.
For either Hamiltonian $H\in\{H_{\mathrm{br}}(A),H_u\}$, denote by
$\mathcal G_{\beta,H}^{O_1,O_2}$ the connected correlation
\eqref{app:stab:connected} with Hamiltonian $H$.
Every operator on the finite Fock space is a CAR polynomial.
Lemma~\ref{app:stab:decay} therefore applies to all $O_1,O_2$,
including number- and parity-changing operators, with decay rate
$\gamma_H$ and observable prefactor $C_{O_1,O_2}$:
\begin{equation}
\begin{aligned}
\left|\mathcal G_{\beta,H}^{O_1,O_2}(\tau_{\mathrm{im}})\right|
&\le C_{O_1,O_2}
e^{-\gamma_H d_\beta(\tau_{\mathrm{im}},0)},\\
\gamma_H&=
\begin{cases}
\gamma_A,&H=H_{\mathrm{br}}(A),\\
\Delta/4,&H=H_u.
\end{cases}
\end{aligned}
\label{app:stab:theorem-thermal-decay}
\end{equation}
The prefactor is independent of $\beta$ and $\tau_{\mathrm{im}}$.

\medskip
\noindent\emph{Zero-temperature limit.}
Let $P_0$ be the ground-state projection and
$m_0=\operatorname{rank}P_0$ the ground-state degeneracy.
Define the Gibbs density matrix
$\rho_{\beta,H}\coloneqq e^{-\beta H}/\operatorname{Tr}e^{-\beta H}$
and the trace norm $\|X\|_1\coloneqq\operatorname{Tr}\sqrt{X^\dagger X}$.
For fixed $H$ and every fixed $\tau_{\mathrm{im}}\ge0$,
\begin{equation}
\lim_{\beta\to\infty}
\left\|\rho_{\beta,H}-\frac{P_0}{m_0}\right\|_1=0,
\qquad
\lim_{\beta\to\infty}d_\beta(\tau_{\mathrm{im}},0)
=\tau_{\mathrm{im}}.
\label{app:stab:theorem-zero-temperature}
\end{equation}
For the imaginary-time evolved operator
$O_1(\tau_{\mathrm{im}})
=e^{\tau_{\mathrm{im}}H}O_1e^{-\tau_{\mathrm{im}}H}$,
taking this limit in \eqref{app:stab:theorem-thermal-decay}
gives the limiting connected correlation
$\mathcal G_{\infty,H}^{O_1,O_2}$:
\begin{equation}
\begin{aligned}
\mathcal G_{\infty,H}^{O_1,O_2}(\tau_{\mathrm{im}})
&\coloneqq
\frac{\operatorname{Tr}[P_0O_1(\tau_{\mathrm{im}})O_2]}{m_0}
-\frac{\operatorname{Tr}(P_0O_1)\operatorname{Tr}(P_0O_2)}{m_0^2},\\
\left|\mathcal G_{\infty,H}^{O_1,O_2}(\tau_{\mathrm{im}})\right|
&\le C_{O_1,O_2}e^{-\gamma_H\tau_{\mathrm{im}}}.
\end{aligned}
\label{app:stab:theorem-ground-decay}
\end{equation}

\medskip
\noindent\emph{Uniqueness of the ground state.}
If $m_0>1$, choose a normalized ground state $|\psi\rangle$
and its rank-one projection $Q=|\psi\rangle\langle\psi|$.
Since $Q(\tau_{\mathrm{im}})=Q$ and $Q^2=Q$, we obtain
\begin{equation}
\begin{aligned}
\mathcal G_{\infty,H}^{Q,Q}(\tau_{\mathrm{im}})
&=\frac{\operatorname{Tr}(P_0Q^2)}{m_0}
-\left(\frac{\operatorname{Tr}(P_0Q)}{m_0}\right)^2\\
&=\frac{1}{m_0}-\frac{1}{m_0^2}
=\frac{m_0-1}{m_0^2}>0.
\end{aligned}
\label{app:stab:theorem-degeneracy-correlation}
\end{equation}
However, \eqref{app:stab:theorem-ground-decay} holds for every
$\tau_{\mathrm{im}}\ge0$, with $C_{Q,Q}$ independent of
$\tau_{\mathrm{im}}$ and $\gamma_H>0$. Taking the limit gives
\begin{equation}
0<\frac{m_0-1}{m_0^2}
\le\lim_{\tau_{\mathrm{im}}\to\infty}
\left(C_{Q,Q}e^{-\gamma_H\tau_{\mathrm{im}}}\right)=0,
\label{app:stab:theorem-degeneracy-contradiction}
\end{equation}
which is a contradiction. Thus $m_0=1$.

\medskip
\noindent\emph{Spectral-gap bound.}
The criterion of Lemma~\ref{app:stab:decay-gap-criterion}
can be seen directly as follows.
Let $|0\rangle$ be the unique normalized ground state with energy
$E_0$, and let $|n\rangle$ be a normalized excited eigenstate with
energy $E_n$. Define the excitation energy
$\varepsilon_n=E_n-E_0$ and the transition operators
$O_1=|0\rangle\langle n|$, $O_2=|n\rangle\langle0|$.
Their ground-state expectations vanish, and
\begin{equation}
\begin{aligned}
\mathcal G_{\infty,H}^{O_1,O_2}(\tau_{\mathrm{im}})
&=\langle0|
e^{\tau_{\mathrm{im}}H}O_1
e^{-\tau_{\mathrm{im}}H}O_2
|0\rangle\\
&=e^{-\varepsilon_n\tau_{\mathrm{im}}}.
\end{aligned}
\label{app:stab:theorem-transition-correlation}
\end{equation}
Together with \eqref{app:stab:theorem-ground-decay}, this gives
\begin{equation}
e^{(\gamma_H-\varepsilon_n)\tau_{\mathrm{im}}}
\le C_{O_1,O_2}
\quad\text{for all }\tau_{\mathrm{im}}\ge0,
\qquad\Longrightarrow\qquad
\varepsilon_n\ge\gamma_H.
\label{app:stab:theorem-excitation-bound}
\end{equation}
Taking the smallest excitation energy proves
\eqref{app:stab:gaps}:
\begin{equation}
\operatorname{gap}\!\left(H_{\mathrm{br}}(A)\right)
\ge\gamma_A=c_{\mathrm{gap}}g_A,
\qquad
\operatorname{gap}(H_u)\ge\frac{\Delta}{4}.
\label{app:stab:theorem-gap-conclusion}
\end{equation}
All limits were taken at fixed volume and amplitude; only finiteness
of the observable prefactors was needed.

Finally, let $|\Omega_{\mathrm{br}}(A)\rangle$ be the normalized
bright ground state and $P_{\mathrm{br}}(A)$ its projection.
Restoring the decoupled dark modes gives the protocol ground
projection $\Pi_{\mathrm{br}}(A)$ from
\eqref{app:model:interfaces}:
\begin{equation}
P_{\mathrm{br}}(A)
=|\Omega_{\mathrm{br}}(A)\rangle
\langle\Omega_{\mathrm{br}}(A)|,
\qquad
\Pi_{\mathrm{br}}(A)
=P_{\mathrm{br}}(A)\otimes I_{\mathrm{dark}}.
\label{app:stab:theorem-dark-projection}
\end{equation}
\end{proof}

%% file: appendix/appendix_C_local_dressing.tex
\section{Local dressing and leakage estimates}

\label{app:local:section}

In this section, we establish a framework for controlling the leakage
from a target subspace under adiabatic evolution and under a weak,
time-dependent perturbation of a fixed Hamiltonian. The construction
follows the local adiabatic expansion of Ref.~\cite{bachmann2018adiabatic}
and the perturbative scheme of Ref.~\cite{teufel2020nonequilibrium}.
We first formulate the common structure and then derive the recursion,
keeping explicit track of the expansion order, the reference gap,
and the localization cutoffs. For completeness, we reproduce the derivations
needed from these references. Appendix~\ref{app:stages} then applies
this framework to the three pulse stages.

\subsection{Setting and main estimate}

\label{app:local:setting}

The convention follows Appendix \ref{app:model:setting}. Let $s\in[0,1]$
be normalized time, write $\partial_{s}$ as a prime, and consider
\begin{equation}
iz\partial_{s}\psi(s)=[H(s)+zV(s)]\psi(s),\qquad z>0.\label{app:local:equation}
\end{equation}
The Hamiltonian $H$ and perturbation $V$ are sums of even fermionic
terms with the range bound $R$. We expand the evolution in the small
parameter $z$: in the adiabatic setting, $z\propto1/T$, where $T$
is the total evolution time, whereas in the perturbative setting,
$z$ is the perturbation strength.

Throughout, we assume that $H$ and $V$ satisfy the Gevrey-$2$ bounds:
\begin{equation}
\sup_{s}\|\partial^{j}_{s}H(s)\|_{\mathrm{loc}},\quad\sup_{s}\|\partial^{j}_{s}V(s)\|_{\mathrm{loc}}\le CD^{j}_{\mathrm{sw}}(j!)^{2},\qquad j\ge0,\label{app:local:gevrey}
\end{equation}
with fixed smoothness parameters $C$ and $D_{\mathrm{sw}}$. The
corresponding constructions are given in \ref{app:stages}. Let $\Pi(s)$
be the smooth ground-space projection of $H(s)$, separated from its
orthogonal complement by a gap at least $g\in(0,1]$, uniformly in
$s$ \footnote{Notice that $\Pi(s)$ is the ground space of $H(s)$, but not the
ground space of $H(s)+zV(s)$ in the presence of $V(s)$. The gap
$g$ is also defined only with respect to $H(s)$, not $H(s)+zV(s)$.}. The ground space can be degenerate. Constants may also depend on
the fixed smoothness bounds $D_{\text{sw}}$ above, but their dependence
on $g$ is displayed explicitly. We treat two cases: 
\begin{enumerate}
\item \emph{Adiabatic evolution:} $V(s)=0$ while $H(s)$ and $\Pi(s)$
are varying with $s$.
\item \emph{Perturbation of a fixed Hamiltonian:} $H$ and $\Pi$ are fixed,
while $V(s)$ varies. We assume $H$ is gapped, but no requirement
on $H+zV(s)$. 
\end{enumerate}
A useful special case of the second setting is an onsite reference
$H=D_{0}$, a sum of commuting, even onsite operators with nonnegative
integer spectra and uniformly bounded strengths. We assume its kernel
is nonempty and take $\Pi=\mathbf{1}_{\{0\}}(D_{0})$, so its gap
is at least one. This structure admits an exact support-preserving
inverse. Fix an integer $n\ge3$, an endpoint ($s=0,1)$ is \emph{flat
to order $n+1$} if the derivatives of the varying $H$ or $V$ of
orders $1,\ldots,n+1$ vanish there.
\begin{thm}[Local dressing and leakage]
\label{app:local:theorem} Under the assumptions above, fix an integer
$n\ge3$. There are positive scales $B=B(n,g)$ and $\mathcal{P}=\mathcal{P}(n,g)$,
polynomially bounded in $n,g^{-1}$ and independent of $N,z$, and
even, finite-range Hermitian operators $S_{k}(s)$, $1\le k\le n$,
with the following properties: Define
\begin{equation}
S_{\mathrm{dr}}(z,s)=\sum^{n}_{k=1}z^{k}S_{k}(s),\qquad W_{n}(z,s)=e^{iS_{\mathrm{dr}}(z,s)}\qquad\Pi_{\mathrm{dr}}(s)=W_{n}(z,s)\Pi(s)W_{n}(z,s)^{\dagger},\label{app:local:dressing-unitary}
\end{equation}
and let 
\[
U(s,0)=\mathcal{T}\exp\!\left[-\frac{i}{z}\int^{s}_{0}\bigl(H(\sigma)+zV(\sigma)\bigr)\,d\sigma\right]
\]
be the propagator of Eq.~\eqref{app:local:equation}. For $0<Bz\le\theta$,
with a sufficiently small fixed $\theta\in(0,1)$, 
\begin{align}
\sup_{s}\|S_{\mathrm{dr}}(z,s)\| & \le CNBz,\label{app:local:dressing-size}\\
\|(I-\Pi_{\mathrm{dr}}(1))U(1,0)\Pi_{\mathrm{dr}}(0)\| & \le\frac{C}{z}\left[N(Bz)^{n+1}+N^{2}\mathcal{P}(n,g)(Bz)e^{-cn}\right].\label{app:local:main-leakage}
\end{align}
In particular, choosing $z=\theta/B$ gives the bound 
\begin{equation}
C\left[NB\theta^{n}+N^{2}\mathcal{P}(n,g)Be^{-cn}\right].\label{app:local:chosen-speed}
\end{equation}
For the onsite Hamiltonian $H=D_{0}$ with nonnegative integer spectra,
one may take $B=C(n+2)^{2}$, and the terms containing $\mathcal{P}$
are absent.

The dressing can be chosen with the following endpoint properties.
For a purely adiabatic evolution, meaning $V=0$ in Eq.~\eqref{app:local:equation},
a flat endpoint $s_{*}$ has $W_{n}(z,s_{*})=I$. For the perturbative
case with the fixed $H$, the dressing can be chosen so that, at any
flat endpoint $s_{*}$, replacing $V(s)$ by the constant $V(s_{*})$
leaves $W_{n}(z,s_{*})$ unchanged. In particular, $V(s_{*})=0$ gives
$W_{n}(z,s_{*})=I$.
\end{thm}

The construction below proves the theorem through local inverse maps,
a quantitative dressing lemma, and a propagator estimate. The recursive
dressing and its endpoint properties follow the methods of Refs.~\cite[Lemmas~4.3--4.4]{bachmann2018adiabatic}\cite[Proposition~5.1]{teufel2020nonequilibrium}.
We retain their argument while keeping the expansion order, reference
gap, and finite localization cutoffs explicit. Related quantitative
bounds appear in switching adiabatic estimates~\cite[Lemma~III.1]{EH2012}
and local Schrieffer--Wolff theory~\cite[Lemma~4.2]{BDL2011}, nonetheless,
those results concern different settings.

\subsection{Inverse maps and localization error}

\label{app:local:inverses}

At each order, we remove the part of an operator that couples the
target subspace selected by $\Pi$ to its orthogonal complement. For
the onsite Hamiltnian $H$, this can be done without enlarging supports.
For a general gapped reference we use spectral filtering, then truncate
the filter in time and space. The following integer-spectrum averaging
identity is standard; see, for example, Ref.~\cite[Sec.~5.4]{ADHH2017}.
Its short proof fixes our sign convention. 
\begin{lem}[Exact onsite inverse]
\label{app:local:onsite-inverse} For the fixed onsite Hamiltonian
$H=D_{0}$ with nonnegative integer spectra in Theorem \eqref{app:local:theorem},
let $F$ be an even Hermitian interaction and define 
\begin{align}
\mathcal{E}_{0}(F) & =\frac{1}{2\pi}\int^{2\pi}_{0}e^{itD_{0}}Fe^{-itD_{0}}\,dt,\nonumber \\
\mathcal{I}_{0}(F) & =\frac{1}{2\pi}\int^{2\pi}_{0}(t-\pi)e^{itD_{0}}Fe^{-itD_{0}}\,dt.\label{app:local:integer-maps}
\end{align}
Here $t$ is an auxiliary integration variable. These maps preserve
Hermiticity, evenness, and assigned support. They commute with $s$-derivatives
and satisfy 
\begin{equation}
i[D_{0},\mathcal{I}_{0}(F)]=F-\mathcal{E}_{0}(F),\quad[D_{0},\mathcal{E}_{0}(F)]=0,\quad\|\mathcal{I}_{0}(F)\|_{\mathrm{loc}}\le\frac{\pi}{2}\|F\|_{\mathrm{loc}}.\label{app:local:integer-identities}
\end{equation}
In particular, $\mathcal{I}_{0}(F)=0$ whenever $[D_{0},F]=0$. 
\end{lem}

\begin{proof}
For two eigenvectors of $D_{0}$ with integer energy difference $m\ne0$,
integration by parts gives 
\[
\frac{1}{2\pi}\int^{2\pi}_{0}(t-\pi)e^{imt}\,dt=\frac{1}{im}.
\]
For $m=0$ this integral vanishes, while the other definition $\mathcal{E}_{0}$
equals one. So we have $\mathcal{E}_{0}(F)$ gives the diagonal part
of $F$ with respect to the spectrum of $D_{0}$, while $\mathcal{I}_{0}(F)$
gives the non-diagonal part. Evaluating matrix elements proves the
two commutator identities.

Since $D_{0}$ is on site, $e^{itD_{0}}$ can be factorized into the
direct product of onsite unitaries. Conjugating by $e^{itD_{0}}$
thus remains the same support. The kernel is real and has absolute
integral $(2\pi)^{-1}\int^{2\pi}_{0}|t-\pi|\,dt=\pi/2$, which proves
the norm and Hermiticity statements. Both maps are independent of
$s$. 
\end{proof}

We next turn to the general gapped Hamiltonian case. Lemma~2.6(ii),(iv)
and Corollary~2.8 of Ref.~\cite{BMNS} provides a real odd integrable
filter $w_{\mathrm{inv}}$ satisfying, in our Fourier convention,
\begin{align}
\widehat{w}_{\mathrm{inv}}(\xi) & =\int_{\mathbb{R}}w_{\mathrm{inv}}(t)e^{it\xi}\,dt=-\frac{i}{\xi}\quad(|\xi|\ge1),\nonumber \\
\int_{|t|\ge T}|w_{\mathrm{inv}}(t)|\,dt & \le C\exp\!\left[-\frac{cT}{\log^{2}(e+T)}\right].\label{app:local:filter-input}
\end{align}
Besides this filter, we also need the fermionic Lieb-Robinson bound.
Let $F_{Z}$ and $B_{Y}$ be even observables supported on $Z,Y\subset\Omega$,
and let $H_{\Omega}$ contain the Hamiltonian terms of $H$ supported
inside $\Omega$. For fixed constants $\mu>0$ and $v_{\mathrm{LR}}>0$,
Ref.~\cite[Theorem~3.1]{NSY2018} gives
\begin{equation}
\|[e^{itH_{\Omega}}F_{Z}e^{-itH_{\Omega}},B_{Y}]\|\le C|Z||Y|\|F_{Z}\|\|B_{Y}\|e^{-\mu\operatorname{dist}(Z,Y)+v_{\mathrm{LR}}|t|}.\label{app:local:lr-input}
\end{equation}

We now state the exact inverse and its finite local approximation
as
\begin{align}
\mathcal{I}\Phi & =\int_{\mathbb{R}}w_{\mathrm{inv}}(gt)e^{itH}\Phi e^{-itH}\,dt,\nonumber \\
\mathcal{J}\Phi_{Z} & =\int^{t_{f}}_{-t_{f}}w_{\mathrm{inv}}(gt)e^{itH_{Z^{\ell}}}\Phi_{Z}e^{-itH_{Z^{\ell}}}\,dt,\qquad\mathcal{J}\Phi=\sum_{Z}\mathcal{J}\Phi_{Z}.\label{app:local:spectral-maps}
\end{align}
The region $Z^{\ell}$ is the $\ell$-neighborhood of $Z$, and $H_{Z^{\ell}}$
contains the Hamiltonian terms of $H$ supported inside that region.
$\Phi_{Z}$ is an operator supported inside the region $Z$. Set $L_{n}=\log(e+n)$
and choose 
\begin{equation}
t_{f}=C_{f}nL^{2}_{n}/g,\qquad\ell=\lceil C_{\ell}t_{f}+C_{\ell}'n\rceil.\label{app:local:cutoffs}
\end{equation}
Here $C_{f},C_{\ell},C_{\ell}'$ are sufficiently large fixed constants. 
\begin{lem}[Spectral inverse and local truncation]
\label{app:local:exact-inverse}\label{app:local:spectral-inverse}Let
$\Phi=\sum_{Z}\Phi_{Z}$ be an interaction with even Hermitian terms
and each $\Phi_{Z}$ supports at a local region $Z$. The maps $\mathcal{I}$
and $\mathcal{J}$ in Eq. \eqref{app:local:spectral-maps} preserve
Hermiticity and evenness. The exact map satisfies
\begin{equation}
[\Phi-i[H,\mathcal{I}\Phi],\Pi]=0,\label{app:local:commutator-inverse}
\end{equation}
Equivalently, $\ensuremath{i[H,\mathcal{I}\Phi]}$ extracts the matrix
elements of $\Phi$ between the ground space $\operatorname{Ran}\Pi$
and its orthogonal complement. The approximation $\mathcal{J}\Phi_{Z}$
is supported on $Z^{\ell}$. Let $S_{\max}\ge1$ bound the number
of sites in the support of each local term of $\Phi$ or $H$, as
well as in each enlarged region $Z^{\ell}$. With the cutoffs in Eq.
\eqref{app:local:cutoffs},
\begin{align}
\|(\mathcal{I}-\mathcal{J})\Phi\| & \le CN\mathcal{P}(n,g)e^{-cn}\|\Phi\|_{\mathrm{loc}},\nonumber \\
\mathcal{P}(n,g) & =Cg^{-1}S^{2}_{\max}(1+\ell)^{2D+2}(1+t_{f}).\label{app:local:inverse-defect}
\end{align}
Thus $\mathcal{J}$ approximates the exact inverse while keeping each
output term inside $Z^{\ell}$.
\end{lem}

\begin{proof}
The commutator identity is the standard spectral-inverse identity
\cite[Proposition~4.1(ii)]{bachmann2018adiabatic}, with the sign
fixed by our Fourier convention in Eq. \eqref{app:local:filter-input}.
Hermiticity and evenness follow from the real filter and the even
Hamiltonians in the integrals. The support of $\mathcal{J}\Phi_{Z}$
follows from its restricted Hamiltonian. Here, we prove the truncation
estimate explicitly.

We fix the normalized time $s$ and suppress it in the notation. For
one term $\Phi_{Z}$, put $\Omega=Z^{\ell}$ and write
\[
\tau_{t}(F)=e^{itH}Fe^{-itH},\qquad\tau^{\Omega}_{t}(F)=e^{itH_{\Omega}}Fe^{-itH_{\Omega}}
\]
for the full and restricted evolutions of an operator $F$. Subtracting
the two integrals in Eq. \eqref{app:local:spectral-maps} gives
\begin{equation}
(\mathcal{I}-\mathcal{J})\Phi_{Z}= \int_{|t|>t_{f}}w_{\mathrm{inv}}(gt)\tau_{t}(\Phi_{Z})\,dt
 +\int^{t_{f}}_{-t_{f}}w_{\mathrm{inv}}(gt)[\tau_{t}(\Phi_{Z})-\tau^{\Omega}_{t}(\Phi_{Z})]\,dt.
\label{app:local:inverse-error-split}
\end{equation}
The first term is the time truncation error. The second term is the
spatial truncation error of the evolution for $\Phi_{Z}$, this is
where the Lieb--Robinson bound enters.

Unitary invariance of the norm and the substitution $v=gt$ bound
the time truncation by
\begin{align}
\left\Vert \int_{|t|>t_{f}}w_{\mathrm{inv}}(gt)\tau_{t}(\Phi_{Z})\,dt\right\Vert  & \le\frac{\|\Phi_{Z}\|}{g}\int_{|v|>gt_{f}}|w_{\mathrm{inv}}(v)|\,dv\nonumber \\
 & \le\frac{C}{g}\exp\!\left[-\frac{cgt_{f}}{\log^{2}(e+gt_{f})}\right]\|\Phi_{Z}\|.\label{app:local:time-tail}
\end{align}
Here we used Eq. \eqref{app:local:filter-input}. To see the dependence
on $n$, recall that $gt_{f}=C_{f}nL^{2}_{n}$ and $L_{n}=\log(e+n)\ge1$.
Since $e+C_{f}nL^{2}_{n}\le(1+C_{f})(e+n)^{3}$,
\[
\log(e+gt_{f})\le[3+\log(1+C_{f})]L_{n},\qquad\frac{gt_{f}}{\log^{2}(e+gt_{f})}\ge\frac{C_{f}}{[3+\log(1+C_{f})]^{2}}\,n.
\]
Consequently the right-hand side of Eq. \eqref{app:local:time-tail}
is at most $Cg^{-1}e^{-c_{1}n}\|\Phi_{Z}\|$, with $c_{1}>0$ independent
of $n,g,N$ once $C_{f}$ is fixed.

For the spatial truncation, take $t\ge0$ and interpolate between
the two evolutions:
\[
\frac{d}{du}\tau_{t-u}\bigl(\tau^{\Omega}_{u}(\Phi_{Z})\bigr)=-i\tau_{t-u}\bigl([H-H_{\Omega},\tau^{\Omega}_{u}(\Phi_{Z})]\bigr).
\]
Integrating from $u=0$ to $u=t$ gives the Duhamel identity
\begin{equation}
\tau_{t}(\Phi_{Z})-\tau^{\Omega}_{t}(\Phi_{Z})=i\int^{t}_{0}\tau_{t-u}\bigl([H-H_{\Omega},\tau^{\Omega}_{u}(\Phi_{Z})]\bigr)\,du.\label{app:local:boundary-duhamel}
\end{equation}
Write $H=\sum_{Y}H_{Y}$ in its fixed local decomposition. Since $\tau^{\Omega}_{u}(\Phi_{z})$
supports inside $\Omega$, the terms $H_{Y}$ supported entirely outside
$\Omega$ commute with the even operator $\tau^{\Omega}_{u}(\Phi_{Z})$.
Hence only terms crossing the boundary contribute:
\begin{equation}
[H-H_{\Omega},\tau^{\Omega}_{u}(\Phi_{Z})]=\sum_{\substack{Y\cap\Omega\ne\varnothing\\
Y\cap\Omega^{c}\ne\varnothing
}
}[H_{Y},\tau^{\Omega}_{u}(\Phi_{Z})].\label{app:local:boundary-terms}
\end{equation}

We can then apply Eq. \eqref{app:local:lr-input}. Since every $H_{Y}$
has range at most $R$, a non-vanishing $[H_{Y},\tau^{\Omega}_{u}(\Phi_{Z})]$
satisfies $\operatorname{dist}(Z,Y)\ge\ell-R$. Therefore
\[
\|[H_{Y},\tau^{\Omega}_{u}(\Phi_{Z})]\|\le C|Z||Y|\|H_{Y}\|\|\Phi_{Z}\|e^{-\mu(\ell-R)+v_{\mathrm{LR}}u}.
\]
From the definition of the local norm in Appendix \eqref{app:model:setting},
we further have
\[
\sum_{\substack{Y\cap\Omega\ne\varnothing\\
Y\cap\Omega^{c}\ne\varnothing
}
}\|H_{Y}\|\le\sum_{x\in\Omega}\sum_{Y\ni x}\|H_{Y}\|\le|\Omega|\|H\|_{\mathrm{loc}}.
\]
Substituting these bounds into Eq. \eqref{app:local:boundary-duhamel},
using $\|H\|_{\mathrm{loc}}\le C$, and absorbing $e^{\mu R}$ into
the constant yields
\begin{align}
\|\tau_{t}(\Phi_{Z})-\tau^{\Omega}_{t}(\Phi_{Z})\| & \le C|Z||\Omega|\|\Phi_{Z}\|\int^{t}_{0}e^{-\mu\ell+v_{\mathrm{LR}}u}\,du\nonumber \\
 & \le CS^{2}_{\max}|t|e^{-\mu\ell+v_{\mathrm{LR}}|t|}\|\Phi_{Z}\|.\label{app:local:restricted-evolution-error}
\end{align}
The same bound holds for $t<0$ by reversing the integration direction.

We now insert this estimate into the second term of Eq. \eqref{app:local:inverse-error-split}.
The filter obeys $\int_{\mathbb{R}}|w_{\mathrm{inv}}(gt)|\,dt\le C/g$,
so
\begin{equation}
\begin{split}  \left\Vert \int^{t_{f}}_{-t_{f}}w_{\mathrm{inv}}(gt)[\tau_{t}(\Phi_{Z})-\tau^{\Omega}_{t}(\Phi_{Z})]\,dt\right\Vert
\le\frac{C}{g}S^{2}_{\max}t_{f}e^{-\mu\ell+v_{\mathrm{LR}}t_{f}}\|\Phi_{Z}\|.
\end{split}
\label{app:local:space-tail}
\end{equation}
The choice of $\ell$ in Eq. \eqref{app:local:cutoffs} implies
\[
-\mu\ell+v_{\mathrm{LR}}t_{f}\le-(\mu C_{\ell}-v_{\mathrm{LR}})t_{f}-\mu C_{\ell}'n\le-c_{2}n,
\]
provided $\mu C_{\ell}>v_{\mathrm{LR}}$ and $C_{\ell}'>0$, with
$c_{2}=\mu C_{\ell}'$. Finally, sum over $Z$ using $\sum_{Z}\|\Phi_{Z}\|\le CN\|\Phi\|_{\mathrm{loc}}$.
The two truncation errors give
\[
\|(\mathcal{I}-\mathcal{J})\Phi\|\le\frac{CN}{g}\bigl(e^{-c_{1}n}+S^{2}_{\max}t_{f}e^{-c_{2}n}\bigr)\|\Phi\|_{\mathrm{loc}}.
\]
This implies Eq. \eqref{app:local:inverse-defect} with $c=\min\{c_{1},c_{2}\}$;
the displayed $\mathcal{P}$ is a common upper bound for the prefactors.
\end{proof}

In the adiabatic case $V=0$, the target ground space $\operatorname{Ran}\Pi(s)$
changes with $s$. To compare the physical evolution with one that
follows this space exactly, we introduce an auxiliary Hermitian generator
$K(s)$~\cite[Corollary~4.2]{bachmann2018adiabatic}. A propagator
generated by $H+zK$ transports $\Pi(s)$ exactly only if
\[
\Pi'=-\frac{i}{z}[H+zK,\Pi]=-i[K,\Pi],
\]
where the second equality uses $[H,\Pi]=0$. The next lemma constructs
$K$ from the spectral inverse.
\begin{lem}[Generator of adiabatic ground-space transport]
\label{app:local:projection-transport}In the adiabatic case, the
Hermitian operator
\begin{equation}
K(s)=\mathcal{I}H'(s)\qquad\text{satisfies}\qquad\Pi'(s)=-i[K(s),\Pi(s)].\label{app:local:exact-transport}
\end{equation}
Let $U_{\mathrm{ad}}(s,0)$ solve
\[
iz\partial_{s}U_{\mathrm{ad}}(s,0)=[H(s)+zK(s)]U_{\mathrm{ad}}(s,0),\qquad U_{\mathrm{ad}}(0,0)=I.
\]
Then it transports the entire ground space exactly~\cite[Proposition~2.4 and Corollary~2.8]{BMNS}:
\begin{equation}
U_{\mathrm{ad}}(s,0)\Pi(0)U_{\mathrm{ad}}(s,0)^{\dagger}=\Pi(s).\label{app:local:projection-intertwining}
\end{equation}
The finite local approximation to $K$ is
\begin{equation}
K_{\mathrm{tr}}=\mathcal{J}H',\qquad\|K-K_{\mathrm{tr}}\|\le CN\mathcal{P}(n,g)e^{-cn}\|H'\|_{\mathrm{loc}}.\label{app:local:transport-sources}
\end{equation}
\end{lem}

\begin{proof}
Equation \eqref{app:local:exact-transport} is the spectral-flow identity
of Ref. \cite[Corollary~4.2]{bachmann2018adiabatic}, in our Fourier
convention. We include the short argument to explain how $H'$ determines
the motion of $\Pi$. Differentiating $[H,\Pi]=0$ gives
\[
[H,\Pi']=-[H',\Pi].
\]
Applying Eq. \eqref{app:local:commutator-inverse} to $\Phi=H'$ and
using $[H,\Pi]=0$ yields
\[
[H',\Pi]=i[[H,K],\Pi]=i[H,[K,\Pi]].
\]
Thus $[H,\Pi'+i[K,\Pi]]=0$. Differentiating $\Pi^{2}=\Pi$ shows
that $\Pi'$ has only off-diagonal blocks with respect to $\Pi$;
the same is true of $[K,\Pi]$. The gap between $\operatorname{Ran}\Pi$
and its complement then implies $\Pi'+i[K,\Pi]=0$.

For the propagator, direct differentiation gives
\[
\frac{d}{ds}\bigl(U^{\dagger}_{\mathrm{ad}}\Pi U_{\mathrm{ad}}\bigr)=U^{\dagger}_{\mathrm{ad}}\left(\Pi'+\frac{i}{z}[H+zK,\Pi]\right)U_{\mathrm{ad}}=0,
\]
which proves Eq. \eqref{app:local:projection-intertwining}. The bound
on $K-K_{\mathrm{tr}}$ directly from Eq. \eqref{app:local:inverse-defect}
applied to $H'$.
\end{proof}

The generator $K$ is needed only in the adiabatic case, where $H$
varies. For a fixed Hamiltonian in the perturbative case, $H'=0$
and we set $K=K_{\mathrm{tr}}=0$ throughout. Both cases can therefore
be written as
\begin{equation}
H+zV=(H+zK)+z(V-K).\label{app:local:source-split}
\end{equation}
The first term generates an evolution that follows $\Pi(s)$ exactly.
The dressing must compensate the matrix elements of $V-K$ between
$\operatorname{Ran}\Pi$ and its complement. This source is $-K$
for adiabatic evolution and $V$ for a fixed Hamiltonian in the perturbative
case. 

\subsection{Construction of the local dressing}

\label{app:local:finite-order}

We construct a unitary change of variables that removes transitions
between $\operatorname{Ran}\Pi(s)$ and its orthogonal complement,
order by order in $z$. The construction follows Refs.~\cite[Lemma~4.3]{bachmann2018adiabatic}
\cite[Sec.~6.1]{teufel2020nonequilibrium}, with the finite inverse
$\mathcal{J}$ in place of the exact inverse. In the onsite case,
$\mathcal{I}=\mathcal{J}=\mathcal{I}_{0}$. The gap parameter, cutoffs,
and assigned supports are held fixed when taking $s$-derivatives.
Each commutator is assigned to the union of its input supports, and
the bath modes attached to a system site count as part of that site.

First we derive the equation that the change of variables must satisfy.
For a Hermitian generator $S_{\text{dr}}(z,s)$, write the physical
state as $\psi(s)=e^{iS_{\text{dr}}(z,s)}\chi(s)$. $\chi(s)$ then
satisfies the equation:
\begin{equation}
iz\chi'=\widetilde{H}[S_{\text{dr}}]\chi,\qquad\widetilde{H}[S_{\text{dr}}]=e^{-iS_{\text{dr}}}(H+zV)e^{iS_{\text{dr}}}-ize^{-iS_{\text{dr}}}\partial_{s}e^{iS_{\text{dr}}}.\label{app:local:changed-equation}
\end{equation}
The motivation for this dressing is that approximate transport of
$\text{Ran}\Pi(s)$ under $\widetilde{H}[S_{\text{dr}}]$ corresponds
to approximate transport of the dressed subspace $\text{Ran}\Pi_{\text{dr}}(s)$
by the physical evolution. Here $\Pi_{\text{dr}}(s)=e^{iS_{\text{dr}}}\Pi(s)e^{-iS_{\text{dr}}}$
is defined in Theorem \eqref{app:local:theorem}.

The decomposition in Eq.~\eqref{app:local:source-split} separates
the generator $H+zK$, which transports $\Pi(s)$ exactly, from the
source $z(V-K)$ that the dressing must remove between the two subspaces.
Adding any Hermitian operator commuting with $\Pi(s)$ leaves that
transport equation unchanged. We therefore seek $S_{\mathrm{dr}}$
for which $\widetilde{H}[S_{\text{dr}}]-(H+zK)$ commutes with $\Pi(s)$,
up to a small remainder. To keep the coefficients local, we use $K_{\mathrm{tr}}$
in the recursion and retain $z(K_{\mathrm{tr}}-K)$ as an error. For
fixed $H$ in the perturbative case, both transport generators vanish.

To expand Eq.~\eqref{app:local:changed-equation}, define $\operatorname{ad}_{S}F=[S,F]$.
The BCH formula directly gives: 
\begin{align}
e^{-iS}Fe^{iS} & =\sum_{m\ge0}\frac{(-i)^{m}}{m!}\operatorname{ad}^{m}_{S}F,\nonumber \\
-ize^{-iS}\partial_{s}e^{iS} & =z\sum_{m\ge0}\frac{(-i)^{m}}{(m+1)!}\operatorname{ad}^{m}_{S}S'.\label{app:local:bch-series}
\end{align}
For the recursion we only extract finitely many coefficients of these
series. Bounds for the full, untruncated unitary are proved separately
after the coefficient estimates.

Suppose $S_{1},\ldots,S_{k-1}$ have already been chosen, and set
$S_{<k}(z,s)=\sum^{k-1}_{j=1}z^{j}S_{j}(s)$, with $S_{<1}=0$. Let
$[z^{k}]$ denote extraction of the coefficient of $z^{k}$. The part
of the order-$k$ coefficient known before choosing $S_{k}$ is 
\begin{equation}
F_{k}=[z^{k}]\bigl(\widetilde{H}[S_{<k}]-H-zK_{\mathrm{tr}}\bigr).\label{app:local:recursion-forcing}
\end{equation}
This quantity depends only on $S_{1},\cdots,S_{k-1}$ but not on $S_{k}$.
To see how it determines $S_{k}$, insert $S_{<k}+z^{k}S_{k}$ into
Eq.~\eqref{app:local:bch-series}. The new commutator with $H$ contributes
$iz^{k}[H,S_{k}]$. Every other new term has at least one additional
power of $z$. Consequently, 
\begin{equation}
[z^{k}]\bigl(\widetilde{H}[S_{<k}+z^{k}S_{k}]-H-zK_{\mathrm{tr}}\bigr)=F_{k}+i[H,S_{k}].\label{app:local:order-cancellation}
\end{equation}
The spectral-inverse identity \eqref{app:local:commutator-inverse}
would make this coefficient commute with $\Pi$ if we chose $S_{k}=-\mathcal{I}F_{k}$.
We instead use the local inverse and define 
\begin{equation}
S_{k}=-\mathcal{J}F_{k},\qquad D_{k}=F_{k}-i[H,\mathcal{I}F_{k}],\label{app:local:recursion}
\end{equation}
which establishes the recursive relation for $S_{k}$. It follows
that $[D_{k},\Pi]=0$ and that the actual retained coefficient is
\begin{equation}
F_{k}+i[H,S_{k}]=D_{k}+i[H,(\mathcal{I}-\mathcal{J})F_{k}].\label{app:local:retained-defect}
\end{equation}
The last term is precisely the price of using a finite-range inverse.
For $\mathcal{I}_{0}$ it vanishes, and $D_{k}=\mathcal{E}_{0}(F_{k})$. 

For clarity, at first order $S_{<1}=0$, so $F_{1}=V-K_{\mathrm{tr}}$.
At second order, the double commutator with $H$, the single commutator
with $zV$, and the leading derivative term give, respectively, 
\[
F_{2}=-\frac{1}{2}[S_{1},[S_{1},H]]+i[V,S_{1}]+S_{1}'.
\]
In the adiabatic case $V=0$, this yields $S_{1}=\mathcal{J}K_{\mathrm{tr}}=\mathcal{J}^{2}H'$.
The two inverse applications in this first coefficient must therefore
be included when counting its support.
\begin{lem}[Dressing with a controlled remainder]
\label{app:local:normal-form} The recursion \eqref{app:local:recursion-forcing}--\eqref{app:local:recursion}
produces even, finite-range Hermitian $S_{k}$ and Hermitian $D_{k}$
with $[D_{k},\Pi]=0$. Define $S_{\mathrm{dr}}$ and $W_{n}$ by Eq.~\eqref{app:local:dressing-unitary},
and set 
\begin{equation}
D^{(n)}(z,s)=\sum^{n}_{k=1}z^{k}D_{k}(s).\label{app:local:comparison-polynomial}
\end{equation}
There are scales $B=B(n,g)$ and $\mathcal{P}=\mathcal{P}(n,g)$ that
grow at most polynomially in $n,g^{-1}$ and are independent of $N,z$,
such that, for $0<Bz\le\theta$, 
\begin{equation}
W^{\dagger}_{n}(H+zV)W_{n}-izW^{\dagger}_{n}W_{n}'=H+zK+D^{(n)}+R_{n},\label{app:local:exact-generator}
\end{equation}
where the Hermitian remainder obeys 
\begin{align}
\sup_{s}\|R_{n}(z,s)\| & \le CN(Bz)^{n+1}+CN^{2}\mathcal{P}(n,g)(Bz)e^{-cn},\nonumber \\
\sup_{s}\|S_{\mathrm{dr}}(z,s)\| & \le CNBz.\label{app:local:remainder}
\end{align}
For $H=D_{0}$, $B=C(n+2)^{2}$ suffices, the second remainder term
is absent, and $[D_{k},D_{0}]=0$. 
\end{lem}

The proof has three parts: bounds for the local operations, bounds
for the recursively defined coefficients, and control of the full
unitary remainder. We use the following three sublemmas to complete
these parts. The operators $D_{k}$ need not have finite range; they
serve only to define a Hermitian comparison evolution that preserves
the reference projection.

\subsubsection{Supports and derivatives}

The recursion uses only three operations: the finite inverse, a commutator,
and one $s$-derivative. We first bound their costs uniformly through
order $n$. This requires keeping both the support size and the number
of remaining derivatives.

A nonzero commutator joins overlapping supports; differentiation does
not enlarge them, and $\mathcal{J}$ adds an $\ell$-neighborhood.
Each coefficient of total order $k$ is therefore supported in a connected
region of diameter at most $Ck(R+\ell+1)$. In particular, choose
\begin{equation}
S_{\max}=C[1+n(\ell+R)]^{D}\label{app:local:support-bound}
\end{equation}
as a common bound on the number of sites in every retained support
and inverse region. The induction establishing this bound is given
in the proof below. For the onsite inverse, no neighborhood is added
and we can use the sharper bound $S_{\max}=s_{0}(n+1)$, where $s_{0}$
bounds the size of a bare perturbation support.

The derivative term in Eq.~\eqref{app:local:bch-series} carries
one power of $z$. Thus one derivative consumes one order of the expansion.
Assign order zero to $H$ and order one to $V,K_{\mathrm{tr}}$; at
order $k$ we retain derivatives up to $n+1-k$. Define 
\begin{equation}
\|\Phi\|_{k,*}=\max_{0\le j\le n+1-k}\frac{\sup_{s}\|\partial^{j}_{s}\Phi(s)\|_{\mathrm{loc}}}{D^{j}_{\mathrm{jet}}(j!)^{2}},\qquad0\le k\le n+1.\label{app:local:jet-norm}
\end{equation}
The factorial weights match the Gevrey bounds. The derivative scale
$D_{\mathrm{jet}}$ is chosen at least a fixed multiple of $D_{\mathrm{sw}}$.
For varying $H$ in the adiabatic case, we take $D_{\mathrm{jet}}=CS_{\max}(1+t_{f})$,
to control derivatives of the restricted propagators inside $\mathcal{J}$.
For fixed $H$ in the perturbative case, those propagators have no
$s$-dependence, so a fixed $D_{\mathrm{jet}}$ suffices. The following
bounds justify both choices.
\begin{sublem}[Bounds for the retained local operations]
\label{app:local:operation-lemma} Every support and inverse region
retained through order $n$ has at most $S_{\max}$ sites. For retained
interactions of orders $a,b$ with $a+b\le n+1$, the derivative norms
\eqref{app:local:jet-norm} satisfy 
\begin{align}
\|\mathcal{J}\Phi\|_{k,*} & \le G\|\Phi\|_{k,*},\nonumber \\
\|[\Phi,\Psi]\|_{a+b,*} & \le\mathfrak{L}\|\Phi\|_{a,*}\|\Psi\|_{b,*},\nonumber \\
\|\Phi'\|_{k+1,*} & \le\mathfrak{F}\|\Phi\|_{k,*},\label{app:local:operations}
\end{align}
where 
\begin{equation}
G=C(1+\ell)^{D}/g,\qquad\mathfrak{L}=CS_{\max},\qquad\mathfrak{F}=CD_{\mathrm{jet}}(n+2)^{2}.\label{app:local:cost-definitions}
\end{equation}
For the onsite inverse they can be sharpened to $G=O(1)$, $\mathfrak{L}=O(n)$,
and $\mathfrak{F}=O(n^{2})$. For the general finite inverse, by substituting
the expression of $\ell$, they satisfy 
\begin{align}
G & \le Cn^{D}L^{2D}_{n}g^{-(D+1)}, & \mathfrak{L} & \le Cn^{2D}L^{2D}_{n}g^{-D},\nonumber \\
\mathfrak{F}_{\mathrm{ad}} & \le Cn^{2D+3}L^{2D+2}_{n}g^{-(D+1)}, & \mathfrak{F}_{\mathrm{pert}} & \le Cn^{2},\label{app:local:explicit-costs}
\end{align}
where the two subscripts distinguish the adiabatic case and the perturbative
case. The inverse-defect prefactor in Eq.~\eqref{app:local:inverse-defect}
then obeys 
\begin{equation}
\mathcal{P}(n,g)\le Cn^{6D+3}L^{8D+6}_{n}g^{-(4D+4)}.\label{app:local:defect-polynomial}
\end{equation}
\end{sublem}

\begin{proof}
We first check the supports. Choose $L_{0}=C(R+\ell+1)$ large enough
that $S_{1}=-\mathcal{J}(V-\mathcal{J}H')$ has diameter at most $L_{0}$.
Suppose the diameter of $S_{j}$ is at most $(2j-1)L_{0}$ for $j<k$.
A term in $F_{k}$ coming from the nested commutator with $H$ contains
at least two lower generators whose orders sum to $k$. Their supports
must overlap along the nested commutators. Adding the support of $H$
and the final inverse region therefore gives diameter at most $(2k-1)L_{0}$,
after fixing $L_{0}$ large enough. Terms from $zV$ have generator
orders summing to $k-1$; terms from the derivative series have total
generator order $k-1$. The same bound follows in these two cases.
This proves the induction, and Eq.~\eqref{app:local:support-bound}
follows. For $H=D_{0}$, commutation with the onsite reference and
application of $\mathcal{I}_{0}$ add no sites. Each order-$k$ local
term is therefore supported on at most $ks_{0}$ sites.

Next consider even interactions $\Phi=\sum_{Z}\Phi_{Z}$ and $\Psi=\sum_{Y}\Psi_{Y}$,
with support sizes at most $s_{\Phi},s_{\Psi}$. From the definition,
it is easy to verify that
\begin{equation}
\|[\Phi,\Psi]\|_{\mathrm{loc}}\le2(s_{\Phi}+s_{\Psi})\|\Phi\|_{\mathrm{loc}}\|\Psi\|_{\mathrm{loc}}.\label{app:local:commutator-bound}
\end{equation}
For $j$ derivatives of the product of two terms, the Leibniz coefficient
and the factorial weights give $\binom{j}{r}(r!)^{2}((j-r)!)^{2}/(j!)^{2}=\binom{j}{r}^{-1}$.
Since $\sum^{j}_{r=0}\binom{j}{r}^{-1}\le3$, the commutator bound
in Eq.~\eqref{app:local:operations} follows with $\mathfrak{L}=CS_{\max}$. 

For one derivative, 
\[
\frac{\|\partial^{j+1}_{s}\Phi\|_{\mathrm{loc}}}{D^{j}_{\mathrm{jet}}(j!)^{2}}\le D_{\mathrm{jet}}(j+1)^{2}\|\Phi\|_{k,*}.
\]
Taking the maximum proves 
\begin{equation}
\|\Phi'\|_{k+1,*}\le D_{\mathrm{jet}}(n+2)^{2}\|\Phi\|_{k,*}.\label{app:local:derivative-cost}
\end{equation}
In particular, the last generator $S_{n}$ still has the one derivative
needed in the transformed Hamiltonian.

It remains to bound $\mathcal{J}$. For a retained region $\Omega$,
write $U_{\Omega}(s,t)=e^{itH_{\Omega}(s)}$. Its first derivative
is 
\[
\partial_{s}U_{\Omega}(s,t)=i\int^{t}_{0}e^{i(t-v)H_{\Omega}(s)}H_{\Omega}'(s)e^{ivH_{\Omega}(s)}\,dv.
\]
Repeated differentiation produces terms with $m$ insertions $H^{(\alpha_{1})}_{\Omega},\ldots,H^{(\alpha_{m})}_{\Omega}$,
where $\alpha_{r}\ge1$ and $\sum_{r}\alpha_{r}=j$. Their integration
simplex has volume $|t|^{m}/m!$, and the derivative coefficient is
$j!/\prod_{r}\alpha_{r}!$. All intervening propagators are unitary.
Using $\|H^{(\alpha)}_{\Omega}\|\le CS_{\max}D^{\alpha}_{\mathrm{sw}}(\alpha!)^{2}$
and $\prod_{r}\alpha_{r}!\le j!$, the sum is bounded by 
\[
D^{j}_{\mathrm{sw}}(j!)^{2}\sum^{j}_{m=1}\binom{j-1}{m-1}\frac{(CS_{\max}|t|)^{m}}{m!},
\]
where the combinatorial factor counts the $m$-tuples of positive
integers $(\alpha_{1},\cdots,\alpha_{m})$ satisfying $\sum^{m}_{r=1}\alpha_{r}=j$
for fixed $m$. The binomial theorem, after dropping $1/m!\le1$,
gives 
\begin{equation}
\|\partial^{j}_{s}U_{\Omega}(s,t)\|\le[C(1+S_{\max}|t|)]^{j}(j!)^{2}.\label{app:local:unitary-jets}
\end{equation}
Hence, for $|t|\le t_{f}$, the stated choice of $D_{\mathrm{jet}}$
makes the normalized derivatives of $U_{\Omega}$ and $U^{\dagger}_{\Omega}$
at most $q^{j}$, for a fixed $0<q<1$. In applying Leibniz's rule
to $U_{\Omega}\Phi_{Z}U^{\dagger}_{\Omega}$, let $a,b,c\ge0$, with
$a+b+c=j$, be the derivative orders on the three factors, respectively.
After normalization by $D^{j}_{\mathrm{jet}}(j!)^{2}$, the powers
of $D_{\mathrm{jet}}$ cancel, while the Leibniz coefficient and the
Gevrey factorial weights combine to give

\[
\frac{j!}{a!b!c!}\frac{(a!)^{2}(b!)^{2}(c!)^{2}}{(j!)^{2}}=\frac{a!b!c!}{j!}\le1.
\]
Since $b=j-a-c$, the remaining sum over $a,c$ on the two propagators
is therefore bounded by $\sum_{a,c\ge0}q^{a+c}=(1-q)^{-2}$.

The filter contributes $\int^{t_{f}}_{-t_{f}}|w_{\mathrm{inv}}(gt)|\,dt\le C/g$.
For the support enlargement where each local term $\Phi_{Z}$ is extended
to the region $Z^{\ell}$ by $\mathcal{J}$, at each site $x$ one
has 
\[
\sum_{Z:x\in Z^{\ell}}\|\Phi_{Z}\|\le\sum_{y:\operatorname{dist}(x,y)\le\ell}\sum_{Z\ni y}\|\Phi_{Z}\|\le C(1+\ell)^{D}\|\Phi\|_{\mathrm{loc}}.
\]
These two factors give $G=C(1+\ell)^{D}/g$. If $H$ is fixed in the
perturbative case, $\mathcal{J}$ commutes with $s$-derivatives.
The same proof then needs only a fixed $D_{\mathrm{jet}}$. For $\mathcal{I}_{0}$,
Eq.~\eqref{app:local:integer-identities} gives $G=O(1)$ directly.

Finally, Eq.~\eqref{app:local:cutoffs}, with $n\ge3$ and $g\le1$,
gives the intermediate bounds 
\[
1+\ell,\ 1+t_{f}\le CnL^{2}_{n}/g,\qquad S_{\max}\le Cn^{2D}L^{2D}_{n}g^{-D}.
\]
Substitution into $G$, $\mathfrak{L}$, and $\mathfrak{F}=CD_{\mathrm{jet}}(n+2)^{2}$
gives Eq.~\eqref{app:local:explicit-costs}. Likewise, 
\[
g^{-1}S^{2}_{\max}(1+\ell)^{2D+2}(1+t_{f})\le Cg^{-1}(n^{2D}L^{2D}_{n}g^{-D})^{2}(nL^{2}_{n}/g)^{2D+3},
\]
which is exactly Eq.~\eqref{app:local:defect-polynomial}. 
\end{proof}

\subsubsection{Growth of the coefficients}

Let $h_{\mathrm{jet}}=\|H\|_{0,*}\le h_{0}$, with fixed $h_{0}\ge1$,
and choose a source bound 
\begin{equation}
\nu_{*}\ge\max\{1,\|V\|_{1,*}+\|K_{\mathrm{tr}}\|_{1,*}\}.\label{app:local:source-size}
\end{equation}
The inverse estimate and the Gevrey bounds allow $\nu_{*}\le C(1+G)$
in the adiabatic case. For a fixed Hmailtonian in the perturbative
case, $\nu_{*}=O(1)$. 
\begin{sublem}[Coefficient bounds]
\label{app:local:coefficient-lemma} The recursion \eqref{app:local:recursion-forcing}--\eqref{app:local:recursion}
obeys 
\begin{equation}
\|S_{k}\|_{k,*}\le r_{x}(B/8)^{k},\qquad\|F_{k}\|_{k,*}\le(r_{x}/G)(B/8)^{k},\label{app:local:coefficient-bounds}
\end{equation}
where 
\begin{equation}
B=C_{B}(G^{2}\mathfrak{L}^{2}\nu_{*}+G\mathfrak{F}),\qquad r_{x}=\frac{c_{x}}{G\mathfrak{L}^{2}}.\label{app:local:scale}
\end{equation}
Here $C_{B},c_{x}>0$ are fixed constants.
\end{sublem}

\begin{proof}
This is a scalar-majorant estimate of the type used in local perturbation
theory~\cite[Lemma~4.2]{BDL2011}; the calculation here includes
the time-derivative term and the cost of $\mathcal{J}$.

Introduce a scalar variable $w$ to count powers of $z$, and a formal
series $x(w)=\sum_{k\ge1}x_{k}w^{k}$ with $x_{k}\ge0$ that bounds
the generator coefficients. Here, $w$ will be continued to the complex
space and agree with $z$ on the real line. The three parts of Eq.~\eqref{app:local:bch-series}
give the following scalar bounds. The terms with $H$ and at least
two commutators contribute $h_{\mathrm{jet}}\sum_{m\ge2}(\mathfrak{L}x)^{m}/m!$.
The $zV$ terms contribute at most $w\|V\|_{1,*}e^{\mathfrak{L}x}$;
the subtraction of $zK_{\mathrm{tr}}$ is bounded by adding $w\|K_{\mathrm{tr}}\|_{1,*}e^{\mathfrak{L}x}$.
The derivative series contributes at most $w\mathfrak{F}xe^{\mathfrak{L}x}$,
since replacing $(m+1)!$ by $m!$ only increases its bound. Multiplication
by $G$ accounts for $S_{k}=-\mathcal{J}F_{k}$. We therefore define
$x$ recursively by 
\begin{equation}
x=Gh_{\mathrm{jet}}(e^{\mathfrak{L}x}-1-\mathfrak{L}x)+Gw(\nu_{*}+\mathfrak{F}x)e^{\mathfrak{L}x}.\label{app:local:scalar-majorant}
\end{equation}
The coefficient of $w^{k}$ on the right depends only on $x_{1},\ldots,x_{k-1}$:
the first term is quadratic in $x$, and the others contain an extra
$w$. This defines a unique nonnegative formal series. Induction using
Eq.~\eqref{app:local:operations} now gives $\|F_{k}\|_{k,*}\le x_{k}/G$
and $\|S_{k}\|_{k,*}\le x_{k}$.

We next establish a radius of convergence independent of $N$. Let
$T_{w}(x)$ denote the right-hand side of Eq. (\eqref{app:local:scalar-majorant}),
and consider the disks

\[
|x|\le r_{x}=\frac{c_{x}}{G\mathfrak{L}^{2}},\qquad|w|\le r_{w}=\frac{8}{B}.
\]
Choosing $c_{x}$ sufficiently small and then $C_{B}$ sufficiently
large, direct estimates give

\[
|T_{w}(x)|\le\frac{r_{x}}{2},\qquad|\partial_{x}T_{w}(x)|\le\frac{1}{2}
\]
uniformly on these disks. Thus, for every $|w|\le r_{w}$, $T_{w}$
maps the$x$-disk into itself and is a contraction with a uniform
contraction constant.

The iterates starting from $x=0$ are analytic in $w$ and converge
uniformly on $|w|\le r_{w}$. Their limit therefore defines an analytic
solution $x(w)$ on $|w|<r_{w}$ satisfying $|x(w)|\le r_{x}$. Its
Taylor coefficients coincide with the formal coefficients $x_{k}$
defined above. Cauchy's coefficient estimate consequently gives

\[
x_{k}\le r_{x}r^{-k}_{w}=r_{x}(B/8)^{k}.
\]
Using the coefficientwise bounds established above, we obtain

\begin{equation}
\|S_{k}\|_{k,*}\le x_{k}\le r_{x}(B/8)^{k},\qquad\|F_{k}\|_{k,*}\le\frac{x_{k}}{G}\le\frac{r_{x}}{G}(B/8)^{k},\qquad1\le k\le n,
\end{equation}
which proves Eq. (\eqref{app:local:coefficient-bounds}).
\end{proof}

\subsubsection{Remainder of the full unitary}
\begin{sublem}[Algebraic and localization remainders]
\label{app:local:remainder-lemma} For the dressing defined by the
recursion and for $0<Bz\le\theta$, the remainder in Eq.~\eqref{app:local:exact-generator}
has the form 
\begin{equation}
R_{n}=R_{\mathrm{alg}}+\sum^{n}_{k=1}z^{k}i[H,(\mathcal{I}-\mathcal{J})F_{k}]+z(K_{\mathrm{tr}}-K),\label{app:local:defect-identity}
\end{equation}
where $R_{\mathrm{alg}}$ is the part of $\widetilde{H}[S_{\mathrm{dr}}]$
of total degree greater than $n$ in $z$, and 
\begin{equation}
\|R_{\mathrm{alg}}\|\le CN(Bz)^{n+1}.\label{app:local:algebraic-remainder}
\end{equation}
The sum of the remaining terms is at most $CN^{2}\mathcal{P}(n,g)(Bz)e^{-cn}$
and vanishes for the onsite inverse. 
\end{sublem}

\begin{proof}
We first bound $R_{\mathrm{alg}}$. Fix $s$ and suppress it in the
notation. For a function of $z$, the symbol $[\cdot]_{>q}$ which
is defined above denotes its Taylor remainder after retaining all
powers through $z^{q}$. The three terms in the transformed generator
give 
\begin{align}
R_{\mathrm{alg}}={} & \bigl[e^{-iS_{\mathrm{dr}}}He^{iS_{\mathrm{dr}}}\bigr]_{>n}+z\bigl[e^{-iS_{\mathrm{dr}}}Ve^{iS_{\mathrm{dr}}}\bigr]_{>n-1}\nonumber \\
 & +z\left[\int^{1}_{0}e^{-ivS_{\mathrm{dr}}}S_{\mathrm{dr}}'e^{ivS_{\mathrm{dr}}}\,dv\right]_{>n-1}.\label{app:local:algebraic-split}
\end{align}
They arise from the reference Hamiltonian, the perturbation, and the
time derivative of the dressing. We remark that $S_{\text{dr}}$ as
well as its derivatives is a function of $z$.

To bound these three contributions, Eq.~\eqref{app:local:coefficient-bounds}
gives 
\begin{equation}
\sum^{n}_{k=1}w^{k}\|S_{k}\|_{k,*}\le X(w):=r_{x}\frac{Bw/8}{1-Bw/8},\qquad0\le w<8/B.\label{app:local:generator-majorant}
\end{equation}
Substituting $w$ by $z$ along the real line, we obtain 
\begin{equation}
\lVert S_{\text{dr}}\rVert_{\text{loc}}\leq r_{x}\frac{Bz/8}{1-Bz/8},\qquad0\leq z<8/B.
\end{equation}
The derivative bound also gives $\|S_{\text{dr}}'\|_{\mathrm{loc}}\le\mathfrak{F}X(z)$.
For $F=H,V,S_{\mathrm{dr}}'$, all assigned supports have size at
most $S_{\max}$. At the $j$th commutator, the two input supports
have sizes at most $S_{\max}$ and $jS_{\max}$. Iterating Eq.~\eqref{app:local:commutator-bound}
therefore yields 
\begin{equation}
\|\operatorname{ad}^{m}_{S_{\mathrm{dr}}}F\|_{\mathrm{loc}}\le\|F\|_{\mathrm{loc}}\bigl(C_{\mathrm{ad}}S_{\max}\|S_{\mathrm{dr}}\|_{\mathrm{loc}}\bigr)^{m}(m+1)!.\label{app:local:all-depth-bound}
\end{equation}
The factor $(m+1)!$, coming from the growth of supports, can be canceled
by BCH denominator $m!$, leaving the factor $m+1$, whose sum is
$\sum_{m\ge0}(m+1)y^{m}=(1-y)^{-2}$ for $y=C_{\mathrm{ad}}S_{\max}\|S_{\mathrm{dr}}\|_{\mathrm{loc}}<1$.
It is easy to verify that the same upper bound applies to the derivative
series $e^{iS_{\text{dr}}}\partial_{s}e^{-iS_{\text{dr}}}$. The absolute
convergence to justify the BCH and derivative expansions is parallel
to the real-unitary Taylor argument of Ref.~\cite[Lemma~4.3]{bachmann2018adiabatic}.

To collect the coefficient bounds, define 
\begin{equation}
T(z)=\frac{C\{\|H\|_{\mathrm{loc}}+z\|V\|_{1,*}+z\mathfrak{F}X(z)\}}{[1-C_{\mathrm{ad}}S_{\max}X(z)]^{2}}=\sum_{k\ge0}t_{k}z^{k},\label{app:local:tail-majorant}
\end{equation}
where the denominator comes from the bound on summation over BCH series
in Eq. (\eqref{app:local:all-depth-bound}). All coefficients $t_{k}$
are nonnegative. Expanding each nested commutator and applying the
triangle inequality shows that $t_{k}$ bounds the sum of the local-norm
estimates for terms of total degree $k$ in the transformed generator.
The three numerator terms correspond to the three contributions in
Eq.~\eqref{app:local:algebraic-split}.

Set $w_{0}=4/B$, so that $X(w_{0})=r_{x}$. Choose the fixed constants
in Eq.~\eqref{app:local:scale} so that $C_{\mathrm{ad}}S_{\max}r_{x}\le1/4$.
Together with $\|H\|_{\mathrm{loc}}\le C$, $\|V\|_{1,*}/B\le C$,
and $\mathfrak{F}r_{x}/B\le C$, this gives $T(w_{0})\le C$. Since
$0<z<w_{0}$ and $t_{k}\ge0$, we can bound the high-order terms directly:
\[
\begin{aligned}\|R_{\mathrm{alg}}\| & \le N\sum_{k\ge n+1}t_{k}z^{k}\\
 & \le N\left(\frac{z}{w_{0}}\right)^{n+1}\sum_{k\ge n+1}t_{k}w^{k}_{0}\le N\left(\frac{Bz}{4}\right)^{n+1}T(w_{0})\le CN(Bz)^{n+1}.
\end{aligned}
\]
Here the factor $N$ converts the local norm $\lVert\cdot\rVert_{\text{loc}}$
to the operator norm. This proves Eq.~\eqref{app:local:algebraic-remainder}.

We next control the remaining non-diagonal terms in Eq. \eqref{app:local:retained-defect}
caused by the difference between $\mathcal{I}-\mathcal{J}$, including
the one from the source term $z(K_{\mathrm{tr}}-K)$. The coefficient
bound yields 
\[
\sum^{n}_{k=1}z^{k}\|F_{k}\|_{k,*}\le\frac{r_{x}}{G}\frac{Bz/8}{1-Bz/8}\le C\frac{r_{x}}{G}Bz.
\]
Using $\|H\|\le CN$, $r_{x}/G\le C$, and the inverse-defect estimate
\eqref{app:local:inverse-defect}, we obtain 
\[
\begin{aligned}\sum^{n}_{k=1}z^{k}\|i[H,(\mathcal{I}-\mathcal{J})F_{k}]\| & \le CN^{2}\mathcal{P}(n,g)e^{-cn}\sum^{n}_{k=1}z^{k}\|F_{k}\|_{k,*}\\
 & \le CN^{2}\mathcal{P}(n,g)(Bz)e^{-cn}.
\end{aligned}
\]
The source term error satisfies 
\[
z\|K_{\mathrm{tr}}-K\|\le CNz\mathcal{P}(n,g)e^{-cn}\|H'\|_{\mathrm{loc}}\le CN^{2}\mathcal{P}(n,g)(Bz)e^{-cn},
\]
where we used $\|H'\|_{\mathrm{loc}}\le C$, $N\ge1$, and $B\ge1$.
Combining these bounds proves the stated remainder estimate. For fixed
$H$ in perturbative case, the term $K_{\mathrm{tr}}-K$ vanishes;
for the onsite inverse, $\mathcal{I}=\mathcal{J}=\mathcal{I}_{0}$
, thus, those error terms vanish. 
\end{proof}

We can now summarize the components in proof of Lemma~\ref{app:local:normal-form}:
\begin{proof}[Proof of Lemma~\ref{app:local:normal-form}]
Hermiticity and $[D_{k},\Pi]=0$ follow from the recursion and the
exact inverse identities. Sublemma~\ref{app:local:operation-lemma}
bounds all retained supports and gives polynomial operation costs.
Sublemma~\ref{app:local:coefficient-lemma} then gives the coefficient
bounds with $B$ in Eq.~\eqref{app:local:scale}. The source satisfies
$\nu_{*}\le C(1+G)$, so $B$ grows at most polynomially in $n,g^{-1}$.
The coefficient sum is explicit: 
\[
\|S_{\mathrm{dr}}\|\le CN\sum^{n}_{k=1}z^{k}\|S_{k}\|_{k,*}\le CNr_{x}\frac{Bz/8}{1-Bz/8}\le CNBz.
\]
Here $r_{x}\le C$ and $Bz\le\theta<1$. Sublemma~\ref{app:local:remainder-lemma}
supplies the exact transformed generator and its remainder. The prefactor
$\mathcal{P}$ is bounded in Eq.~\eqref{app:local:defect-polynomial}.
For the onsite inverse, all inverse defects vanish and $D_{k}=\mathcal{E}_{0}(F_{k})$
commutes with $D_{0}$. Its sharper operation costs also give $B=C(n+2)^{2}$. 
\end{proof}

\subsection{Propagation and flat endpoints}

\label{app:local:endpoints} 
\begin{lem}[Propagation of the dressed subspace]
\label{app:local:propagation-lemma}Let $U$ denote the propagator
of Eq.~\eqref{app:local:equation} and put $\Pi_{\mathrm{dr}}(s)=W_{n}(z,s)\Pi(s)W_{n}(z,s)^{\dagger}$.
Then 
\begin{equation}
\|(I-\Pi_{\mathrm{dr}}(1))U(1,0)\Pi_{\mathrm{dr}}(0)\|\le\frac{1}{z}\int^{1}_{0}\|R_{n}(z,s)\|\,ds.\label{app:local:leakage}
\end{equation}
\end{lem}

\begin{proof}
We use the standard comparison-propagator argument of Ref.~\cite[proof of Theorem~2.3, Eqs.~(4.19)--(4.20)]{bachmann2018adiabatic}.
Let $U_{\mathrm{cmp}}$ be generated by $H_{\mathrm{cmp}}=H+D^{(n)}+zK$,
with the same factor $z$ on the left of the Schrödinger equation.
Since $[D^{(n)},\Pi]=0$, Eq.~\eqref{app:local:exact-transport}
gives 
\[
\Pi'=-\frac{i}{z}[H_{\mathrm{cmp}},\Pi],\qquad U_{\mathrm{cmp}}(s,0)\Pi(0)=\Pi(s)U_{\mathrm{cmp}}(s,0).
\]
The transformed propagator $\widetilde{U}(s,0)=W_{n}(z,s)^{\dagger}U(s,0)W_{n}(z,0)$
is generated by $H_{\mathrm{cmp}}+R_{n}$. Duhamel's identity therefore
reads 
\[
\widetilde{U}(1,0)-U_{\mathrm{cmp}}(1,0)=-\frac{i}{z}\int^{1}_{0}U_{\mathrm{cmp}}(1,s)R_{n}(z,s)\widetilde{U}(s,0)\,ds.
\]
The comparison propagator has zero leakage. Multiplying this identity
from the left and right by $I-\Pi(1)$ and $\Pi(0)$, respectively,
and bounding the leakage by $z^{-1}\int^{1}_{0}\|R_{n}\|\,ds$. Conjugation
by the endpoint dressings gives Eq.~\eqref{app:local:leakage}. The
bound is independent of the rank of $\Pi$. 
\end{proof}

\begin{lem}[Dressing at flat endpoints]
\label{app:local:endpoint-lemma} Use the recursion \eqref{app:local:recursion-forcing}--\eqref{app:local:recursion}
and an endpoint $s_{*}\in\{0,1\}$ flat to order $n+1$. If $V(s_{*})=0$,
then $W_{n}(z,s_{*})=I$. If $H$ is fixed in the perturbative case,
$W_{n}(z,s_{*})$ agrees with the static dressing obtained by replacing
$V(s)$ with the constant $V(s_{*})$. For the onsite reference $H=D_{0}$,
$W_{n}(z,s_{*})$ is already $I$ when $[D_{0},V(s_{*})]=0$. 
\end{lem}

\begin{proof}
The endpoint induction is established in Ref.~\cite[Lemma~4.4]{bachmann2018adiabatic}
for adiabatic evolution and in Ref.~\cite[Proposition~5.1(iii)]{teufel2020nonequilibrium}
for agreement with the static construction. We need only check that
the finite inverse and the retained derivative budget preserve that
argument.

For the adiabatic case $V=0$, we prove by induction on $k$ that
\[
\partial^{j}_{s}S_{k}(s_{*})=0,\qquad1\le k\le n,\quad0\le j\le n+1-k.
\]
At first order, $F_{1}=-K_{\mathrm{tr}}=-\mathcal{J}H'$. Since the
endpoint is flat to order $n+1$, the product rule gives
\begin{equation}
\partial^{j}_{s}K_{\mathrm{tr}}(s_{*})=\sum^{j}_{r=0}\binom{j}{r}(\partial^{r}_{s}\mathcal{J})(s_{*})H^{(j-r+1)}(s_{*})=0,\qquad0\le j\le n.\label{app:local:flat-source}
\end{equation}
Every factor $H^{(j-r+1)}(s_{*})$ vanishes by the assumed flatness.
Thus $\partial^{j}_{s}F_{1}(s_{*})=0$ for $0\le j\le n$, and the
product rule applied to $S_{1}=-\mathcal{J}F_{1}$ gives the same
vanishing for $S_{1}$ as well as its derivatives.

Since each term in $F_{k}$ from Eq. \eqref{app:local:recursion-forcing}
contains a lower-order generator $S_{k_{0}}$ or its derivative $S_{k_{0}}'$,
with $1\le k_{0}<k$. 
\[
\partial^{j}_{s}F_{k}(s_{*})=0,\qquad0\le j\le n+1-k.
\]
Using $S_{k}=-\mathcal{J}F_{k}$ and the product rule once more, we
obtain
\[
\partial^{j}_{s}S_{k}(s_{*})=-\sum^{j}_{a=0}\binom{j}{a}(\partial^{a}_{s}\mathcal{J})(s_{*})\,\partial^{j-a}_{s}F_{k}(s_{*})=0,\qquad0\le j\le n+1-k.
\]
This completes the induction. In particular, every $S_{k}(s_{*})$
vanishes, so
\[
S_{\mathrm{dr}}(z,s_{*})=0,\qquad W_{n}(z,s_{*})=I.
\]

The other two statements can be proved similarly by induction. For
fixed $H$ in the perturbative case, the inverse has no $s$-dependence.
Positive derivatives of $V$ vanish at the endpoint, so the standard
induction makes the values $S_{k}(s_{*})$ equal to the static case
with the constant $V(s_{*})$. This proves the statement that in the
perturbative case, $W_{n}(z,s_{*})$ agrees with the static dressing
obtained by replacing $V(s)$ with the constant $V(s_{*})$.

Finally, if $H=D_{0}$ and $[D_{0},V(s_{*})]=0$, the exact onsite
inverse gives $S_{1}(s_{*})=-\mathcal{I}_{0}V(s_{*})=0$. Every higher
order static generator contains a lower generator, so all higher generators
vanish too by induction. This proves $W_{n}(z,s_{*})=I$. 
\end{proof}

\subsection{Proof of the main estimate}
\begin{proof}[Proof of Theorem~\ref{app:local:theorem}]
Choose the coefficients and scales from Lemma~\ref{app:local:normal-form}.
Its generator bound is given in Eq.~\eqref{app:local:dressing-size}.
The remainder estimate is uniform over the unit interval in $s$.
Lemma~\ref{app:local:propagation-lemma} therefore gives 
\[
\|(I-\Pi_{\mathrm{dr}}(1))U(1,0)\Pi_{\mathrm{dr}}(0)\|\le\frac{C}{z}\left[N(Bz)^{n+1}+N^{2}\mathcal{P}(n,g)(Bz)e^{-cn}\right],
\]
which is Eq.~\eqref{app:local:main-leakage}. At $z=\theta/B$, the
two terms become $CNB\theta^{n}$ and $CN^{2}\mathcal{P}(n,g)Be^{-cn}$,
respectively, proving Eq.~\eqref{app:local:chosen-speed}. The onsite
specialization follows from the same lemmas with the exact support-preserving
inverse. Lemma~\ref{app:local:endpoint-lemma} proves all endpoint
statements. 
\end{proof}

%% file: appendix/appendix_D_stages.tex
\section{The three stages of the physical pulse}

\label{app:stages}

We apply Theorem \ref{app:local:theorem} to the three consecutive
stages, using the gaps from Theorem \ref{app:stab:theorem}. All construction
constants are independent of $N$ and $n$. The propagators below
use the bath-rotating frame of Section \ref{app:model:frame}, which
leaves the reduced system state unchanged.

We choose the following smooth function for the parameter switching
in $H_{SE}(t)$:
\begin{equation}
f_{\uparrow}(s)=\frac{{\displaystyle \int^{s}_{0}e^{-1/[v(1-v)]}\,dv}}{{\displaystyle \int^{1}_{0}e^{-1/[v(1-v)]}\,dv}},\qquad0\le s\le1.\label{app:stages:switch}
\end{equation}
This switch function increases from zero to one, all its positive-order
endpoint derivatives vanish, and its Gevrey-2 bounds are standard~\cite[Sec.~I.1]{EH2012}:
$\sup_{s}|f^{(k)}_{\uparrow}(s)|\le CD^{k}_{\mathrm{sw}}(k!)^{2}$
with the derivative scale $D_{\mathrm{sw}}$ of Eq.\eqref{app:local:gevrey}.
By introducing $\vartheta(s)=\pi/2f_{\uparrow}(s)$, the functions
$\sin\vartheta(s)$, $\cos\vartheta(s)$, and $\vartheta'(s)$ satisfy
the same Gevrey-$2$ bounds, with possibly larger cnostants $C$ and
$D_{\text{sw}}$. The pulse angle $\vartheta$ is distinct from the
small convergence parameter $\theta$.

\subsection{Entrance: transfer of the full input subspace}

Use the projections $\Pi_{\mathrm{in}}$ and $\Pi_{\mathrm{br}}(A)$
from Eq.~\eqref{app:model:interfaces}. The pulse in this stage is
described by
\begin{equation}
\begin{split}B_{\mathrm{ent}} & =C_{\mathrm{ent}}(n+2)^{2},\qquad r=\frac{\Lambda B_{\mathrm{ent}}}{\theta},\qquad z_{\mathrm{ent}}=\frac{\Lambda}{r}=\frac{\theta}{B_{\mathrm{ent}}},\\
A(t) & =\frac{r}{\sqrt{2}}\sin\vartheta(\Lambda t),\qquad\omega(t)=r\cos\vartheta(\Lambda t),\qquad0\le t\le\Lambda^{-1},
\end{split}
\label{app:stages:entrance-pulse}
\end{equation}
where $\theta$ is a fixed constant $\theta<1$, $B_{\text{ent}}\propto(n+2)^{2}$is
the bound related to an onsite fixed Hamiltonian in Theorem~\ref{app:local:theorem},
whose choice will be clear below. Here we add a subscript to it to
emphasize that this is used for the entrance stage. The endpoint is
$A_{\mathrm{ent}}=r/\sqrt{2}$, with zero $\omega$. Choose the fixed
constant $C_{\mathrm{ent}}$ so that $A_{\mathrm{ent}}\ge\Lambda$.
Let $U_{\mathrm{ent}}$ denote the entrance propagator. 
\begin{thm}[Entrance estimate]
\label{app:stages:entrance} There are fixed constants $C,c_{0}>0$
such that, if $CN\theta^{n+1}<c_{0}$, the entrance propagator obeys
\begin{equation}
\bigl\|(I-\Pi_{\mathrm{br}}(A_{\mathrm{ent}}))U_{\mathrm{ent}}\Pi_{\mathrm{in}}\bigr\|\le\eta_{\mathrm{ent}}:=CNB_{\mathrm{ent}}\theta^{n}.\label{app:stages:entrance-error}
\end{equation}
Its physical duration is $\Lambda^{-1}$ and its amplitude $A(t)$
and chirp $\omega(t)$ peaks are $r/\sqrt{2}$ and $r$, respectively.
\end{thm}

We first state the lemmas supporting this theorem, and then, prove
this theorem at the end of this subsection.

\subsubsection{The exact reference equation and its leakage}

We first reduce the time-dependent Hamiltonian to a fixed onsite reference
Hamiltonian plus perturbation by rotating the system and dark modes.
This reduction also guarantees the choice of $B_{\text{ent}}$ in
Eq. (\eqref{app:stages:entrance-pulse}). Set $s=\Lambda t$. From
the bright/dark transformation in Section~\ref{app:model:frame},
the physical bath number is $N_{E}=NI+\sum_{j}(b^{\dagger}_{j}d_{j}+d^{\dagger}_{j}b_{j})$.
Define the Hermitian generator $Y$ of a system--dark rotation and
the resulting unitary by 
\begin{equation}
D_{0}:=N_{E},\qquad Y=i\sum_{j}(c^{\dagger}_{j}d_{j}-d^{\dagger}_{j}c_{j}),\qquad U_{0}(s)=e^{i\vartheta(s)Y}.\label{app:stages:entrance-reference}
\end{equation}

\begin{lem}[Exact entrance equation]
\label{app:stages:entrance-reference-lemma} By writing the bath-rotating-frame
state as $\ket{\tilde{\psi}(t)}=U_{0}(s)\chi(s)$, the entrance equation
in the second rotation picture with respect to $U_{0}(s)$ is 
\begin{equation}
iz_{\mathrm{ent}}\partial_{s}\chi=[D_{0}+z_{\mathrm{ent}}V(s)]\chi,\qquad V(s)=\frac{U^{\dagger}_{0}H_{u}U_{0}}{\Lambda}+\vartheta'(s)Y.\label{app:stages:entrance-equation}
\end{equation}
The onsite reference Hamiltonian $D_{0}$ satisfies $\operatorname{Proj}\ker D_{0}=\Pi_{\mathrm{in}}$.
The source $V$ has uniform finite-range Gevrey-2 local bounds. 
\end{lem}

\begin{proof}
We remark that $\ket{\psi(t)}$ is the original system-bath state
and $\ket{\tilde{\psi}(t)}$ is the state in the bath-rotating frame
defined in Sec. \eqref{app:model:frame}. The standard calculation
gives
\begin{align}
U_{0}c_{j}U^{\dagger}_{0} & =c_{j}\cos\vartheta+d_{j}\sin\vartheta,\nonumber \\
U_{0}d_{j}U^{\dagger}_{0} & =d_{j}\cos\vartheta-c_{j}\sin\vartheta,\nonumber \\
U_{0}D_{0}U^{\dagger}_{0} & =NI+\cos\vartheta\sum_{j}(b^{\dagger}_{j}d_{j}+d^{\dagger}_{j}b_{j})-\sin\vartheta\sum_{j}(c^{\dagger}_{j}b_{j}+b^{\dagger}_{j}c_{j}).\label{app:stages:reference-rotation} \nonumber
\end{align}
The exact generator then obeys 
\begin{equation}
H_{u}-\sqrt{2}A\sum_{j}(c^{\dagger}_{j}b_{j}+b^{\dagger}_{j}c_{j})+\omega N_{E}=rU_{0}D_{0}U^{\dagger}_{0}+H_{u}+N(\omega-r)I.\label{app:stages:entrance-frame-identity}
\end{equation}
Since $\partial_{t}=\Lambda\partial_{s}$, substitution of the state
$U_{0}(s)\chi(s)$ into the physical equation gives 
\[
i\Lambda\chi'=\bigl[rD_{0}+U^{\dagger}_{0}H_{u}U_{0}-i\Lambda U^{\dagger}_{0}U_{0}'\bigr]\chi,
\]
up to an additive constant which can be omitted. Now $U^{\dagger}_{0}U_{0}'=i\vartheta'Y$.
Dividing by $r$ and using $\Lambda/r=z_{\mathrm{ent}}$ proves Eq.~\eqref{app:stages:entrance-equation}. 
\end{proof}

Eq. (\eqref{app:stages:entrance-equation}) is already in the form
of $H+zV$ in Sec. \eqref{app:local:setting} with $z=z_{\text{ent}}=\Lambda/r$.
It is obvious that $D_{0}=N_{E}$ has nonnegative integer spectra,
thus, satisfies the assumption in Theorem. \eqref{app:local:theorem}.
We introduce $U_{\chi}(s,0)$ as the propagator of Eq.~\eqref{app:stages:entrance-equation}.
Applying the dressing notation of Theorem~\ref{app:local:theorem}
with $\Pi=\Pi_{\mathrm{in}}$, $R_{\mathrm{ent},n}$ the remainder
of Lemma~\ref{app:local:normal-form} gives:
\begin{lem}[Entrance leakage into the dressed subspace]
\label{app:stages:entrance-leakage-lemma} For sufficiently large
fixed $C_{\mathrm{ent}}$, the entrance dressing has the exact initial
value 
\begin{equation}
W_{n}(z_{\mathrm{ent}},0)=I.\label{app:stages:entrance-initial-identity}
\end{equation}
Its Hamiltonian remainder and its propagator leakage satisfy 
\begin{equation}
\sup_{s}\|R_{\mathrm{ent},n}(s)\|\le\mathcal{R}_{\mathrm{ent}}:=CN\theta^{n+1}.\label{app:stages:entrance-remainder}
\end{equation}
\begin{equation}
\|(I-\Pi_{\mathrm{dr}}(1))U_{\chi}(1,0)\Pi_{\mathrm{in}}\|\le\frac{1}{z_{\mathrm{ent}}}\int^{1}_{0}\|R_{\mathrm{ent},n}(s)\|\,ds\le CNB_{\mathrm{ent}}\theta^{n}.\label{app:stages:entrance-dressed-error}
\end{equation}
\end{lem}

\begin{proof}
Flatness of the switch gives 
\begin{equation}
V(0)=H_{u}/\Lambda,\qquad[D_{0},V(0)]=0,\qquad V^{(k)}(0)=0\quad(k\ge1).\label{app:stages:entrance-initial-data}
\end{equation}

The onsite operation bounds in Sublemma~\ref{app:local:operation-lemma}
give $G=O(1)$, $\mathfrak{L}=O(n)$, $\mathfrak{F}=O(n^{2})$, and
$\nu_{*}=O(1)$. Equation~\eqref{app:local:scale} therefore permits
$B_{\text{ent}}=B=C_{\text{ent}}(n+2)^{2}$, as chosen in Eq.~\eqref{app:stages:entrance-pulse}.
Eq.~\eqref{app:stages:entrance-initial-identity} directly follows
the initial commuting relation and the flat derivatives in Eq. (\eqref{app:stages:entrance-initial-data}).
Applying Theorem \eqref{app:local:theorem} proves Eq.~\eqref{app:stages:entrance-dressed-error}. 
\end{proof}

\subsubsection{Identifying the common ground-state branch}

We first give a ground-state matching lemma that will also be used
at exit.
\begin{lem}[Ground-state matching along a gapped path]
\label{app:stages:matching-lemma}Let $H(\mu)$, $\mu\in[\mu_{0},\mu_{1}]$,
be a norm-continuous family of Hamiltonians on a finite-dimensional
Hilbert space, with a nondegenerate ground state, ground projection
$P_{\mathrm{gs}}(\mu)$, and gap at least $\gamma(\mu)>0$. Let $Q(\mu)$
be a norm-continuous family of orthogonal projections satisfying
\[
\frac{\|[H(\mu),Q(\mu)]\|}{\gamma(\mu)}<\frac{1}{2}\qquad(\mu_{0}\le\mu\le\mu_{1}).
\]
If, at some $\mu_{*}\in[\mu_{0},\mu_{1}]$,
\[
\operatorname{Tr}[P_{\mathrm{gs}}(\mu_{*})Q(\mu_{*})]>\frac{1}{2},
\]
then, throughout the interval,
\begin{equation}
\|(I-Q(\mu))P_{\mathrm{gs}}(\mu)\|\le\sqrt{2}\,\frac{\|[H(\mu),Q(\mu)]\|}{\gamma(\mu)}.\label{app:stages:ground-state-matching}
\end{equation}
If $Q(\mu)$ has rank one, the same bound holds for $\|Q(\mu)-P_{\mathrm{gs}}(\mu)\|$.
\end{lem}

\begin{proof}
Fix $\mu$ and suppress its dependence. Let $\psi$ be the normalized
ground vector of $H$, let $E$ be its ground energy, and write $p=\langle\psi,Q\psi\rangle$.
Since $Q^{2}=Q$,
\[
\|(I-P_{\mathrm{gs}})Q\psi\|^{2}=p-p^{2}.
\]
The gap and $(H-E)\psi=0$ give
\begin{equation}
\gamma\sqrt{p(1-p)}\le\|(H-E)Q\psi\|=\|[H,Q]\psi\|\le\|[H,Q]\|.\label{app:stages:overlap}
\end{equation}

Eq. (\eqref{app:stages:overlap}) has two branch solutions, one gives
$p$ close to $1$ and the other gives $p$ close to $0$, which can
be fixed by the continuous argument. The ground projection is continuous,
so $p(\mu)=\operatorname{Tr}[P_{\mathrm{gs}}(\mu)Q(\mu)]$ is continuous.
Equation~\eqref{app:stages:overlap} and the assumed commutator bound
forbids $p(\mu)=1/2$, thus, one can not switch from one branch to
the other. The overlap at the point $\mu_{*}$ therefore gives $p(\mu)>1/2$
throughout the interval, and
\[
\|(I-Q)P_{\mathrm{gs}}\|=\sqrt{1-p}\le\sqrt{2}\sqrt{p(1-p)}\le\sqrt{2}\frac{\|[H,Q]\|}{\gamma}.
\]
If $Q$ has rank one, then $\|Q-P_{\mathrm{gs}}\|=\|(I-Q)P_{\mathrm{gs}}\|$,
which gives the stated projector bound.
\end{proof}

\begin{lem}[Entrance endpoint identification]
\label{app:stages:entrance-identification-lemma} Let $\Pi_{\mathrm{dr}}(1)$
be the endpoint dressed projection in the $\chi$ picture of the entrance
equation Eq. (\eqref{app:stages:entrance-equation}). Returning to
the bath-rotating frame gives a rank-one system--bright factor:
\begin{equation}
U_{0}(1)\Pi_{\mathrm{dr}}(1)U_{0}(1)^{\dagger}=P_{\mathrm{app}}\otimes I_{\mathrm{dark}},\qquad\operatorname{rank}P_{\mathrm{app}}=1.\label{app:stages:entrance-factorization}
\end{equation}
For sufficiently small fixed $c_{0}$, the condition $\mathcal{R}_{\mathrm{ent}}<c_{0}$
implies
\begin{equation}
\|P_{\mathrm{app}}-P_{\mathrm{br}}(A_{\mathrm{ent}})\|\le C\mathcal{R}_{\mathrm{ent}}.\label{app:stages:entrance-identification}
\end{equation}
\end{lem}

These are two separate claims. First, the endpoint dressing leaves
the dark sector unrestricted. Second, its rank-one system--bright
factor must be identified with the ground state of the interacting
bright Hamiltonian $H_{\text{br}}(A)$ in Eq. (\eqref{app:model:Hbr}).
We remind that the state $\widetilde{\psi}=U_{0}\chi$ is in the bath-rotating
frame of Appendix \eqref{app:model:setting}, i.e., the same frame
as $P_{\text{app}}$ since we conjugate $\Pi_{\text{dr}}(1)$ by $U_{0}(1)$. 

To show the first statement, notice that the endpoint is flat, so
the flat-endpoint result of Lemma \ref{app:local:endpoint-lemma}
identifies the endpoint dressing with the static recursion for $D_{0}+z_{\mathrm{ent}}V(1)$,
with the same expansion parameter $z_{\text{ent}}$, It can be directly
verified that $V(1)$ as well as $D_{0}$ acts only on the bright-
and dark-bath sectors since $\vartheta'(1)=0$. Therefore, $\Pi_{\text{dr}}(1)$
can be factorized as an Identity operator on the system part, tensor
producted with an operator on the bright- and dark- baths. Conjugating
by $U_{0}(1)$ then proves the factorization statement.

To show the second statement, we notice that 
\begin{equation}
H_{\text{br}}(A_{\text{ent}})\otimes I_{\text{dark}}=rU_{0}(1)(D_{0}+z_{\text{ent}}V(1))U^{\dagger}_{0}(1),
\end{equation}
up to some additive scalars. So we first prove that $\Pi_{\text{dr}}(1)$
is close to the ground state of $D_{0}+z_{\text{ent}}V(1)$, denoted
as $P_{\text{dr}}(z_{\text{ent}})$, then conjugating by $U_{0}(1)$
would to the second statement.

To proceed, from Lemma. \eqref{app:local:normal-form} we have 
\begin{equation}
\lVert[\Pi_{\text{dr}}(1),D_{0}+z_{\text{ent}}V(1)]\rVert=\lVert[W^{\dagger}_{n}\Pi_{\text{dr}}(1)W_{n},D_{0}+D^{(n)}+R_{n,\text{ent}}(1)]\rVert=\lVert[W^{\dagger}_{n}\Pi_{\text{dr}}(1)W_{n},R_{n,\text{ent}}(1)]\rVert,
\end{equation}
where $W_{n}$ is the dressing unitary for $D_{0}+z_{\text{ent}}V(1)$,
and in the last equality, we have used that $W^{\dagger}_{n}\Pi_{\text{dr}}(1)W_{n}$
is the ground state of $D_{0}$ as well as $[D_{0},D^{(n)}]=0$. By
the upper bound of $R_{n,\text{ent}}(1)$, we then have 
\begin{equation}
\lVert[\Pi_{\text{dr}}(1),D_{0}+z_{\text{ent}}V(1)]\rVert\leq2\mathcal{R}_{\text{ent}}.\label{eq:appendix:D:dresisng}
\end{equation}

In order to show that $\Pi_{\text{dr}}(1)$ is indeed close to the
ground space of $D_{0}+z_{\text{ent}}V(1)$, we consider a virtual
class of Hamiltonians $H_{\text{vir}}(z)=D_{0}+zV(1)$ with $z\in[0,z_{\text{ent}}]$.
It is just for the proof purpose, not related to a real evolution
path. $H_{\text{vir}}(z)$ is then identified with $H(\mu)$ in Lemma.
\eqref{app:stages:matching-lemma}, and we define $Q(z)$ as the dressing
projector built from $H_{\text{vir}}(z)$ with $Q(z_{\text{ent}})=\Pi_{\text{dr}}(1)$:
$Q(z):=W_{n}(z,1)\Pi_{\mathrm{in}}W_{n}(z,1)^{\dagger},$where $W_{n}(z,1)$
is the static dressing for $D_{0}+zV(1)$. $P_{\text{gs}}(z)$ is
the ground state of $H_{\text{vir}}(z)$. Theorem \ref{app:stab:theorem}
establishes the gap of $H_{\text{vir}}(z)$ as 
\begin{equation}
\operatorname{gap}H_{\text{vir}}(z)\ge(c_{\mathrm{gap}}/r)\min\{A_{\mathrm{ent}},A^{2}_{\mathrm{ent}}z_{\text{ent}}/(z\Lambda)\}.
\end{equation}
Therefore, $P_{\text{gs}}(z)$ and $Q(z)$ are continuous. Further,
at $z=0$, we have $P_{\text{gs}}(0)=Q_{\text{gs}}(0)$, and 
\begin{equation}
\lVert[Q_{\text{gs}}(z),H_{\text{vir}}(z)]\rVert\leq2\mathcal{R}_{\text{ent}}
\end{equation}
follows the same argument from Eq. (\eqref{eq:appendix:D:dresisng}).
Therefore, applying Lemma. \eqref{app:stages:matching-lemma} to the
bath sector gives rise to
\begin{equation}
\lVert(I-P_{\text{gs}}(z_{\text{ent}}))\Pi_{\text{dr}}(1)\rVert\leq C\mathcal{R}_{\text{ent}}
\end{equation}
for some constant $C$. Conjugating by $U_{0}(1)$ then proves Lemma.
\eqref{app:stages:entrance-identification-lemma}.

We now prove the Theorem \eqref{app:stages:entrance}.
\begin{proof}[Proof of Theorem~\ref{app:stages:entrance}]
Lemma~\ref{app:stages:entrance-reference-lemma} identifies the
physical entrance propagator with $U_{0}(1)U_{\chi}(1,0)$, up to
a scalar phase. Let $Q=U_{0}(1)\Pi_{\mathrm{dr}}(1)U_{0}(1)^{\dagger}$
and $P=\Pi_{\mathrm{br}}(A_{\mathrm{ent}})$. Insert $I=Q+(I-Q)$
immediately after the propagator together with the triangle inequality
gives 
\[
\|(I-P)U_{\mathrm{ent}}\Pi_{\mathrm{in}}\|\le\|(I-Q)U_{\mathrm{ent}}\Pi_{\mathrm{in}}\|+\|Q-P\|.
\]
Lemmas \ref{app:stages:entrance-leakage-lemma} and \ref{app:stages:entrance-identification-lemma}
bound the two resulting errors by $CNB_{\mathrm{ent}}\theta^{n}$
and $C\mathcal{R}_{\mathrm{ent}}$, respectively. Their sum has the
first scale, proving Eq.~\eqref{app:stages:entrance-error}. The
duration and control peaks follow directly from Eq.~\eqref{app:stages:entrance-pulse}. 
\end{proof}

\subsection{Transport: normalized amplitude segments}

During transport, $\omega=0$ and the dark sector is decoupled. The
bright gap decreases when the amplitude decreases, so we normalize
the evolution separately on intervals whose amplitudes differ by at
most a factor of two. We stop at a positive amplitude $A_{c}$ and
then release the remaining coupling in the exit stage. For the complete
protocol choose 
\begin{equation}
B_{e}=C_{e}n^{6D}L^{8D}_{n}\ge1,\qquad a_{c}=\frac{\theta}{B_{e}},\qquad A_{c}=\Lambda a_{c}<\Lambda.\label{app:stages:exit-scale}
\end{equation}
where $C_{e}$ is a sufficiently large fixed constant. Starting from
$A_{0}=A_{\mathrm{ent}}$, define the lower endpoints recursively
by 
\begin{equation}
A_{j+1}=\max\{A_{j}/2,A_{c}\},\qquad0\le j<M_{\mathrm{tr}},\qquad M_{\text{tr}}=\lceil\log_{2}\frac{A_{\text{ent}}}{A_{c}}\rceil,\qquad A_{M_{\mathrm{tr}}}=A_{c},\label{app:stages:amplitude-sequence}
\end{equation}
For segment $j$, let $s\in[0,1]$ be normalized time. The physical
time elapsed since the beginning of the segment is $t=s\tau_{j}$,
where $\tau_{j}$ is the segment duration. We interpolate the amplitude
as
\begin{equation}
A^{(j)}(s)=A_{j}-(A_{j}-A_{j+1})f_{\uparrow}(s),\quad a_{j}=A_{j}/\Lambda,\quad J_{j}=\max\{1,a_{j}\},\quad g_{j}=c_{g}\min\{1,a^{2}_{j}\},\label{app:stages:segment-normalization}
\end{equation}
where $0<c_{g}\le\min\{1,c_{\mathrm{gap}}/4\}$ is fixed. The scale
$\Lambda J_{j}$ bounds the local strength of the segment Hamiltonian,
and $g_{j}$ is a lower bound on the gap after dividing the segment
Hamiltonian by $\Lambda J_{j}$. We choose the dimensionless duration
$\Lambda J_{j}\tau_{j}=B_{j}/\theta$, where
\begin{equation}
B_{j}=C_{\mathrm{tr}}n^{7D}L^{10D}_{n}g^{-(5D+5/2)},\qquad\tau_{j}=\frac{B_{j}}{\theta\Lambda J_{j}}.\label{app:stages:transport-scale}
\end{equation}

Let $U_{\mathrm{tr},j}$ be the physical propagator on segment $j$,
$U_{\mathrm{tr}}=U_{\mathrm{tr},M_{\mathrm{tr}}-1}\cdots U_{\mathrm{tr},0}$,
and $\tau_{\mathrm{tr}}=\sum^{M_{\mathrm{tr}}-1}_{j=0}\tau_{j}$. 
\begin{thm}[Transport estimate]
\label{app:stages:transport} For every endpoint amplitude $0<A_{c}<\Lambda$,
set $a_{c}=A_{c}/\Lambda$ and use the sequence and durations in Eqs.~\eqref{app:stages:amplitude-sequence}--\eqref{app:stages:transport-scale}.
In particular, the bounds below apply to the choice in Eq.~\eqref{app:stages:exit-scale}.
For sufficiently large fixed $C_{\mathrm{tr}}$, segment $j$ satisfies
\begin{equation}
\|(I-\Pi_{\mathrm{br}}(A_{j+1}))U_{\mathrm{tr},j}\Pi_{\mathrm{br}}(A_{j})\|\le\eta_{j}:=C\left[NB_{j}\theta^{n}+N^{2}\mathcal{P}(n,g_{j})B_{j}e^{-cn}\right].\label{app:stages:segment-error}
\end{equation}
The endpoint projections are the true interacting ground projections.
The total leakage $\|(I-\Pi_{\mathrm{br}}(A_{c}))U_{\mathrm{tr}}\Pi_{\mathrm{br}}(A_{\mathrm{ent}})\|$
is at most $\sum^{M_{\mathrm{tr}}-1}_{j=0}\eta_{j}$, and the total
duration satisfies
\begin{equation}
\Lambda\tau_{\mathrm{tr}}\le Cn^{7D}L^{10D}_{n}\bigl(1+a^{-(10D+5)}_{c}\bigr).\label{app:stages:geometric-cost}
\end{equation}
\end{thm}

We first state the Lemmas in supporting this theorem, and prove it
at the end of this subsection.

\subsubsection{The normalized gap and segment error}

For segment $j$, define $H_{j}(s)=H_{\mathrm{br}}(A^{(j)}(s))/(\Lambda J_{j})$.
The physical equation is $i\partial_{t}\ket{\tilde{\psi}}=H_{\mathrm{br}}(A^{(j)}(t/\tau_{j}))\ket{\tilde{\psi}}$.
with time measured from the beginning of the segment. Substituting
$s=t/\tau_{j}$ and dividing by $\Lambda J_{j}$ gives
\begin{equation}
iz_{j}\partial_{s}\ket{\tilde{\psi}}=H_{j}(s)\ket{\tilde{\psi}},\qquad z_{j}=(\Lambda J_{j}\tau_{j})^{-1}=\theta/B_{j}.\label{app:stages:transport-equation}
\end{equation}
This is the adiabatic case of Theorem~\ref{app:local:theorem}, with
$H=H_{j}$ and $V=0$. We check its gap and local bounds, then verify
that $B_{j}$ controls the required expansion scale.
\begin{lem}[Uniform segment bounds]
\label{app:stages:transport-segment-lemma} For every segment, $\operatorname{gap}H_{j}(s)\ge g_{j}$
and
\begin{equation}
\sup_{s}\|\partial^{k}_{s}H_{j}(s)\|_{\mathrm{loc}}\le C\min\{1,a_{j}\}D^{k}_{\mathrm{sw}}(k!)^{2}\label{app:stages:amplitude-derivatives}
\end{equation}
for every integer $k\ge1$. The scale $B_{j}$ in Eq.~\eqref{app:stages:transport-scale}
is sufficient for Theorem~\ref{app:local:theorem}; its application
gives Eq.~\eqref{app:stages:segment-error} with the true ground
projections at both endpoints.
\end{lem}

\begin{proof}
Throughout segment $j$, including the final segment, $A_{j}/2\le A^{(j)}(s)\le A_{j}$.
If $a_{j}\ge2$, the whole segment lies above $\Lambda$, and division
of the linear bright-gap bound by $\Lambda J_{j}=A_{j}$ gives $\operatorname{gap}H_{j}\ge c_{\mathrm{gap}}/2$.
If $1\le a_{j}<2$, the part below $\Lambda$ has $\operatorname{gap}H_{j}\ge c_{\mathrm{gap}}a_{j}/4\ge c_{\mathrm{gap}}/4$,
while the part above $\Lambda$ obeys the linear bound. If $a_{j}\le1$,
the quadratic bound gives $\operatorname{gap}H_{j}\ge c_{\mathrm{gap}}a^{2}_{j}/4$.
These follow from Eq.~\eqref{app:stab:gaps}, with the gap scale
of Eq.~\eqref{app:stab:free-gap-scale}. Since $c_{g}\le\min\{1,c_{\mathrm{gap}}/4\}$,
they prove $\operatorname{gap}H_{j}(s)\ge g_{j}$.

The local strength of the physical Hamiltonian $H_{\text{br}}(A^{(j)}(s))$
is at most $C\Lambda J_{j}$, so that of $H_{j}$ is uniformly bounded.
Only the coupling amplitude depends on $s$. The switch bound gives
$|\partial^{k}_{s}A^{(j)}|\le CA_{j}D^{k}_{\mathrm{sw}}(k!)^{2}$;
division by $\Lambda J_{j}$ supplies the factor $a_{j}/J_{j}=\min\{1,a_{j}\}$
and proves Eq.~\eqref{app:stages:amplitude-derivatives}. In the
derivative norm of Eq.~\eqref{app:local:jet-norm}, this yields
\[
\|H_{j}'\|_{1,*}\le C\min\{1,a_{j}\}\le C'\sqrt{g_{j}}.
\]
The construction in Appendix~\ref{app:local:section} accounts for
the changing ground projection through $K_{\mathrm{tr}}=\mathcal{J}H_{j}'$.
The inverse bound gives $\|K_{\mathrm{tr}}\|_{1,*}\le CG\sqrt{g_{j}}$.
Taking the constant in $G$ large enough that $G\sqrt{g_{j}}\ge1$,
we may choose $\nu_{*}=CG\sqrt{g_{j}}$ in Eq.~\eqref{app:local:source-size}. 

For the adiabatic construction, use $\mathfrak{F}=\mathfrak{F}_{\mathrm{ad}}$
in Eq.~\eqref{app:local:explicit-costs}. Substitution into Eq.~\eqref{app:local:scale}
gives
\begin{equation}
B\le C(G^{3}\mathfrak{L}^{2}\sqrt{g_{j}}+G\mathfrak{F})\le C_{\text{tr}}n^{7D}L^{10D}_{n}g^{-(5D+5/2)}_{j}.\label{app:stages:transport-scale-derivation}
\end{equation}
Indeed, the two contributions obey
\begin{align*}
G^{3}\mathfrak{L}^{2}\sqrt{g_{j}} & \le Cn^{7D}L^{10D}_{n}g^{-(5D+5/2)}_{j},\\
G\mathfrak{F} & \le Cn^{3D+3}L^{4D+2}_{n}g^{-(2D+2)}_{j}.
\end{align*}
Thus a sufficiently large fixed $C_{\mathrm{tr}}$ gives $B\le B_{j}$,
and the chosen duration ensures $Bz_{j}\le B_{j}z_{j}=\theta$. 

Lemma~\ref{app:local:normal-form} now gives
\[
\sup_{s}\|R_{n}(z_{j},s)\|\le CN\theta^{n+1}+CN^{2}\mathcal{P}(n,g_{j})\theta e^{-cn}.
\]
The reference projection is the ground projection of $H_{j}(s)$.
Since $V=0$ and the amplitude switch is flat at both endpoints, Lemma~\ref{app:local:endpoint-lemma}
gives identity dressings there. Lemma~\ref{app:local:propagation-lemma}
therefore bounds leakage between the true endpoint ground projections
by
\[
\frac{1}{z_{j}}\int^{1}_{0}\|R_{n}(z_{j},s)\|\,ds\le C\left[NB_{j}\theta^{n}+N^{2}\mathcal{P}(n,g_{j})B_{j}e^{-cn}\right].
\]
This is Eq.~\eqref{app:stages:segment-error} on the system--bright
factor. Tensoring with the dark identity gives the same operator norm
on the full space.
\end{proof}

\subsubsection{Adding the segment durations}
\begin{lem}[Geometric transport cost]
\label{app:stages:transport-duration-lemma} The durations in Eq.~\eqref{app:stages:transport-scale}
obey Eq.~\eqref{app:stages:geometric-cost}, while $M_{\mathrm{tr}}=\lceil\log_{2}(A_{\mathrm{ent}}/A_{c})\rceil$.
\end{lem}

\begin{proof}
Substituting $J_{j}$ and $g_{j}$ into Eq.~\eqref{app:stages:transport-scale}
gives
\[
\Lambda\tau_{j}\le Cn^{7D}L^{10D}_{n}\begin{cases}
a^{-1}_{j}, & a_{j}\ge1,\\
a^{-(10D+5)}_{j}, & a_{j}<1.
\end{cases}
\]
Here $C$ absorbs the fixed $\theta$ and $c_{g}$. The upper amplitudes
halve until the final segment. For $A_{j}\ge\Lambda$, summing the
first line uses
\begin{equation}
\sum_{j:A_{j}\ge\Lambda}\frac{1}{A_{j}}\le\frac{2}{\Lambda}.\label{app:stages:large-amplitude-sum}
\end{equation}
For $A_{j}<\Lambda$, the last upper amplitude lies in $(A_{c},2A_{c}]$.
Reading the halving sequence backward gives
\begin{equation}
\sum_{j:A_{j}<\Lambda}\left(\frac{\Lambda}{A_{j}}\right)^{10D+5}\le\frac{a^{-(10D+5)}_{c}}{1-2^{-(10D+5)}}.\label{app:stages:small-amplitude-sum}
\end{equation}
The two inequalities together prove Eq.~\eqref{app:stages:geometric-cost}. 
\end{proof}

We can now prove Theorem~\ref{app:stages:transport}.
\begin{proof}[Proof of Theorem~\ref{app:stages:transport}]
Lemma~\ref{app:stages:transport-segment-lemma} gives the leakage
on every segment. For two consecutive propagators, write their shared
projection as $P$ and their input and output projections as $P_{0},P_{2}$.
Inserting $I=P+(I-P)$ and using the triangle inequality gives
\begin{equation}
\|(I-P_{2})U_{2}U_{1}P_{0}\|\le\|(I-P_{2})U_{2}P\|+\|(I-P)U_{1}P_{0}\|.\label{app:stages:composition-bound}
\end{equation}
Iterating this inequality with the shared projection $\Pi_{\mathrm{br}}(A_{j})$
gives the stated total leakage. Lemma~\ref{app:stages:transport-duration-lemma}
gives the duration bound.
\end{proof}

\subsection{Exit: matching and release through the system gap}

The exit decreases the coupling from $A_{c}$ to zero to transfer
the incoming bright ground sector to the system ground sector. We
bound the transfer error below. The bright gap closes as $A\to0$,
whereas the system gap remains positive. We therefore use the construction
of the perturbative case in Appendix~\ref{app:local:section}, with
the uncoupled system Hamiltonian $H_{u}$ as a fixed reference.

We still work on the system--bright factor. Let $E_{u}$ and $P_{u}$
be the ground energy and ground projection of $H_{u}$, respectively,
and let $I_{\mathrm{br}}$ be the bright identity. Define
\begin{equation}
H_{\mathrm{ref}}=\frac{H_{u}-E_{u}}{\Lambda}\otimes I_{\mathrm{br}},\qquad\Pi=P_{u}\otimes I_{\mathrm{br}},\qquad\delta_{0}=\frac{\Delta}{4\Lambda}.\label{app:stages:exit-reference}
\end{equation}
Then $H_{\mathrm{ref}}\Pi=0$ and $H_{\mathrm{ref}}\ge\delta_{0}(I-\Pi)$.
The referenced ground space allows every bright state and has rank
$2^{N}$. Restoring the dark identity gives $\Pi\otimes I_{\mathrm{dark}}=\Pi_{\mathrm{out}}$.

By introducing $a=A/\Lambda$, the dimensionless bright Hamiltonian
can be written as $H_{\mathrm{ref}}+aX_{\mathrm{br}},$ with
\begin{equation}
X_{\mathrm{br}}=-\sqrt{2}\sum_{j}(c^{\dagger}_{j}b_{j}+b^{\dagger}_{j}c_{j}),\qquad f(s)=1-f_{\uparrow}(s)\qquad a(s)=a_{c}f(s).\label{app:stages:exit-coupling}
\end{equation}
Here $\omega=0$, and $a_{c}$ is defined by $A_{c}/\Lambda$ in Eq.~\eqref{app:stages:exit-scale}.
We choose $s=A_{c}t$, where $t$ is the physical time elapsed since
the beginning of the Exit, so that the inverse duration and the referenced
energy scale give the dimensionless parameter $z=a_{c}$. After removing
the scalar ground energy, the physical equation $i\partial_{t}\ket{\tilde{\psi}}=\Lambda[H_{\mathrm{ref}}+a_{c}f(s)X_{\mathrm{br}}]\ket{\tilde{\psi}}$
becomes
\begin{equation}
ia_{c}\partial_{s}\ket{\tilde{\psi}}=[H_{\mathrm{ref}}+a_{c}f(s)X_{\mathrm{br}}]\ket{\tilde{\psi}},\qquad\tau_{\mathrm{exit}}=\frac{1}{A_{c}}.\label{app:stages:exit-equation}
\end{equation}
We notice that $B_{e}a_{c}=\theta$ as defined in Eq.~\eqref{app:stages:exit-scale},
which links this duration to the expansion scale in the perturbative
case with fixed-referenced Hamiltonian.

Let $U_{\mathrm{exit}}$ denote this propagator. Write $\mathcal{P}_{e}=\mathcal{P}(n,\delta_{0})$
for the localization-defect prefactor of Theorem~\ref{app:local:theorem}.
There are two steps. First, the incoming bright ground state must
lie close to the statically dressed reference ground space. Second,
the exit dynamics must carry that dressed space to $\Pi$. We bound
the first error by $\sqrt{2}\,\eta_{\mathrm{match}}$ and the second
by $\eta_{\mathrm{dyn}}$, where
\begin{align}
\eta_{\mathrm{match}} & =C\left[NB^{2}_{e}\theta^{n-1}+N^{3}\mathcal{P}_{e}B^{2}_{e}e^{-cn}\right],\label{app:stages:matching-error}\\
\eta_{\mathrm{dyn}} & =C\left[NB_{e}\theta^{n}+N^{2}\mathcal{P}_{e}B_{e}e^{-cn}\right].\label{app:stages:exit-dynamical-error}
\end{align}

\begin{thm}[Exit matching and dynamics]
\label{app:stages:exit} Use $a_{c}=\theta/B_{e}$ and the pulse
in Eq.~\eqref{app:stages:exit-equation}. For sufficiently large
fixed $C_{e}$, if $\eta_{\mathrm{match}}<1/8$, then
\begin{equation}
\|(I-\Pi_{\mathrm{out}})U_{\mathrm{exit}}\Pi_{\mathrm{br}}(A_{c})\|\le\sqrt{2}\,\eta_{\mathrm{match}}+\eta_{\mathrm{dyn}}.\label{app:stages:exit-error}
\end{equation}
The duration is $\tau_{\mathrm{exit}}=1/A_{c}$. The pulse decreases
directly from $A_{c}$ to zero.
\end{thm}

We also first state the supporting lemmas, and prove this theorem
at the end of this subsection.

\subsubsection{The fixed-reference scale and static dressing}

We first need a characterization for the static dressing to match
the endpoint. For every static amplitude $0\le a\le a_{c}$, we apply
Appendix~\ref{app:local:section} with $H=H_{\mathrm{ref}}$, $V=X_{\mathrm{br}}$,
$\text{gap}(H_{\text{ref}})=\delta_{0}=\frac{\Delta}{4\Lambda}$,
and $z=a$. Let $W^{\mathrm{stat}}_{n}(a)$ be the static dressing
for $H_{\mathrm{ref}}+aX_{\mathrm{br}}$, and set
\[
\Pi^{\mathrm{stat}}_{\mathrm{dr}}(a)=W^{\mathrm{stat}}_{n}(a)\Pi W^{\mathrm{stat}}_{n}(a)^{\dagger}.
\]
Write $R_{\mathrm{stat}}(a)$ for its dressing remainder, as in Lemma~\ref{app:local:normal-form}.
\begin{lem}[Static remainder and commutator]
\label{app:stages:exit-remainder-lemma} For the static construction,
a sufficient expansion scale at reference gap $0<g\le1$ is
\begin{equation}
B_{\mathrm{fix}}(n,g)=C_{\mathrm{fix}}n^{6D}L^{8D}_{n}g^{-(4D+2)}.\label{app:stages:fixed-reference-scale}
\end{equation}
For sufficiently large fixed $C_{e}$, Eq.~\eqref{app:stages:exit-scale}
gives $B_{e}\ge\max\{1,B_{\mathrm{fix}}(n,\delta_{0})\}$. For $0\le a\le a_{c}$,
the static remainder and commutator satisfy
\begin{equation}
\|R_{\mathrm{stat}}(a)\|\le CN(B_{e}a)^{n+1}+CN^{2}\mathcal{P}_{e}(B_{e}a)e^{-cn}.\label{app:stages:static-remainder}
\end{equation}
\begin{equation}
\|[H_{\mathrm{ref}}+aX_{\mathrm{br}},\Pi^{\mathrm{stat}}_{\mathrm{dr}}(a)]\|\le2\|R_{\mathrm{stat}}(a)\|.\label{app:stages:static-commutator}
\end{equation}
\end{lem}

\begin{proof}
The referenced Hamiltonian is fixed, so $H_{\mathrm{ref}}'=0$ and
$K_{\mathrm{tr}}=0$. For the static construction take $V=X_{\mathrm{br}}$
and $z=a$. Since both $H_{\mathrm{ref}}$ and $V$ are independent
of $s$, the recursion of Appendix~\ref{app:local:section} produces
$s$-independent coefficients $S_{k}$ and $D_{k}$. Thus the derivative
term in the normal form vanishes. The first-order source $X_{\mathrm{br}}$
has bounded local strength, so $\nu_{*}=O(1)$. The coefficient argument
leading to Eq.~\eqref{app:local:scale} reduces to a sufficient scale
$B\le CG^{2}\mathfrak{L}^{2}$. By Eq.~\eqref{app:local:explicit-costs},
\[
G^{2}\mathfrak{L}^{2}\le Cn^{6D}L^{8D}_{n}g^{-(4D+2)}.
\]
This proves Eq.~\eqref{app:stages:fixed-reference-scale}. At $g=\delta_{0}$,
the fixed gap powers can be absorbed in $C_{e}$, giving $B\le B_{e}$.
Since $B_{e}a\le B_{e}a_{c}=\theta$, Lemma~\ref{app:local:normal-form}
gives Eq.~\eqref{app:stages:static-remainder} for $0<a\le a_{c}$.
At $a=0$, the static dressing is the identity and the remainder vanishes
$R_{\text{stat}}(0)=0$.

For the static problem, the normal form reads
\[
W^{\mathrm{stat}\dagger}_{n}(H_{\mathrm{ref}}+aX_{\mathrm{br}})W^{\mathrm{stat}}_{n}=H_{\mathrm{ref}}+D^{(n)}+R_{\mathrm{stat}},
\]
where both $H_{\mathrm{ref}}$ and the correction $D^{(n)}$ commute
with $\Pi$. Conjugating the commutator back gives
\[
[H_{\mathrm{ref}}+aX_{\mathrm{br}},\Pi^{\mathrm{stat}}_{\mathrm{dr}}]=W^{\mathrm{stat}}_{n}[R_{\mathrm{stat}},\Pi]W^{\mathrm{stat}\dagger}_{n}.
\]
Unitarity and $\|[R_{\mathrm{stat}},\Pi]\|\le2\|R_{\mathrm{stat}}\|$
prove Eq.~\eqref{app:stages:static-commutator}.
\end{proof}

\subsubsection{One-sided matching of the incoming ground state}
\begin{lem}[Static matching]
\label{app:stages:exit-matching-lemma} If $\eta_{\mathrm{match}}<1/8$,
then
\begin{equation}
\|(I-\Pi^{\mathrm{stat}}_{\mathrm{dr}}(a_{c}))P_{\mathrm{br}}(A_{c})\|\le\sqrt{2}\,\eta_{\mathrm{match}}.\label{app:stages:one-sided-matching}
\end{equation}
\end{lem}

\begin{proof}
We apply Lemma~\ref{app:stages:matching-lemma} to the static amplitude
path. First we find a positive amplitude $a$ where the true ground
state has small leakage from $\Pi$ and the dressed projection is
close to $\Pi$.

For $0<a\le a_{c}$, let $\psi_{a}$ be a normalized ground vector
of $H_{\mathrm{br}}(\Lambda a)$, and set $\pi(a)=\langle\psi_{a},\Pi^{\mathrm{stat}}_{\mathrm{dr}}(a)\psi_{a}\rangle$.
Choose a fixed $C_{X}$ with $\|X_{\mathrm{br}}\|\le C_{X}N$, and
let $E(a)$ be the ground energy of $H_{\mathrm{ref}}+aX_{\mathrm{br}}$.
A trial vector in $\operatorname{Ran}\Pi$ gives $E(a)\le a\|X_{\mathrm{br}}\|$.
Consequently,
\[
\langle\psi_{a},H_{\mathrm{ref}}\psi_{a}\rangle=E(a)-a\langle\psi_{a},X_{\mathrm{br}}\psi_{a}\rangle\le2a\|X_{\mathrm{br}}\|.
\]
Using $H_{\mathrm{ref}}\ge\delta_{0}(I-\Pi)$ gives
\begin{equation}
1-\langle\psi_{a},\Pi\psi_{a}\rangle\le\frac{2C_{X}Na}{\delta_{0}}.\label{app:stages:anchor-energy}
\end{equation}
The dressing-generator bound in Eq.~\eqref{app:local:dressing-size}
also gives $\|W^{\mathrm{stat}}_{n}(a)-I\|\le CNB_{e}a$ and hence
$\|\Pi^{\mathrm{stat}}_{\mathrm{dr}}(a)-\Pi\|\le CNB_{e}a$. Combining
the two estimates yields
\[
\pi(a)\ge1-\frac{2C_{X}Na}{\delta_{0}}-CNB_{e}a.
\]
Thus the two required small quantities are $Na/\delta_{0}$ and $NB_{e}a$.

Choose a sufficiently small fixed $c_{\mathrm{low}}>0$, with $c_{\mathrm{low}}\delta_{0}<\theta$,
and define
\begin{equation}
a_{\mathrm{low}}=\frac{c_{\mathrm{low}}\delta_{0}}{NB_{e}},\qquad0<a_{\mathrm{low}}<a_{c}.\label{app:stages:comparison-amplitude}
\end{equation}
Here $a_{\mathrm{low}}/a_{c}=c_{\mathrm{low}}\delta_{0}/(N\theta)<1$,
and the preceding overlap estimate gives
\begin{equation}
\pi(a_{\mathrm{low}})\ge1-Cc_{\mathrm{low}}/B_{e}-Cc_{\mathrm{low}}\delta_{0}>3/4.\label{app:stages:anchor-overlap}
\end{equation}
This amplitude is used only to identify the ground-state branch. It
adds no physical pulse segment.

We verify the hypotheses of Lemma~\ref{app:stages:matching-lemma}
on $[a_{\mathrm{low}},a_{c}]$ with
\[
H(a)=H_{\mathrm{ref}}+aX_{\mathrm{br}},\qquad Q(a)=\Pi^{\mathrm{stat}}_{\mathrm{dr}}(a),\qquad P_{\mathrm{gs}}(a)=P_{\mathrm{br}}(\Lambda a).
\]
Theorem~\ref{app:stab:theorem} gives a unique ground state and gap
at least $\gamma(a)=c_{\mathrm{gap}}a^{2}$ on this interval. The
ground projection and the finite static dressing are continuous. It
remains to bound the commutator relative to this gap uniformly.

After division by $a^{2}$, the two terms of Eq.~\eqref{app:stages:static-remainder}
are proportional to $a^{n-1}$ and $a^{-1}$. Their maxima occur at
$a_{c}$ and $a_{\mathrm{low}}$, respectively. Substitution gives
\begin{align*}
NB^{n+1}_{e}a^{n-1}_{c} & =NB^{2}_{e}\theta^{n-1},\\
\frac{N^{2}\mathcal{P}_{e}B_{e}}{a_{\mathrm{low}}}e^{-cn} & =\frac{N^{3}\mathcal{P}_{e}B^{2}_{e}}{c_{\mathrm{low}}\delta_{0}}e^{-cn}.
\end{align*}
Together with Eq.~\eqref{app:stages:static-commutator}, these bounds
give
\begin{equation}
\begin{aligned}\frac{\|[H_{\mathrm{ref}}+aX_{\mathrm{br}},\Pi^{\mathrm{stat}}_{\mathrm{dr}}(a)]\|}{c_{\mathrm{gap}}a^{2}} & \le\frac{2\|R_{\mathrm{stat}}(a)\|}{c_{\mathrm{gap}}a^{2}}\\
 & \le C\left[NB^{2}_{e}\theta^{n-1}+N^{3}\mathcal{P}_{e}B^{2}_{e}e^{-cn}\right]=\eta_{\mathrm{match}}<\frac{1}{8}.
\end{aligned}
\label{app:stages:overlap-uniform}
\end{equation}
The constant absorbs the fixed $c_{\mathrm{gap}},c_{\mathrm{low}}$,
and $\delta_{0}$. Together with Eq.~\eqref{app:stages:anchor-overlap},
this verifies the hypotheses of Lemma~\ref{app:stages:matching-lemma}.
Its conclusion at $a=a_{c}$ proves Eq.~\eqref{app:stages:one-sided-matching}.
The incoming ground projection has rank one, while the dressed projection
has rank $2^{N}$. Thus, this result bounds leakage from the incoming
state into the complement of that dressed space.
\end{proof}

\subsubsection{Dynamical release and composition of the errors}

We can now prove Theorem~\ref{app:stages:exit}.
\begin{proof}[Proof of Theorem~\ref{app:stages:exit}]
Apply the fixed referenced Hamiltonian construction in the perturbative
case to Eq.~\eqref{app:stages:exit-equation}, with $z=a_{c}$ and
$V(s)=f(s)X_{\mathrm{br}}$. The referenced Hamiltonian is fixed,
so $H_{\mathrm{ref}}'=K_{\mathrm{tr}}=0$. Since $f$ has bounded
Gevrey-2 derivatives, this source has $\nu_{*}=O(1)$ and $\mathfrak{F}=\mathfrak{F}_{\mathrm{pert}}\le Cn^{2}$.
Equations~\eqref{app:local:scale} and \eqref{app:local:explicit-costs}
therefore give $B\le Cn^{6D}L^{8D}_{n}\delta^{-(4D+2)}_{0}\le B_{e}$,
so the fixed-reference construction applies to the exit evolution.
The switch has $f(0)=1$, $f(1)=0$, and all positive endpoint derivatives
zero. Lemma~\ref{app:local:endpoint-lemma} then gives the exact
endpoints
\begin{equation}
W_{n}(a_{c},0)=W^{\mathrm{stat}}_{n}(a_{c}),\qquad W_{n}(a_{c},1)=I.\label{app:stages:exact-exit-endpoints}
\end{equation}

From Lemma~\ref{app:local:normal-form} and $B_{e}a_{c}=\theta$,
we have 
\begin{equation}
a^{-1}_{c}\int^{1}_{0}\|R_{n}(a_{c},s)\|\,ds\le\eta_{\mathrm{dyn}}.
\end{equation}
Lemma~\ref{app:local:propagation-lemma} and the endpoint identities
in Eq.~\eqref{app:stages:exact-exit-endpoints} therefore give
\begin{equation}
\|(I-\Pi)U_{\mathrm{exit}}\Pi^{\mathrm{stat}}_{\mathrm{dr}}(a_{c})\|\le\eta_{\mathrm{dyn}}.\label{app:stages:exit-dressed-leakage}
\end{equation}

To include the initial matching error, insert $I=\Pi^{\mathrm{stat}}_{\mathrm{dr}}(a_{c})+(I-\Pi^{\mathrm{stat}}_{\mathrm{dr}}(a_{c}))$
between $U_{\mathrm{exit}}$ and $P_{\mathrm{br}}(A_{c})$. Unitarity,
Lemma~\ref{app:stages:exit-matching-lemma}, and Eq.~\eqref{app:stages:exit-dressed-leakage}
give
\begin{align*}
\|(I-\Pi)U_{\mathrm{exit}}P_{\mathrm{br}}(A_{c})\| & \le\|(I-\Pi)U_{\mathrm{exit}}\Pi^{\mathrm{stat}}_{\mathrm{dr}}(a_{c})\|\\
 & \quad+\|(I-\Pi^{\mathrm{stat}}_{\mathrm{dr}}(a_{c}))P_{\mathrm{br}}(A_{c})\|\\
 & \le\eta_{\mathrm{dyn}}+\sqrt{2}\,\eta_{\mathrm{match}}.
\end{align*}
Tensoring with $I_{\mathrm{dark}}$ gives $\Pi\otimes I_{\mathrm{dark}}=\Pi_{\mathrm{out}}$
and proves Eq.~\eqref{app:stages:exit-error}. Equation~\eqref{app:stages:exit-equation}
gives the duration. In this analysis, we only require a fixed system
gap to control the reference construction and propagation bound. 
\end{proof}

%% file: appendix/appendix_E_resources.tex
\section{Global accuracy, physical resources, and an interacting example}

\label{app:resources:section}

The preceding estimates were obtained at an arbitrary integer expansion
order $n$. We now choose that order once for the entire pulse, including
the branch-identification conditions, and complete the proof of Proposition~\ref{app:model:main}.

\subsection{One order for all stages}

\label{app:resources:order}

Keep the local data and the construction constants fixed as in Appendices
\ref{app:model:setting}--\ref{app:stages}. In particular, $\theta$
and the prefactors in $B_{\mathrm{ent}}$ and $B_{e}$ are fixed before
choosing $n$; their stage parameters are given in Eqs. \eqref{app:stages:entrance-pulse}
and \eqref{app:stages:exit-scale}. We now choose $n$ so that all
stage errors and branch-identification conditions hold simultaneously.

The amplitude ratio determining the segment count is 
\[
\frac{A_{\mathrm{ent}}}{A_{c}}=\frac{B_{\mathrm{ent}}B_{e}}{\sqrt{2}\,\theta^{2}}.
\]
Both factors in the numerator are fixed powers of $n$ and $L_{n}$,
so the bound on $M_{\mathrm{tr}}$ in Lemma~\ref{app:stages:transport-duration-lemma}
gives $M_{\mathrm{tr}}\le C(1+\log n)$. During the transport, the
gap bound for every segment satisfies $g_{j}=c_{g}\min\{1,a^{2}_{j}\}\ge c_{g}a^{2}_{c}=c_{g}\theta^{2}/B^{2}_{e}$.
Writing $g_{\min}=\min_{j}g_{j}$, we have $g_{\min}\ge cB^{-2}_{e}$.
Consequently, every inverse gap in a cutoff, coefficient bound, or
inverse-truncation error is bounded by a fixed polynomial in $n$
and $L_{n}$. More explicitly, the entrance terms contain $NB_{\mathrm{ent}}$
or $N$, the transport terms contain $NB_{j}$ and $N^{2}\mathcal{P}(n,g_{j})B_{j}$,
and the exit terms contain at most $N^{3}\mathcal{P}_{e}B^{2}_{e}$.
Each factor other than $N$ is now bounded by a fixed power of $n$
and $L_{n}$. Summing the transport terms adds at most $C(1+\log n)$.
Choose a fixed integer $K_{\mathrm{err}}$ larger than all these powers,
and let $c_{f}>0$ be no larger than the fixed exponential decay constants
in the stage estimates. A common bound for their sum and for the raw
branch-selection errors is 
\begin{equation}
\mathcal{E}(N,n)=CN^{3}n^{K_{\mathrm{err}}}L^{K_{\mathrm{err}}}_{n}(1+\log n)\bigl(\theta^{n-1}+e^{-c_{f}n}\bigr),\label{app:resources:envelope}
\end{equation}
after increasing the fixed prefactor $C$. The factor $N^{3}$ comes
from static exit matching, while the inverse-filter contribution decays
exponentially in $n$. 
\begin{lem}[Simultaneous order choice]
\label{app:resources:orderlemma} There are constants $n_{0}\ge3$
and $c_{\log}>0$, independent of $N$ and $\epsilon$, for which
\begin{equation}
n=\max\left\{ n_{0},\left\lceil c_{\log}\log\!\left(\frac{2N}{\epsilon}\right)\right\rceil \right\} \label{app:resources:n}
\end{equation}
satisfies all branch and expansion conditions and gives entrance leakage
at most $\epsilon/8$, total transport leakage at most $\epsilon/8$,
and combined exit matching and propagation leakage at most $\epsilon/4$. 
\end{lem}

\begin{proof}
Set $\alpha=\min\{-\log\theta,c_{f}\}>0$. Since $\theta$ is fixed,
the last factor in \eqref{app:resources:envelope} is at most $Ce^{-\alpha n}$.
For all sufficiently large fixed $n_{0}$, its polynomial prefactor
obeys 
\[
n^{K_{\mathrm{err}}}L^{K_{\mathrm{err}}}_{n}(1+\log n)\le e^{\alpha n/2}\qquad(n\ge n_{0}).
\]
Moreover $N^{3}/\epsilon\le\left(2N/\epsilon\right)^{4}$. Choose
$c_{\log}$ so that $\alpha c_{\log}/4\ge4$. Then 
\begin{equation}
\frac{\mathcal{E}(N,n)}{\epsilon}\le C\left(2N/\epsilon\right)^{4}e^{-\alpha n/2}\le Ce^{-\alpha n/4}.\label{app:resources:absorb}
\end{equation}
Increasing the fixed lower order makes this smaller than any prescribed
finite set of positive budgets. Choose those budgets to impose the
three claimed leakages, the small static remainder required for entrance
branch selection, and the exit overlap condition used in Theorem~\ref{app:stages:exit}.
Since $\epsilon\le1$, this also enforces the corresponding absolute
smallness conditions.
\end{proof}

\subsection{From subspace leakage to the full trace norm}

\label{app:resources:composition}

Use the composition inequality \eqref{app:stages:composition-bound}
twice, with the projections in Eq.~\eqref{app:model:interfaces}
and the three propagators of Theorems~\ref{app:stages:entrance}--\ref{app:stages:exit}.
With $U_{\mathrm{rot}}=U_{\mathrm{exit}}U_{\mathrm{tr}}U_{\mathrm{ent}}$,
Lemma~\ref{app:resources:orderlemma} yields 
\begin{equation}
\|(I-\Pi_{\mathrm{out}})U_{\mathrm{rot}}\Pi_{\mathrm{in}}\|\le\frac{\epsilon}{2}.\label{app:resources:isometry}
\end{equation}
The final physical bath rotation does not change the reduced system
state. For any system density matrix $\rho_{S}$, let $\sigma_{\mathrm{in}}=\rho_{S}\otimes|0\rangle_{B}\langle0|$.
Then $\Pi_{\mathrm{in}}\sigma_{\mathrm{in}}=\sigma_{\mathrm{in}}$
and $\operatorname{Tr}\sigma_{\mathrm{in}}=1$. Since $\Pi_{\mathrm{out}}=P_{u}\otimes I_{B}$,
Eq.~\eqref{app:resources:isometry} gives
\begin{equation}
\begin{aligned}1-\operatorname{Tr}[P_{u}\rho_{\mathrm{out}}(\rho_S)] & =\operatorname{Tr}\!\left[\sigma_{\mathrm{in}}\Pi_{\mathrm{in}}U^{\dagger}_{\mathrm{rot}}(I-\Pi_{\mathrm{out}})U_{\mathrm{rot}}\Pi_{\mathrm{in}}\right]\\
 & \le\left\Vert (I-\Pi_{\mathrm{out}})U_{\mathrm{rot}}\Pi_{\mathrm{in}}\right\Vert ^{2}\operatorname{Tr}\sigma_{\mathrm{in}}\le\frac{\epsilon^{2}}{4}.
\end{aligned}
\label{app:resources:probability}
\end{equation}
Since $P_{u}$ is rank one, the Fuchs--van de Graaf inequality \cite[Theorem~1, Eq.~(46)]{FuchsGraaf}
and Eq. \eqref{app:resources:probability} give 
\begin{equation}
\|\rho_{\mathrm{out}}(\rho_S)-P_{u}\|_{1}\le2\sqrt{1-\operatorname{Tr}[P_{u}\rho_{\mathrm{out}}(\rho_S)]}\le\epsilon.\label{app:resources:fidelity}
\end{equation}
This proves Eq. \eqref{app:model:accuracy} for arbitrary mixed inputs.

\subsection{Duration and common control waveform}

Combining the entrance duration in Theorem \ref{app:stages:entrance},
the exit duration in Eq. \eqref{app:stages:exit-equation}, and the
transport estimate \eqref{app:stages:geometric-cost} gives

\begin{equation}
\Lambda\tau\le1+\frac{B_{e}}{\theta}+Cn^{7D}L^{10D}_{n}\bigl(1+a^{-(10D+5)}_{c}\bigr).\label{app:resources:duration}
\end{equation}

Substituting $a_{c}=\theta/B_{e}$ into the last term gives $n^{7D}L^{10D}_{n}[1+(B_{e}/\theta)^{10D+5}]$.
Since $B_{e}=C_{e}n^{6D}L^{8D}_{n}$ and $\theta$ is fixed, every
term on the right of Eq.~\eqref{app:resources:duration} is bounded
by a fixed power of $n$ times a fixed power of $L_{n}$. Absorbing
the fixed powers of $L_{n}$ into an additional power of $n$ gives
$\Lambda\tau\le Cn^{\kappa}$ for a finite exponent $\kappa$ depending
only on $D$ and the fixed local data. Since \eqref{app:resources:n}
implies $n\le C\log(2N/\epsilon)$, this proves \eqref{app:model:time}.

The amplitude increases to $A_{\mathrm{ent}}$ during entrance and
then decreases. The phase rate decreases from $r$ to zero during
entrance and stays zero afterward. Hence 
\begin{equation}
\max_{t}A(t)=\frac{r}{\sqrt{2}},\qquad\max_{t}|\dot{\phi}(t)|=r=\frac{\Lambda C_{\mathrm{ent}}(n+2)^{2}}{\theta},\label{app:resources:peaks}
\end{equation}
which proves \eqref{app:model:peak}. 

The switch profile, entrance scale, amplitude list, segment durations,
and release cutoff depend only on $N,\epsilon$ and the promised local,
energy, and gap bounds. Integrating the specified $\omega(t)$ fixes
the single phase $\phi(t)$. The interacting ground projections and
the local dressing operators enter the proof but never enter the prescribed
waveform. Thus the controller does not need the individual entries
of $h$ or the coefficients of $V_{X}$. The native Hamiltonian remains
on, the bath remains the same $2N$ modes throughout, and there is
one final discard. This completes Proposition~\ref{app:model:main}.